\documentclass[]{article}
\usepackage[english]{babel}
\usepackage{tikz}
\usepackage{amssymb}
\usepackage[mathscr]{eucal}
\usepackage[T1]{fontenc}
\usepackage[utf8]{inputenc}
\usepackage{calrsfs}
\usepackage{float}
\usepackage{amsmath}
\usepackage{amsthm}
\usepackage{enumerate}
\usepackage{caption}
\usepackage{multirow}
\usepackage[titletoc,toc]{appendix}
\usepackage{authblk}
\usetikzlibrary[shapes.misc]
\usetikzlibrary[shapes.geometric]
\newtheorem{de}{Definition}[section]
\newtheorem{theo}{Theorem}
\newtheorem{prop}[de]{Proposition}
\newtheorem{cor}[de]{Corollary}
\newtheorem{lem}[de]{Lemma}

\newtheorem{obs}[de]{Observation}
\newtheorem{con}{Conjecture}

\title{On the family of $r$-regular graphs with Grundy number $r+1$}

\author[1,2]{Nicolas Gastineau\thanks{Author partially supported by the Burgundy Council}}
\author[2]{Hamamache Kheddouci}
\author[1]{Olivier Togni}
\affil[1]{\textit{LE2I, UMR CNRS 6303} \\ \textit{Universit\'e de Bourgogne, 21078 Dijon cedex, France} }
\affil[2]{\textit{Université de Lyon, CNRS, Université Lyon 1, LIRIS, UMR5205, F-69622, France} }

\begin{document}
\maketitle

\begin{abstract}
The Grundy number of a graph $G$, denoted by $\Gamma(G)$, is the largest $k$ such that there exists a partition of
$V(G)$, into $k$ independent sets $V_1,\ldots, V_k$ and every vertex of $V_i$ is adjacent to at least one vertex in
$V_j$, for every $j<i$. The objects which are studied in this article are families of $r$-regular graphs such that $\Gamma(G) = r + 1$. Using the notion of independent module, a characterization of this family is given for $r=3$. 
Moreover, we determine classes of graphs in this family, in particular the class of $r$-regular graphs without induced $C_4$, for $r \le 4$. Furthermore, our propositions imply results on partial Grundy number.
\end{abstract}
\section{Introduction}
We consider only undirected connected graphs in this paper.
Given a graph $G=(V,E)$, a \textit{proper $k$-coloring} of $G$ is a surjective mapping $c:V \rightarrow\{1,\ldots,k\}$ such that $c(u)\neq c(v)$ for any $uv\in E$; the \textit{color class} $V_i$ is the set $\{u\in V| c(u)=i\}$ and a vertex $v$ has color $i$ if $v\in V_i$.
A vertex $v$ of color $i$ is a \textit{Grundy vertex} if $v$ is adjacent to at least one vertex colored $j$, for every $j<i$.
A \textit{Grundy $k$-coloring} is a proper $k$-coloring such that every vertex is a Grundy vertex. 
A \textit{partial Grundy $k$-coloring} is a proper $k$-coloring such that every color class contains a Grundy vertex.
The \textit{Grundy number} (\textit{partial Grundy number}, respectively) of $G$ denoted by $\Gamma(G)$ ($\partial\Gamma(G)$, respectively) is the largest $k$ such that $G$ admits a Grundy $k$-coloring (partial Grundy $k$-coloring, respectively).

Let $N(v)=\{u\in V(G)|uv\in E(G)\}$ be the neighborhood  of $v$.
A set $X$ of vertices is an \textit{independent module} if $X$ is an independent set and all vertices in $X$ have the same neighborhood. The vertices in an independent module of size 2 are called \textit{false twins}.
Let $P_n$, $C_n$, $K_n$ and $I_n$ be respectively, the path, cycle complete and empty graph of order $n$.
The concepts of Grundy $k$-coloring and domination are connected. In a Grundy coloring, $V_{1}$ is a dominating set. 
Given a graph $G$ and an ordering $\phi$ on $V(G)$ with $\phi=v_1,\ldots,v_n$, the greedy algorithm assigns to $v_i$ the minimum color that was not assigned in the set $\{v_1,\ldots,v_{i-1}\} \cap N(v_i)$.
Let $\Gamma_{\phi}(G)$ be the number of colors used by the greedy algorithm with the ordering $\phi$ on $G$. We obtain the following result \cite{ER2003}: $\Gamma(G)=\max\limits_{\phi\in S_n}(\Gamma_{\phi}(G))$.

The Grundy coloring is a well studied problem. Zaker \cite{ZA2006} proved that determining the Grundy number of a given graph, even for complements of bipartite graphs, is an NP-complete problem. However, for a fixed $t$, determining if a given graph has Grundy number at least $t$ is decidable in polynomial time. This result follows from the existence of a finite list of graphs, called $t$-atoms, such that any graph with Grundy number at least $t$ contains a $t$-atom as an induced subgraph.
It has been proven that there exists a Nordhaus-Gaddum type inequality for the Grundy number \cite{FU2008,ZA2006}, that
there exist upper bounds for $d$-degenerate, planar and outerplanar graphs \cite{BA2008,CH2012}, and that there exist
connections between the products of graphs and the Grundy number \cite{EF2007,AS2010,CA2012}.
Recently, Havet and Sampaio \cite{HA2013} have proven that the problem of deciding if for a given graph $G$ we have $\Gamma(G)=\Delta(G)+1$, even if $G$ is bipartite, is NP-complete. Moreover, they have proven that the dual of Grundy $k$-coloring problem is in FPT by finding an algorithm in $O (2k^{2k}.|E|+2^{2k} k^{3k+5/2})$ time.\newline\
Note that a Grundy $k$-coloring is a partial Grundy $k$-coloring, hence $\Gamma(G)\le\partial\Gamma(G)$.
Given a graph $G$ and a positive integer $k$, the problem of determining if a partial Grundy $k$-coloring exists, even for chordal graphs, is NP-complete but there exists a polynomial algorithm for trees \cite{SH2005}.

Another coloring parameter with domination constraints on the colors is the \textit{$b$-chromatic number}, denoted by $\varphi(G)$, which is the largest $k$ such that there exists a proper $k$-coloring and for every color class $V_i$, there exists a vertex adjacent to at least one vertex colored $j$, for every $j$, with $j\neq i$.
Note that a $b$-coloring is a partial Grundy $k$-coloring, hence $\varphi(G)\le\partial\Gamma(G)$.
The $b$-chromatic number of regular graphs has been investigated in a series of papers (\cite{EL2009,KL2010,CA2011,SH2012}). Our aim is to establish similar results for the Grundy coloring.
We present two main results: A characterization of the Grundy number of every cubic graph and the following theorem: For $r\le 4$, every $r$-regular graphs without induced $C_4$ has Grundy number $r+1$. We conjecture that this assertion is also true for $r>4$.

\begin{con}
For any integer $r\ge1$, every $r$-regular graph without induced $C_4$ has Grundy number $r+1$.
\end{con}

Section 2 gives characterizations of some classes of graphs with Grundy number at most $k$, $2\le k\le \Delta(G)$, using the notion of independent module. 
Section 3 contains the first main theorem: A description of the cubic graphs with Grundy number at most 3 that also allows us to prove that every cubic graph except $K_{3,3}$ has partial Grundy number $4$. 
This theorem implies the existence of a linear algorithm to determine the Grundy number of cubic graphs. In Section 4, we present examples of infinite families of regular graphs with Grundy number exactly or at most $k$, $3\le k\le r$. 
To determine these families we use recursive definitions. The last section contains the second main theorem of this article: 4-regular graphs without induced $C_4$ have Grundy number 5.

\section{General results}
The reader has to be aware of the resemblance of name between the following notion and that of partial Grundy $k$-coloring.
\begin{de}
Let $G$ be a graph. A Grundy partial $k$-coloring is a Grundy $k$-coloring of a subset $S$ of $V(G)$.
\end{de}
\begin{obs}[\cite{AS2010},\cite{EF2007}]
If $G$ admits a Grundy partial $k$-coloring, then $\Gamma(G)\ge k$.
\end{obs}
This property has an important consequence: For a graph $G$, with $\Gamma(G)\ge t$ and any Grundy partial $t$-coloring, there exist smallest subgraphs $H$ of $G$ such that $\Gamma(H)=t$.
The family of $t$-atoms corresponds to these subgraphs. This concept was introduced by Zaker \cite{ZA2006}. The family of $t$-atoms is finite and the presence of a $t$-atom can be determined in polynomial time for a fixed $t$. The following definition is slightly different from Zaker's one, insisting more on the construction of every $t$-atom.
\begin{de}[\cite{ZA2006}]
For any integer $t$, we define the family of $t$-atoms, denoted by $\mathcal{A}_t$, $t=1,\ldots$ by induction. Let the family $\mathcal{A}_1$ contain only $K_1$.
A graph $G$ is in $\mathcal{A}_{t+1}$ if there exists a graph $G'$ in $\mathcal{A}_t$ and an integer $m$, $m\le|V(G')|$, such that $G$ is composed of $G'$ and an independent set $I_m$ of order $m$, adding edges between $G'$ and $I_m$ such that every vertex in $G'$ is connected to at least one vertex in $I_m$.
Moreover a $t$-atom $A$ is minimal, if there is no $t$-atom included in $A$ other than itself.
\end{de}

\begin{theo}[\cite{ZA2006}] 
For a given graph $G$, $\Gamma(G)\ge t$ if and only if $G$ contains an induced minimal $t$-atom.
\end{theo}
We now present conditions related to the presence of modules that allows us to upper-bound the Grundy number.
\begin{prop}[\cite{AS2010}]
Let $G$ be a graph and $X$ be an independent module. In every Grundy coloring of $G$, the vertices in $X$ must have the same color.
\label{ff}
\end{prop}
\begin{de}
Let $G$ be an $r$-regular graph. A vertex $v$ is a $(0,\ell)$-twin-vertex if there exists an independent module of cardinality $r+2-\ell$ that contains $v$.
\end{de}
\begin{prop}
Let $G$ be an $r$-regular graph. The color of an $(0,\ell)$-twin-vertex is at most $\ell$ in every Grundy coloring of $G$.
\label{0twin}
\end{prop}
\begin{proof}
Let $v$ be a $(0,\ell)$-twin-vertex colored $\ell+1$ in $G$. By Definition, $v$ is in an independent module $X$ of cardinality $r+2-\ell$ and by Proposition \ref{ff}, every other vertex of $X$ should be colored $\ell+1$. Let $u$ be a neighbor of $v$. There are at most $\ell-2$ neighbors of $u$ in $V(G-X)$.
Therefore, $u$ cannot be colored $\ell$.
\end{proof}
\begin{de}
A vertex $v$ of a graph $G$ is a $(1,\ell)$-twin-vertex if $N(v)$ can be partitioned into at least $\ell-1$ independent modules.
\end{de}
\begin{prop}
Let $G$ be a graph. The color of an $(1,\ell)$-twin-vertex is at most $\ell$ in every Grundy coloring of $G$.
\end{prop}
\begin{proof}
By Proposition \ref{ff}, vertices of the neighborhood of $v$ can only have $\ell-1$ different colors. Therefore, the color of $v$ is at most $\ell$.
\end{proof}
\begin{de}
A vertex $v$ of a graph $G$ is a $(2,\ell)$-twin-vertex if $N(v)$ is independent and every vertex in $N(v)$ is a $(1,\ell)$-twin-vertex.
\end{de}
\begin{prop}
Let $G$ be a graph. The color of an $(2,\ell)$-twin-vertex is at most $\ell$ in every Grundy coloring of $G$.
\end{prop}
\begin{proof}
Let $v$ be a $(2,\ell)$-twin-vertex in $G$. Every vertex in $N(v)$ is a $(1,\ell)$-twin-vertex. If a vertex in $N(v)$ is colored $\ell$, then $v$ could only have a color at most $\ell-1$. 
If the vertices in the neighborhood of $v$ have colors at most $\ell-1$, then in every Grundy coloring of $G$, $v$ has a color at most $\ell$.
\end{proof}
\begin{cor}
Let $G$ be a graph. If every vertex is a $(1,\ell)$-twin-vertex or a $(2,\ell)$-twin-vertex, then $\Gamma(G)\le\ell$.
\label{2twin}
\end{cor}
\begin{cor}
Let $G$ be a regular graph. If every vertex is an $(i,\ell)$-twin-vertex, for some $i$, $0 \le i \le 2$, then $\Gamma(G)\le\ell$.
\label{itwin}
\end{cor}
\begin{prop}[\cite{AS2010},\cite{ZA2006}]
Let $G$ be a graph. We have $\Gamma(G)\le 2$ if and only if $G=K_{n,m}$ for some integers $n>0$ and $m>0$.
\label{g2}
\end{prop}
\section{Grundy numbers of cubic graphs}
In the following sections, the figures describe Grundy partial $k$-colorings. By a dashed edge we denote a possible edge. The vertices not connected by edges in the figures cannot be adjacent as it would contradict the hypothesis.
\begin{prop}[\cite{EF2007}]
Let $G$ be a connected 2-regular graph. $\partial\Gamma(G)=\Gamma(G)=2$ if and only if $G=C_4$.
\label{indc2}
\end{prop}

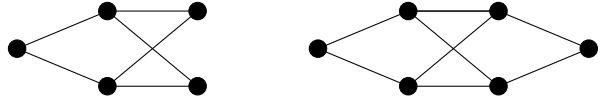
\begin{figure}[t]
\begin{center}
\begin{tikzpicture}
\draw (0,0.5) -- (1.2,1);
\draw (0,0.5) -- (1.2,0);
\draw (1.2,0) -- (2.4,1);
\draw (1.2,0) -- (2.4,0);
\draw (1.2,1) -- (2.4,1);
\draw (1.2,1) -- (2.4,0);

\draw (0+4,0.5) -- (1.2+4,1);
\draw (0+4,0.5) -- (1.2+4,0);
\draw (1.2+4,0) -- (2.4+4,1);
\draw (1.2+4,0) -- (2.4+4,0);
\draw (1.2+4,1) -- (2.4+4,1);
\draw (1.2+4,1) -- (2.4+4,0);
\draw (1.2+4,1) -- (2.4+4,1);
\draw (2.4+4,1) -- (3.6+4,0.5);
\draw (2.4+4,0) -- (3.6+4,0.5);
\node at (0,0.5) [circle,draw=black,fill=black, scale=0.7] {};
\node at (1.2,0)  [circle,draw=black,fill=black, scale=0.7] {};
\node at (2.4,0) [circle,draw=black,fill=black, scale=0.7] {};
\node at (1.2,1) [circle,draw=black,fill=black, scale=0.7] {};
\node at (2.4,1)  [circle,draw=black,fill=black, scale=0.7] {};
\node at (0+4,0.5)  [circle,draw=black,fill=black, scale=0.7] {};
\node at (1.2+4,0) [circle,draw=black,fill=black, scale=0.7] {};
\node at (2.4+4,0) [circle,draw=black,fill=black, scale=0.7] {};
\node at (1.2+4,1) [circle,draw=black,fill=black, scale=0.7] {};
\node at (2.4+4,1) [circle,draw=black,fill=black, scale=0.7] {};
\node at (3.6+4,0.5) [circle,draw=black,fill=black, scale=0.7] {};
\end{tikzpicture}
\end{center}
\caption{The graphs $K_{2,3}$ (on the left) and $K^{*}_{3,3}$ (on the right).}
\label{figk}
\end{figure}

The following definition gives a construction of the cubic graphs in which every vertex is an $(i,3)$-twin-vertex, for some $i$, $0 \le i \le 2$. Figure \ref{figcuu} gives the list of every graph of order at most 16 in this family.
\begin{de}
Let $K_{2,3}$ and $K^{*}_{3,3}$ be the graphs from Figure~\ref{figk}.
We define recursively the family of graphs $\mathcal{F}^{*}_3$ as follows:
\begin{enumerate}
\item $K_{2,3}  \in \mathcal{F}^{*}_3$ and $K^{*}_{3,3}\in \mathcal{F}^{*}_3$;
\item the disjoint union of two elements of $\mathcal{F}^{*}_3$ is in $\mathcal{F}^{*}_3$;
\item if $G$ is a graph in $\mathcal{F}^{*}_3$, then the graph $H$ obtained from $G$ by adding an edge between two vertices of degree at most $2$ is also in $\mathcal{F}^{*}_3$;
\item if $G$ is a graph in $\mathcal{F}^{*}_3$, then the graph $H$ obtained from $G$ by adding a new vertex adjacent to three vertices of degree at most 2 is in $\mathcal{F}^{*}_3$.
\end{enumerate}
The family $\mathcal{F}_3$ is the subfamily of cubic graphs in $\mathcal{F}^{*}_3$.
\end{de}

\begin{prop}\label{ggg2}
Let $G$ be a cubic graph. Every vertex of $V(G)$ is an $(i,3)$-twin vertex, for some $i$, $0\leq i\leq 2$, if and only if $G \in\mathcal{F}_3$.
\end{prop}
\begin{figure}[!]
\begin{center}
\begin{tikzpicture}
\draw (0,0) -- (0.8,0);
\draw (0,0.8) -- (0.8,0.8);
\draw (0,0) -- (-0.7,0.4);
\draw (0,0.8) -- (-0.7,0.4);
\draw (0.8,0) -- (1.5,0.4);
\draw (0.8,0.8) -- (1.5,0.4);
\draw (0.8,0) -- (0,0.8);
\draw (0.8,0.8) -- (0,0);
\draw (1.5,0.4) -- (-0.7,0.4);

\node at (0,0)  [circle,draw=black,fill=black, scale=0.7] {};
\node at (0.8,0)  [circle,draw=black,fill=black, scale=0.7] {};
\node at (0,0.8)  [circle,draw=black,fill=black, scale=0.7] {};
\node at (0.8,0.8)  [circle,draw=black,fill=black, scale=0.7] {};
\node at (1.5,0.4)  [circle,draw=black,fill=black, scale=0.7] {};
\node at (-0.7,0.4)  [circle,draw=black,fill=black, scale=0.7] {};
\node at (0.4,-0.6){$K_{3,3}$};

\draw (2.5,0.7) -- (3,0);
\draw (3.5,0.7) -- (3,0);
\draw (2.5,0.7) -- (2.5,1.5);
\draw (3.5,0.7) -- (3.5,1.5);
\draw (2.5,0.7) -- (3.5,1.5);
\draw (3.5,0.7) -- (2.5,1.5);
\draw (3.5,1.5) -- (2.5,1.5);
\draw (4,0.7) -- (4.5,0);
\draw (5,0.7) -- (4.5,0);
\draw (4,0.7) -- (4,1.5);
\draw (5,0.7) -- (5,1.5);
\draw (4,0.7) -- (5,1.5);
\draw (5,0.7) -- (4,1.5);
\draw (5,1.5) -- (4,1.5);
\draw (5,1.5) -- (4,1.5);
\draw (3,0) -- (4.5,0);
\node at (3,0)  [circle,draw=black,fill=black, scale=0.7] {};
\node at (2.5,0.7)  [circle,draw=black,fill=black, scale=0.7] {};
\node at (3.5,0.7)  [circle,draw=black,fill=black, scale=0.7] {};
\node at (2.5,1.5)  [circle,draw=black,fill=black, scale=0.7] {};
\node at (3.5,1.5) [circle,draw=black,fill=black, scale=0.7] {};
\node at (4.5,0)  [circle,draw=black,fill=black, scale=0.7] {};
\node at (4,0.7)  [circle,draw=black,fill=black, scale=0.7] {};
\node at (5,0.7)  [circle,draw=black,fill=black, scale=0.7] {};
\node at (4,1.5)  [circle,draw=black,fill=black, scale=0.7] {};
\node at (5,1.5) [circle,draw=black,fill=black, scale=0.7] {};

\draw (6.7,0) -- (9.1,0);
\draw (6.7,0.8) -- (9.1,0.8);
\draw (6.7,0) -- (6,0.4);
\draw (6.7,0.8) -- (6,0.4);
\draw (9.1,0) -- (8.3,0.8);
\draw (9.1,0.8) -- (8.3,0);
\draw (7.5,0) -- (6.7,0.8);
\draw (7.5,0.8) -- (6.7,0);
\draw (9.8,0.4) -- (9.1,0);
\draw (9.8,0.4) -- (9.1,0.8);
\draw (9.8,0.4) -- (6,0.4);
\node at (6,0.4)  [circle,draw=black,fill=black, scale=0.7] {};
\node at (9.8,0.4)  [circle,draw=black,fill=black, scale=0.7] {};
\node at (6.7,0)  [circle,draw=black,fill=black, scale=0.7] {};
\node at (7.5,0)  [circle,draw=black,fill=black, scale=0.7] {};
\node at (8.3,0) [circle,draw=black,fill=black, scale=0.7] {};
\node at (9.1,0)  [circle,draw=black,fill=black, scale=0.7] {};
\node at (6.7,0.8)  [circle,draw=black,fill=black, scale=0.7] {};
\node at (7.5,0.8)  [circle,draw=black,fill=black, scale=0.7] {};
\node at (8.3,0.8) [circle,draw=black,fill=black, scale=0.7] {};
\node at (9.1,0.8)  [circle,draw=black,fill=black, scale=0.7] {};

\draw (-0.3,0-3.3) -- (-0.7,0.7-3.3);
\draw (-0.3,0-3.3) -- (0.1,0.7-3.3);
\draw (-0.7,1.5-3.3) -- (-0.7,0.7-3.3);
\draw (0.1,1.5-3.3) -- (0.1,0.7-3.3);
\draw (-0.7,1.5-3.3) -- (0.1,0.7-3.3);
\draw (0.1,1.5-3.3) -- (-0.7,0.7-3.3);
\draw (-0.7,1.5-3.3) -- (-0.3,2.2-3.3);
\draw (0.1,1.5-3.3) -- (-0.3,2.2-3.3);
\draw (1,0-3.3) -- (0.6,0.7-3.3);
\draw (1,0-3.3) -- (1.4,0.7-3.3);
\draw (0.6,1.5-3.3) -- (0.6,0.7-3.3);
\draw (1.4,1.5-3.3) -- (1.4,0.7-3.3);
\draw (0.6,1.5-3.3) -- (1.4,0.7-3.3);
\draw (1.4,1.5-3.3) -- (0.6,0.7-3.3);
\draw (0.6,1.5-3.3) -- (1,2.2-3.3);
\draw (1.4,1.5-3.3) -- (1,2.2-3.3);
\draw (-0.3,0-3.3) -- (1,0-3.3);
\draw (-0.3,2.2-3.3) -- (1,2.2-3.3);
\node at (-0.3,0-3.3)  [circle,draw=black,fill=black, scale=0.7] {};
\node at (-0.3,2.2-3.3)  [circle,draw=black,fill=black, scale=0.7] {};
\node at (1,0-3.3) [circle,draw=black,fill=black, scale=0.7] {};
\node at (1,2.2-3.3)  [circle,draw=black,fill=black, scale=0.7] {};
\node at (-0.7,1.5-3.3)   [circle,draw=black,fill=black, scale=0.7] {};
\node at (-0.7,0.7-3.3)   [circle,draw=black,fill=black, scale=0.7] {};
\node at (0.1,1.5-3.3) [circle,draw=black,fill=black, scale=0.7] {};
\node at (0.1,0.7-3.3) [circle,draw=black,fill=black, scale=0.7] {};
\node at (0.6,1.5-3.3)   [circle,draw=black,fill=black, scale=0.7] {};
\node at (0.6,0.7-3.3)   [circle,draw=black,fill=black, scale=0.7] {};
\node at (1.4,1.5-3.3) [circle,draw=black,fill=black, scale=0.7] {};
\node at (1.4,0.7-3.3) [circle,draw=black,fill=black, scale=0.7] {};

\draw (-0.3+2.9,0-3.3) -- (-0.7+2.9,0.7-3.3);
\draw (-0.3+2.9,0-3.3) -- (0.1+2.9,0.7-3.3);
\draw (-0.7+2.9,1.5-3.3) -- (-0.7+2.9,0.7-3.3);
\draw (0.1+2.9,1.5-3.3) -- (0.1+2.9,0.7-3.3);
\draw (-0.7+2.9,1.5-3.3) -- (0.1+2.9,0.7-3.3);
\draw (0.1+2.9,1.5-3.3) -- (-0.7+2.9,0.7-3.3);
\draw (-0.7+2.9,1.5-3.3) -- (-0.3+2.9,2.2-3.3);
\draw (0.1+2.9,1.5-3.3) -- (-0.3+2.9,2.2-3.3);
\draw (-0.3+2.9,0-3.3) -- (0.8+2.9,1.1-3.3);
\draw (-0.3+2.9,2.2-3.3) -- (0.8+2.9,1.1-3.3);
\draw (0.8+2.9,1.1-3.3) -- (0.8+0.5+2.9,1.1-3.3);
\draw (1.5+0.5+2.9,0.7-3.3) -- (0.8+0.5+2.9,1.1-3.3);
\draw (1.5+0.5+2.9,1.5-3.3) -- (0.8+0.5+2.9,1.1-3.3);
\draw (1.5+0.5+2.9,0.7-3.3) -- (2.3+0.5+2.9,0.7-3.3);
\draw (1.5+0.5+2.9,1.5-3.3) -- (2.3+0.5+2.9,1.5-3.3);
\draw (1.5+0.5+2.9,0.7-3.3) -- (2.3+0.5+2.9,1.5-3.3);
\draw (1.5+0.5+2.9,1.5-3.3) -- (2.3+0.5+2.9,0.7-3.3);
\draw (2.3+0.5+2.9,0.7-3.3) -- (2.3+0.5+2.9,1.5-3.3);

\node at (-0.3+2.9,0-3.3)  [circle,draw=black,fill=black, scale=0.7] {};
\node at (-0.7+2.9,0.7-3.3)  [circle,draw=black,fill=black, scale=0.7] {};
\node at (0.1+2.9,0.7-3.3) [circle,draw=black,fill=black, scale=0.7] {};
\node at (0.1+2.9,1.5-3.3)  [circle,draw=black,fill=black, scale=0.7] {};
\node at (-0.7+2.9,1.5-3.3)   [circle,draw=black,fill=black, scale=0.7] {};
\node at (-0.3+2.9,2.2-3.3)   [circle,draw=black,fill=black, scale=0.7] {};
\node at (0.8+0.5+2.9,1.1-3.3) [circle,draw=black,fill=black, scale=0.7] {};
\node at (0.8+2.9,1.1-3.3) [circle,draw=black,fill=black, scale=0.7] {};
\node at (1.5+0.5+2.9,0.7-3.3) [circle,draw=black,fill=black, scale=0.7] {};
\node at (2.3+0.5+2.9,0.7-3.3)   [circle,draw=black,fill=black, scale=0.7] {};
\node at (1.5+0.5+2.9,1.5-3.3) [circle,draw=black,fill=black, scale=0.7] {};
\node at (2.3+0.5+2.9,1.5-3.3)   [circle,draw=black,fill=black, scale=0.7] {};

\draw (3.3+2.9,0-3.3) -- (3.3+2.9,0.8-3.3);
\draw (4.1+2.9,0-3.3) -- (4.1+2.9,0.8-3.3);
\draw (3.3+2.9,0-3.3) -- (4.1+2.9,0-3.3);
\draw (3.3+2.9,0-3.3) -- (4.1+2.9,0.8-3.3);
\draw (3.3+2.9,0.8-3.3) -- (4.1+2.9,0-3.3);
\draw (3.3+2.9,0.8-3.3) -- (3.7+2.9,1.5-3.3);
\draw (4.1+2.9,0.8-3.3) -- (3.7+2.9,1.5-3.3);

\draw (4.6+2.9,0-3.3) -- (4.6+2.9,0.8-3.3);
\draw (5.4+2.9,0-3.3) -- (5.4+2.9,0.8-3.3);
\draw (4.6+2.9,0-3.3) -- (5.4+2.9,0-3.3);
\draw (4.6+2.9,0-3.3) -- (5.4+2.9,0.8-3.3);
\draw (4.6+2.9,0.8-3.3) -- (5.4+2.9,0-3.3);
\draw (4.6+2.9,0.8-3.3) -- (5+2.9,1.5-3.3);
\draw (5.4+2.9,0.8-3.3) -- (5+2.9,1.5-3.3);

\draw (5.9+2.9,0-3.3) -- (5.9+2.9,0.8-3.3);
\draw (6.7+2.9,0-3.3) -- (6.7+2.9,0.8-3.3);
\draw (5.9+2.9,0-3.3) -- (6.7+2.9,0-3.3);
\draw (5.9+2.9,0-3.3) -- (6.7+2.9,0.8-3.3);
\draw (5.9+2.9,0.8-3.3) -- (6.7+2.9,0-3.3);
\draw (5.9+2.9,0.8-3.3) -- (6.3+2.9,1.5-3.3);
\draw (6.7+2.9,0.8-3.3) -- (6.3+2.9,1.5-3.3);
\draw (5+2.9,1.5-3.3) -- (5+2.9,2-3.3);
\draw (6.3+2.9,1.5-3.3) -- (5+2.9,2-3.3);
\draw (3.7+2.9,1.5-3.3) -- (5+2.9,2-3.3);

\node at (3.3+2.9,0-3.3)  [circle,draw=black,fill=black, scale=0.7] {};
\node at (3.3+2.9,0.8-3.3)  [circle,draw=black,fill=black, scale=0.7] {};
\node at (4.1+2.9,0-3.3) [circle,draw=black,fill=black, scale=0.7] {};
\node at (4.1+2.9,0.8-3.3)  [circle,draw=black,fill=black, scale=0.7] {};
\node at (4.6+2.9,0-3.3)   [circle,draw=black,fill=black, scale=0.7] {};
\node at (4.6+2.9,0.8-3.3)   [circle,draw=black,fill=black, scale=0.7] {};
\node at (5.4+2.9,0-3.3) [circle,draw=black,fill=black, scale=0.7] {};
\node at (5.4+2.9,0.8-3.3) [circle,draw=black,fill=black, scale=0.7] {};
\node at (5.9+2.9,0-3.3)   [circle,draw=black,fill=black, scale=0.7] {};
\node at (5.9+2.9,0.8-3.3) [circle,draw=black,fill=black, scale=0.7] {};
\node at (6.7+2.9,0-3.3)   [circle,draw=black,fill=black, scale=0.7] {};
\node at (6.7+2.9,0.8-3.3)   [circle,draw=black,fill=black, scale=0.7] {};
\node at (5+2.9,2-3.3)   [circle,draw=black,fill=black, scale=0.7] {};
\node at (5+2.9,1.5-3.3) [circle,draw=black,fill=black, scale=0.7] {};
\node at (6.3+2.9,1.5-3.3)   [circle,draw=black,fill=black, scale=0.7] {};
\node at (3.7+2.9,1.5-3.3)  [circle,draw=black,fill=black, scale=0.7] {};

\draw (-0.7,0.7-5.5) -- (-0.2,0-5.5);
\draw (0.3,0.7-5.5) -- (-0.2,0-5.5);
\draw (-0.7,0.7-5.5) -- (-0.7,1.5-5.5);
\draw (0.3,0.7-5.5) -- (0.3,1.5-5.5);
\draw (-0.7,0.7-5.5) -- (0.3,1.5-5.5);
\draw (0.3,0.7-5.5) -- (-0.7,1.5-5.5);
\draw (0.3,1.5-5.5) -- (-0.7,1.5-5.5);
\draw (-0.2,0-5.5) -- (0.3,0-5.5);
\draw (0.3,0-5.5) -- (1,0.4-5.5);
\draw (0.3,0-5.5) -- (1,-0.4-5.5);
\draw (1,0.4-5.5) -- (1.8,0.4-5.5);
\draw (1,-0.4-5.5) -- (1.8,-0.4-5.5);
\draw (1,0.4-5.5) -- (1.8,-0.4-5.5);
\draw (1,-0.4-5.5) -- (1.8,0.4-5.5);
\draw (1.8,0.4-5.5) -- (2.5,0-5.5);
\draw (1.8,-0.4-5.5) -- (2.5,0-5.5);
\draw (3,0-5.5) -- (2.5,0-5.5);
\draw (2.5,0.7-5.5) -- (3,0-5.5);
\draw (3.5,0.7-5.5) -- (3,0-5.5);
\draw (2.5,0.7-5.5) -- (2.5,1.5-5.5);
\draw (3.5,0.7-5.5) -- (3.5,1.5-5.5);
\draw (2.5,0.7-5.5) -- (3.5,1.5-5.5);
\draw (3.5,0.7-5.5) -- (2.5,1.5-5.5);
\draw (2.5,1.5-5.5) -- (3.5,1.5-5.5);
\node at (-0.2,0-5.5) [circle,draw=black,fill=black, scale=0.7] {};
\node at (-0.7,0.7-5.5)  [circle,draw=black,fill=black, scale=0.7] {};
\node at (0.3,0.7-5.5)  [circle,draw=black,fill=black, scale=0.7] {};
\node at (-0.7,1.5-5.5)  [circle,draw=black,fill=black, scale=0.7] {};
\node at (0.3,1.5-5.5)  [circle,draw=black,fill=black, scale=0.7] {};
\node at (4.6+2.9,0-3.3)   [circle,draw=black,fill=black, scale=0.7] {};
\node at (4.6+2.9,0.8-3.3)   [circle,draw=black,fill=black, scale=0.7] {};
\node at (0.3,0-5.5)   [circle,draw=black,fill=black, scale=0.7] {};
\node at (1,0.4-5.5)  [circle,draw=black,fill=black, scale=0.7] {};
\node at (1,-0.4-5.5)  [circle,draw=black,fill=black, scale=0.7] {};
\node at (1.8,0.4-5.5)  [circle,draw=black,fill=black, scale=0.7] {};
\node at (1.8,-0.4-5.5)  [circle,draw=black,fill=black, scale=0.7] {};
\node at (2.5,0-5.5)  [circle,draw=black,fill=black, scale=0.7] {};
\node at (3,0-5.5)  [circle,draw=black,fill=black, scale=0.7] {};
\node at (2.5,0.7-5.5)  [circle,draw=black,fill=black, scale=0.7] {};
\node at (2.5,1.5-5.5)  [circle,draw=black,fill=black, scale=0.7] {};
\node at (3.5,0.7-5.5)  [circle,draw=black,fill=black, scale=0.7] {};
\node at (3.5,1.5-5.5)  [circle,draw=black,fill=black, scale=0.7] {};

\draw (6.7,0-6) -- (9.1,0-6);
\draw (6.7,0.8-6) -- (9.1,0.8-6);
\draw (6.7,0-6) -- (6,0.4-6);
\draw (6.7,0.8-6) -- (6,0.4-6);
\draw (9.1,0-6) -- (8.3,0.8-6);
\draw (9.1,0.8-6) -- (8.3,0-6);
\draw (7.5,0-6) -- (6.7,0.8-6);
\draw (7.5,0.8-6) -- (6.7,0-6);
\draw (9.8,0.4-6) -- (9.1,0-6);
\draw (9.8,0.4-6) -- (9.1,0.8-6);
\draw (9.8,0.4-6) .. controls (9.8,1.5-6) and (6,1.5-6) .. (5,0.4-6);
\draw (5,0.4-6) -- (6,0.4-6);
\draw (5,0.4-6) -- (5,1-6);
\draw (5,1-6) -- (4.5,1.7-6);
\draw (5,1-6) -- (5.5,1.7-6);
\draw (4.5,2.5-6) -- (4.5,1.7-6);
\draw (5.5,2.5-6) -- (5.5,1.7-6);
\draw (4.5,2.5-6) -- (5.5,1.7-6);
\draw (5.5,2.5-6) -- (4.5,1.7-6);
\draw (5.5,2.5-6) -- (4.5,2.5-6);
\node at (6,0.4-6) [circle,draw=black,fill=black, scale=0.7] {};
\node at (9.8,0.4-6)  [circle,draw=black,fill=black, scale=0.7] {};
\node at (6.7,0-6) [circle,draw=black,fill=black, scale=0.7] {};
\node at (7.5,0-6)  [circle,draw=black,fill=black, scale=0.7] {};
\node at (8.3,0-6) [circle,draw=black,fill=black, scale=0.7] {};
\node at (9,0-6)  [circle,draw=black,fill=black, scale=0.7] {};
\node at (6.7,0.8-6) [circle,draw=black,fill=black, scale=0.7] {};
\node at (7.5,0.8-6)  [circle,draw=black,fill=black, scale=0.7] {};
\node at (8.3,0.8-6) [circle,draw=black,fill=black, scale=0.7] {};
\node at (9,0.8-6)  [circle,draw=black,fill=black, scale=0.7] {};
\node at (5,0.4-6) [circle,draw=black,fill=black, scale=0.7] {};
\node at (5,1-6) [circle,draw=black,fill=black, scale=0.7] {};
\node at (4.5,1.7-6) [circle,draw=black,fill=black, scale=0.7] {};
\node at (5.5,1.7-6)  [circle,draw=black,fill=black, scale=0.7] {};
\node at (4.5,2.5-6) [circle,draw=black,fill=black, scale=0.7] {};
\node at (5.5,2.5-6)  [circle,draw=black,fill=black, scale=0.7] {};

\end{tikzpicture}
\end{center}
\caption{The cubic graphs $G$ such that $|V(G)|<18$ and $\Gamma(G)<4$.}
\label{figcuu}
\end{figure}
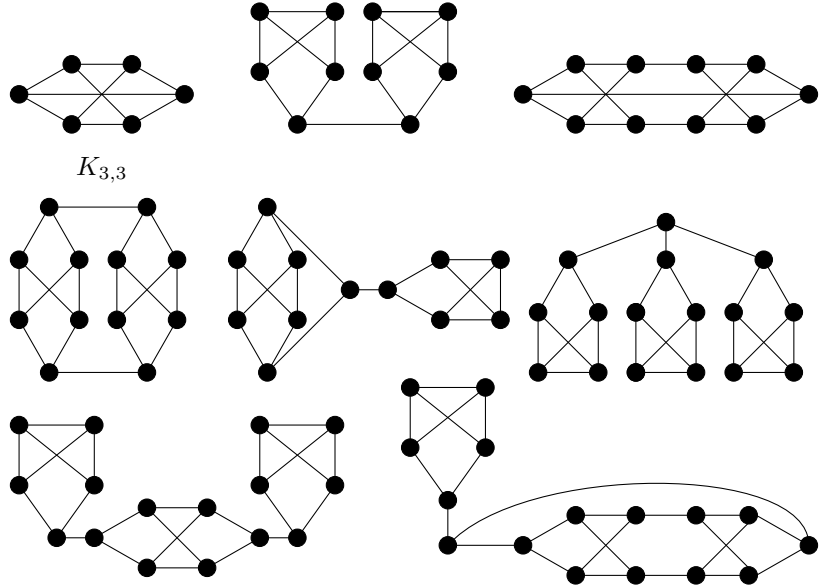
\begin{proof}
Every graph $G$ in $\mathcal{F}_3$ has three kind of vertices: $(0,3)$-twin-vertices (called also false twins), vertices where an edge is added by Point 3 and vertices added by Point 4. Vertices where an edge is added by Point 3 are $(1,3)$-twin-vertex and vice versa. Vertices added by Point 4 are $(2,3)$-twin-vertices and vice versa.
\end{proof}
\begin{theo}
Let $G$ be a cubic graph. $\Gamma(G)\le 3$ if and only if every vertex is an $(i,3)$-twin-vertex, for some $i$, $0 \le i \le 2$.
\end{theo}
\begin{proof}
By Corollary \ref{itwin}, the "if" part is proven.
Assume that $G$ contains a vertex $v$ which is not an $(i,3)$-twin-vertex, for some $i$, $0 \le i \le 2$ and $\Gamma(G)<4$.
In every configuration we want to either find a Grundy partial 4-coloring, contradicting $\Gamma(G)<4$ or proving that $v$ is an $(i,3)$-twin-vertex, for some $i$, with $0 \le i \le 2$.
We will refer to a given Grundy partial 4-coloring by its reference in Figure~\ref{figc}.
We consider three cases: $v$ or a neighbor of $v$ is in a $C_3$, $v$ is in an induced $C_4$ and $v$ or a neighbor of $v$ are not in a $C_3$ and $v$ is not in an induced $C_4$.
Let $C$ be an induced cycle of order 3 or 4 which contains $v$ or a neighbor of $v$ and let $D_1=\{ x\in V(G)|d(x,C)=1\}$, where $d(x,C)$ is the distance from $x$ to $C$ in the graph $G$. To simplify notation, $D_1$ will also denote the subgraph of $G$ induced by $D_1$.
\begin{description}
\item[Case 1:] Assume that $v$ or a neighbor of $v$ is in $C$ and $C=C_3$.
If $|D_1|=1$, then $G=K_4$ and $\Gamma(K_4)=4$.
If $|D_1|=2$ and $D_1=P_2$, then $v$ is a $(0,3)$-twin-vertex or a $(1,3)$-twin-vertex. If $D_1=I_2$ then Figure 3.1.a yields a Grundy partial 4-coloring of $G$.
If $|D_1|=3$, then we have four subcases: $D_1$ is $C_3$ or $P_3$ (Figure 3.1.b), $P_2\cup I_1$ (Figure 3.1.c) or $I_3$ (Figure 3.1.d). In every case $G$ admits a Grundy partial 4-coloring.
\item[Case 2:] Assume that $v$ is in $C$ and $C=C_4$. Note that for two non adjacent vertices of $C$ who have a common neighbor in $D_1$, the vertex $v$ is a $(0,3)$-twin-vertex or a $(1,3)$-twin-vertex. Hence, we will not consider these cases.
If $|D_1|=2$, then $D_1=P_2$ or $D_1=I_2$ (Figure 3.2.a) and in both cases, $G$ admits a Grundy partial 4-coloring.
If $|D_1|=3$, Figure 3.2.b yields a Grundy partial 4-coloring of $G$.
In the case $|D_1|=4$, we first assume that two adjacent vertices of $C$ have their neighbors in $D_1$ adjacent (Figure 3.2.c). Afterwards, we suppose that the previous case does not happen and that two non adjacent vertices of $C$ have their neighbors in $D_1$ adjacent (Figure 3.2.d).
In the case $D_1=I_4$, we first suppose that two vertices of $D_1$ which have two adjacent vertices of $C$ as neighbor, are not adjacent to two common vertices (Figure 3.2.e) and after consider they are (Figure 3.2.f).
\item[Case 3:] Assume that $v$ or a neighbor of $v$ is not in a $C_3$ and $v$ is not in an induced $C_4$.
Firstly, suppose that a neighbor $u$ of $v$ is in an induced $C_4$. Using the coloring from the previous case, $G$ admits a Grundy partial 4-coloring in every cases except in the case where two neighbors of $v$ in the $C_4$ have a common neighbor outside the $C_4$.
However, this case cannot happen for every neighbor of $v$, otherwise $v$ would be a $(2,3)$-twin-vertex. 
Assume that $u$ is the neighbor of $v$ not in the previous configuration.
If $u$ is in an induced $C_4$, then using the coloring from the previous case, $G$ admits a Grundy partial 4-coloring. If $u$ is not in an induced $C_4$, then Figure 3.3.a yields a Grundy partial 4-coloring of $G$. In this figure, the color 2 is given to a neighbor of $u$ not adjacent to both $f_1$ and $f_2$.
Secondly, suppose that $v$ is in an induced $C_5$. Figure 3.3.b yields a Grundy partial 4-coloring of $G$.
Thirdly, if $v$ is not in an induced $C_5$, then Figure 3.3.c yields a Grundy partial 4-coloring of $G$.
\end{description}
Therefore, if $\Gamma(G)\le 3$, then every vertex is an $(i,3)$-twin-vertex, for some $i$, $0 \le i \le 2$.
\end{proof}
Observe that if an edge is added between the two vertices of degree 2 in $K_{3,3}^{*}$, then we obtain $K_{3,3}$ which has Grundy number 2. By Proposition \ref{ggg2}, in all the remaining cases, the cubic graphs which have Grundy number at most 3 are different from complete bipartite graphs. Therefore, they have Grundy number 3.
\begin{figure}[!]
\begin{center}
\begin{tikzpicture}
\draw (0,0) -- (0.6,1);
\draw (1.2,0) -- (0.6,1);
\draw (0,0) -- (1.2,0);
\draw (0,0) -- (-0.6,1);
\draw (-0.6,1) -- (0.6,1);
\draw (1.8,-0.5) -- (1.2,0);
\node at (0,0) [circle,draw=red!50,fill=red!20] {};
\node at (0,0) {2};
\node at (0.6,1) [circle,draw=green!50,fill=green!20] {};
\node at (0.6,1) {3};
\node at (1.2,0)[circle,draw=yellow!50,fill=yellow!20] {};
\node at (1.2,0) {4};
\node at (-0.6,1) [circle,draw=blue!50,fill=blue!20] {};
\node at (-0.6,1) {1};
\node at (1.8,-0.5) [circle,draw=blue!50,fill=blue!20] {};
\node at (1.8,-0.5) {1};
\node at (0.6,-1.2){1.a};

\draw (0+3.5,0) -- (0.6+3.5,1);
\draw (1.2+3.5,0) -- (0.6+3.5,1);
\draw (0+3.5,0) -- (1.2+3.5,0);
\draw (0.6+3.5,1) -- (0.6+3.5,1.5);
\draw (-0.6+3.5,-0.5) -- (0+3.5,0);
\draw (1.8+3.5,-0.5) -- (1.2+3.5,0);
\draw (0.6+3.5,1.5) -- (-0.6+3.5,-0.5);
\draw (0.6+3.5,1.5) -- (1.8+3.5,-0.5);
\draw[dashed] (-0.6+3.5,-0.5) -- (1.8+3.5,-0.5);
\node at (-0.6+3.5,-0.5) [circle,draw=blue!50,fill=blue!20] {};
\node at (-0.6+3.5,-0.5) {1};
\node at (1.2+3.5,0) [circle,draw=blue!50,fill=blue!20] {};
\node at (1.2+3.5,0) {1};
\node at (0+3.5,0) [circle,draw=red!50,fill=red!20] {};
\node at (0+3.5,0) {2};
\node at (1.8+3.5,-0.5) [circle,draw=red!50,fill=red!20] {};
\node at (1.8+3.5,-0.5) {2};
\node at (0.6+3.5,1) [circle,draw=green!50,fill=green!20] {};
\node at (0.6+3.5,1) {3};
\node at (0.6+3.5,1.5) [circle,draw=yellow!50,fill=yellow!20] {};
\node at (0.6+3.5,1.5) {4};
\node at (4.1,-1.2){1.b};

\draw (0+7,0) -- (0.6+7,1);
\draw (1.2+7,0) -- (0.6+7,1);
\draw (0+7,0) -- (1.2+7,0);
\draw (0.6+7,1) -- (0.6+7,1.5);
\draw (1.8+7,-0.5) -- (1.2+7,0);
\draw (0.6+7,1.5) -- (1.8+7,-0.5);
\draw (0.6+7,1.5) -- (1.8+7,1.5);
\draw[dashed] (1.8+7,1.5) -- (1.8+7,-0.5);
\node at (0.6+7,1) [circle,draw=green!50,fill=green!20] {};
\node at (0.6+7,1) {3};
\node at (0.6+7,1.5)[circle,draw=yellow!50,fill=yellow!20] {};
\node at (0.6+7,1.5) {4};
\node at (1.8+7,-0.5) [circle,draw=red!50,fill=red!20] {};
\node at (1.8+7,-0.5) {2};
\node at (0+7,0) [circle,draw=red!50,fill=red!20] {};
\node at (0+7,0) {2};
\node at (1.2+7,0) [circle,draw=blue!50,fill=blue!20] {};
\node at (1.2+7,0) {1};
\node at (1.8+7,1.5) [circle,draw=blue!50,fill=blue!20] {};
\node at (1.8+7,1.5) {1};
\node at (7.6,-1.2){1.c};

\draw (0,0-3.3) -- (0.6,1-3.3);
\draw (1.2,0-3.3) -- (0.6,1-3.3);
\draw (0,0-3.3) -- (1.2,0-3.3);
\draw (0.6,1-3.3) -- (0.6,1.5-3.3);
\draw (-0.6,-0.5-3.3) -- (0,0-3.3);
\draw (1.8,-0.5-3.3) -- (1.2,0-3.3);
\node at (-0.6,-0.5-3.3) [circle,draw=blue!50,fill=blue!20] {};
\node at (-0.6,-0.5-3.3) {1};
\node at (1.8,-0.5-3.3) [circle,draw=blue!50,fill=blue!20] {};
\node at (1.8,-0.5-3.3) {1};
\node at (0.6,1.5-3.3) [circle,draw=blue!50,fill=blue!20] {};
\node at (0.6,1.5-3.3) {1};
\node at (0.6,1-3.3) [circle,draw=red!50,fill=red!20] {};
\node at (0.6,1-3.3) {2};
\node at (0,0-3.3) [circle,draw=green!50,fill=green!20] {};
\node at (0,0-3.3) {3};
\node at (1.2,0-3.3) [circle,draw=yellow!50,fill=yellow!20] {};
\node at (1.2,0-3.3){4};
\node at (0.6,-4){1.d};

\draw (0+2.5,0.5-3) -- (1.2+2.5,1-3);
\draw (0+2.5,0.5-3) -- (1.2+2.5,0-3);
\draw (1.2+2.5,0-3) -- (1.2+2.5,1-3);
\draw (1.2+2.5,0-3) -- (2.4+2.5,0-3);
\draw (1.2+2.5,1-3) -- (2.4+2.5,1-3);
\draw (2.4+2.5,0-3) -- (2.4+2.5,1-3);
\draw (2.4+2.5,0-3) -- (3.6+2.5,0.5-3);
\draw (2.4+2.5,1-3) -- (3.6+2.5,0.5-3);
\draw[dashed] (2.5,0.5-3) -- (3.6+2.5,0.5-3);
\node at (1.2+2.5,1-3) [circle,draw=blue!50,fill=blue!20] {};
\node at (1.2+2.5,1-3) {1};
\node at (2.5,0.5-3) [circle,draw=red!50,fill=red!20] {};
\node at (2.5,0.5-3) {2};
\node at (2.4+2.5,1-3) [circle,draw=red!50,fill=red!20] {};
\node at (2.4+2.5,1-3) {2};
\node at (2.4+2.5,0-3) [circle,draw=green!50,fill=green!20] {};
\node at (2.4+2.5,0-3) {3};
\node at (1.2+2.5,0-3) [circle,draw=yellow!50,fill=yellow!20] {};
\node at (1.2+2.5,0-3){4};
\node at (3.6+2.5,0.5-3) [circle,draw=blue!50,fill=blue!20] {};
\node at (3.6+2.5,0.5-3) {1};
\node at (1.8+2.5,-4){2.a};

\draw (0+6.8,0.5-3) -- (1.2+6.8,1-3);
\draw (0+6.8,0.5-3) -- (1.2+6.8,0-3);
\draw (1.2+6.8,0-3) -- (1.2+6.8,1-3);
\draw (1.2+6.8,0-3) -- (2.4+6.8,0-3);
\draw (1.2+6.8,1-3) -- (2.4+6.8,1-3);
\draw (2.4+6.8,0-3) -- (2.4+6.8,1-3);
\draw (3.6+6.8,1-3) -- (2.4+6.8,1-3);
\draw (3.6+6.8,0-3) -- (2.4+6.8,0-3);
\draw[dashed] (0+6.8,0.5-3) -- (3.6+6.8,0-3);
\draw[dashed] (0+6.8,0.5-3) -- (3.6+6.8,1-3);
\draw[dashed] (3.6+6.8,0-3) -- (3.6+6.8,1-3);

\node at (3.6+6.8,0-3) [circle,draw=blue!50,fill=blue!20] {};
\node at (3.6+6.8,0-3) {1};
\node at (1.2+6.8,1-3) [circle,draw=blue!50,fill=blue!20] {};
\node at (1.2+6.8,1-3) {1};
\node at (0+6.8,0.5-3) [circle,draw=red!50,fill=red!20] {};
\node at(0+6.8,0.5-3) {2};
\node at (2.4+6.8,1-3) [circle,draw=red!50,fill=red!20] {};
\node at (2.4+6.8,1-3) {2};
\node at (2.4+6.8,0-3) [circle,draw=green!50,fill=green!20] {};
\node at (2.4+6.8,0-3) {3};
\node at (1.2+6.8,0-3) [circle,draw=yellow!50,fill=yellow!20] {};
\node at (1.2+6.8,0-3){4};
\node at (3.6+6.8,1-3) [circle,draw=black,fill=black] {};
\node at (1.8+6.8,-4){2.b};

\draw (1.2-0.3,0-5.6) -- (1.2-0.3,1-5.6);
\draw (0-0.3,0-5.6) -- (0-0.3,1-5.6);
\draw (0-0.3,0-5.6) -- (3.6-0.3,0-5.6);
\draw (0-0.3,1-5.6) -- (3.6-0.3,1-5.6);
\draw (2.4-0.3,0-5.6) -- (2.4-0.3,1-5.6);
\node at (2.4-0.3,0-5.6) [circle,draw=blue!50,fill=blue!20] {};
\node at (2.4-0.3,0-5.6) {1};

\node at (0-0.3,1-5.6) [circle,draw=blue!50,fill=blue!20] {};
\node at (0-0.3,1-5.6) {1};
\node at (0-0.3,0.-5.6) [circle,draw=red!50,fill=red!20] {};
\node at (0-0.3,0-5.6) {2};
\node at (2.4-0.3,1-5.6) [circle,draw=red!50,fill=red!20] {};
\node at (2.4-0.3,1-5.6) {2};
\node at (1.2-0.3,0-5.6) [circle,draw=green!50,fill=green!20] {};
\node at (1.2-0.3,0-5.6) {3};
\node at (1.2-0.3,1-5.6) [circle,draw=yellow!50,fill=yellow!20] {};
\node at (1.2-0.3,1-5.6){4};
\node at (3.6-0.3,0-5.6) [circle,draw=black,fill=black] {};
\node at (3.6-0.3,1-5.6) [circle,draw=black,fill=black] {};
\node at (1.8,-6.4){2.c};

\draw (0+4,0-5.6) -- (0+4,1-5.6);
\draw (1.2+4,0-5.6) -- (0+4,0-5.6);
\draw (1.2+4,1-5.6) -- (0+4,1-5.6);
\draw (1.2+4,0-5.6) -- (0.7+4,0.5-5.6);
\draw (1.2+4,1-5.6) -- (0.7+4,0.5-5.6);
\draw (1.2+4,0-5.6) -- (1.7+4,0.5-5.6);
\draw (1.2+4,1-5.6) -- (1.7+4,0.5-5.6);
\draw (1.7+4,0.5-5.6) -- (2.5+4,0.5-5.6);
\node at (0+4,1-5.6) [circle,draw=blue!50,fill=blue!20] {};
\node at (0+4,1-5.6) {1};
\node at (2.5+4,0.5-5.6) [circle,draw=blue!50,fill=blue!20] {};
\node at (2.5+4,0.5-5.6) {1};
\node at (0.7+4,0.5-5.6) [circle,draw=blue!50,fill=blue!20] {};
\node at (0.7+4,0.5-5.6) {1};
\node at (1.2+4,1-5.6) [circle,draw=red!50,fill=red!20] {};
\node at (1.2+4,1-5.6) {2};
\node at (0+4,0-5.6) [circle,draw=red!50,fill=red!20] {};
\node at (0+4,0-5.6) {2};
\node at (1.2+4,0-5.6) [circle,draw=green!50,fill=green!20] {};
\node at (1.2+4,0-5.6) {3};
\node at (1.7+4,0.5-5.6) [circle,draw=yellow!50,fill=yellow!20] {};
\node at (1.7+4,0.5-5.6) {4};
\node at (5.2,-6.4){2.d};

\draw (1.2+7,0-5.6) -- (1.2+7,1-5.6);
\draw (-0.7+7,0-5.6) -- (3.6+7,0-5.6);
\draw (0+7,1-5.6) -- (3.6+7,1-5.6);
\draw (2.4+7,0-5.6) -- (2.4+7,1-5.6);

\node at (1.2+7,1-5.6)  [circle,draw=blue!50,fill=blue!20] {};
\node at (1.2+7,1-5.6)  {1};
\node at (3.6+7,0-5.6)  [circle,draw=blue!50,fill=blue!20] {};
\node at (3.6+7,0-5.6)  {1};
\node at (2.4+7,1-5.6)  [circle,draw=red!50,fill=red!20] {};
\node at (2.4+7,1-5.6) {2};
\node at (0+7,0-5.6)  [circle,draw=red!50,fill=red!20] {};
\node at (0+7,0-5.6) {2};
\node at (-0.7+7,0-5.6)  [circle,draw=blue!50,fill=blue!20] {};
\node at (-0.7+7,0-5.6) {1};
\node at (1.2+7,0-5.6) [circle,draw=green!50,fill=green!20] {};
\node at (1.2+7,0-5.6) {3};
\node at (2.4+7,0-5.6) [circle,draw=yellow!50,fill=yellow!20] {};
\node at (2.4+7,0-5.6){4};
\node at (0+7,1-5.6) [circle,draw=black,fill=black] {};
\node at (3.6+7,1-5.6) [circle,draw=black,fill=black] {};
\node at (8.8,-6.4){2.e};

\draw (1.2,0-8) -- (1.2,1-8);
\draw (0,0-8) -- (3.6,0-8);
\draw (0,1-8) -- (4.8,1-8);
\draw (2.4,0-8) -- (2.4,1-8);
\draw (0,0-8) -- (1.8,0.35-8);
\draw (0,0-8) -- (1.8,0.7-8);
\draw (3.6,0-8) -- (1.8,0.35-8);
\draw (3.6,0-8) -- (1.8,0.7-8);

\node at (1.2,1-8)  [circle,draw=blue!50,fill=blue!20] {};
\node at (1.2,1-8)  {1};
\node at (3.6,0-8)  [circle,draw=blue!50,fill=blue!20] {};
\node at (3.6,0-8)  {1};
\node at (4.8,1-8)  [circle,draw=blue!50,fill=blue!20] {};
\node at (4.8,1-8)  {1};
\node at (1.2,0-8)  [circle,draw=red!50,fill=red!20] {};
\node at (1.2,0-8) {2};
\node at (3.6,1-8)  [circle,draw=red!50,fill=red!20] {};
\node at (3.6,1-8) {2};
\node at (2.4,1-8) [circle,draw=green!50,fill=green!20] {};
\node at (2.4,1-8) {3};
\node at (2.4,0-8) [circle,draw=yellow!50,fill=yellow!20] {};
\node at (2.4,0-8){4};
\node at (0,1-8) [circle,draw=black,fill=black] {};
\node at (0,0-8) [circle,draw=black,fill=black] {};
\node at (1.8,0.35-8) [circle,draw=black,fill=black] {};
\node at (1.8,0.7-8) [circle,draw=black,fill=black] {};
\node at (1.8,-8.8){2.f};

\draw (1.2+7.5,1-8) -- (1.2+6,0.5-8);
\draw (1.2+7.5,1-8) -- (1.2+7.5,0.5-8);
\draw (1.2+7.5,1-8) -- (1.2+9,0.5-8);
\draw (1.2+6,0.5-8) -- (1.2+5.5,0-8);
\draw (1.2+6,0.5-8) -- (1.2+6.5,0-8);
\draw  (1.2+7.5,0.5-8) -- (1.2+7,0-8);
\draw  (1.2+7.5,0.5-8) -- (1.2+8,0-8);
\draw (1.2+9,0.5-8) -- (1.2+8.5,0-8);
\draw (1.2+9,0.5-8) -- (1.2+9.5,0-8);
\draw (1.2+5.5,0-8) -- (1.2+5.5,-0.8-8);
\draw (1.2+7,0-8) -- (1.2+7,-0.8-8);
\draw (1.2+8,0-8) -- (1.2+8,-0.8-8);
\draw (1.2+8.5,0-8) -- (1.2+8.5,-0.8-8);
\draw (1.2+9.5,0-8) -- (1.2+9.5,-0.8-8);
\draw (1.2+7,0-8) -- (1.2+8,-0.8-8);
\draw (1.2+8,0-8) -- (1.2+7,-0.8-8);
\draw (1.2+8.5,0-8) -- (1.2+9.5,-0.8-8);
\draw (1.2+9.5,0-8) -- (1.2+8.5,-0.8-8);
\draw[dashed]  (1.2+6.5,0-8) -- (1.2+7,-0.8-8);
\draw[dashed]  (1.2+6.5,0-8) -- (1.2+8,-0.8-8);
\draw[dashed]  (1.2+5.5,0-8) .. controls (1.2+6.9,-9.6) .. (1.2+8.5,-0.8-8);

\node at (1.2+6,0.5-8) [circle,draw=green!50,fill=green!20] {};
\node at (1.2+6,0.5-8) {3};
\node at (0.9+6,0.6-8) {$u$};
\node at (1.2+7.5,1-8) [circle,draw=yellow!50,fill=yellow!20] {};
\node at (1.2+7.5,1-8) {4};
\node at (0.9+7.5,1.1-8) {$v$};
\node at (1.2+7.5,0.5-8)  [circle,draw=red!50,fill=red!20] {};
\node at (1.2+7.5,0.5-8) {2};
\node at (1.2+5.5,0-8) [circle,draw=red!50,fill=red!20] {};
\node at (1.2+5.5,0-8) {2};
\node at (1.2+9,0.5-8) [circle,draw=blue!50,fill=blue!20] {};
\node at (1.2+9,0.5-8) {1};
\node at (1.2+7,0-8) [circle,draw=blue!50,fill=blue!20] {};
\node at (1.2+7,0-8) {1};
\node at (1.2+6.5,0-8) [circle,draw=blue!50,fill=blue!20] {};
\node at (1.2+6.5,0-8) {1};
\node at (1.2+5.5,-0.8-8) [circle,draw=blue!50,fill=blue!20] {};
\node at (1.2+5.5,-0.8-8) {1};
\node at (1.2+7,-0.8-8) [circle,draw=black,fill=black] {};
\node at (0.8+7,-0.8-8) {$f_1$};
\node at (1.2+8,-0.8-8) [circle,draw=black,fill=black] {};
\node at (0.8+8,-0.8-8) {$f_2$};
\node at (1.2+8,0-8) [circle,draw=black,fill=black] {};
\node at (1.2+9.5,-0.8-8) [circle,draw=black,fill=black] {};
\node at (1.2+9.5,0-8) [circle,draw=black,fill=black] {};
\node at (1.2+8.5,-0.8-8) [circle,draw=black,fill=black] {};
\node at (1.2+8.5,0-8) [circle,draw=black,fill=black] {};

\node at (1.2+7.5,-9.5){3.a};

\draw (1,1-10.3) -- (2.5,1-10.3);
\draw (1,1-10.3) -- (0.5,0.5-10.3);
\draw (1,1-10.3) -- (1.5,0.5-10.3);
\draw (2.5,1-10.3) -- (2,0.5-10.3);
\draw (2.5,1-10.3) -- (3,0.5-10.3);
\draw (1.75,0-10.3) -- (2,0.5-10.3);
\draw (1.75,0-10.3) -- (1.5,0.5-10.3);

\node at (2.5,1-10.3) [circle,draw=green!50,fill=green!20] {};
\node at (2.5,1-10.3) {3};
\node at (1,1-10.3) [circle,draw=yellow!50,fill=yellow!20] {};
\node at (1,1-10.3) {4};
\node at (0.7,1.1-10.3) {$v$};
\node at (1.5,0.5-10.3) [circle,draw=red!50,fill=red!20] {};
\node at (1.5,0.5-10.3) {2};
\node at (2,0.5-10.3) [circle,draw=red!50,fill=red!20] {};
\node at (2,0.5-10.3) {2};
\node at (0.5,0.5-10.3) [circle,draw=blue!50,fill=blue!20] {};
\node at (0.5,0.5-10.3) {1};
\node at (3,0.5-10.3) [circle,draw=blue!50,fill=blue!20] {};
\node at (3,0.5-10.3) {1};
\node at (3,0.5-10.3) [circle,draw=blue!50,fill=blue!20] {};
\node at (3,0.5-10.3)  {1};
\node at (1.75,0-10.3) [circle,draw=blue!50,fill=blue!20] {};
\node at (1.75,0-10.3)  {1};
\node at (1.75,-0.7-10.3){3.b};

\draw (1+4,1-10.3) -- (2.5+4,1-10.3);
\draw (1+4,1-10.3) -- (0.5+4,0.5-10.3);
\draw (1+4,1-10.3) -- (1.5+4,0.5-10.3);
\draw (2.5+4,1-10.3) -- (2+4,0.5-10.3);
\draw (2.5+4,1-10.3) -- (3+4,0.5-10.3);
\draw (2+4,0-10.3) -- (2+4,0.5-10.3);
\draw (1.5+4,0-10.3) -- (1.5+4,0.5-10.3);
\draw[dashed] (1.5+4,0-10.3) .. controls (1.75+4,-0.5-10.3) .. (2+4,0-10.3);
\draw (1.5+4,0.5-10.3) -- (1+4,0-10.3);

\node at (2.5+4,1-10.3) [circle,draw=green!50,fill=green!20] {};
\node at (2.5+4,1-10.3) {3};
\node at (1+4,1-10.3) [circle,draw=yellow!50,fill=yellow!20] {};
\node at (1+4,1-10.3) {4};
\node at (0.7+4,1.1-10.3) {$v$};
\node at (1.5+4,0.5-10.3) [circle,draw=red!50,fill=red!20] {};
\node at (1.5+4,0.5-10.3) {2};
\node at (2+4,0.5-10.3) [circle,draw=red!50,fill=red!20] {};
\node at (2+4,0.5-10.3) {2};
\node at (0.5+4,0.5-10.3) [circle,draw=blue!50,fill=blue!20] {};
\node at (0.5+4,0.5-10.3) {1};
\node at (3+4,0.5-10.3) [circle,draw=blue!50,fill=blue!20] {};
\node at (3+4,0.5-10.3) {1};
\node at (3+4,0.5-10.3) [circle,draw=blue!50,fill=blue!20] {};
\node at (3+4,0.5-10.3)  {1};
\node at (1+4,0-10.3) [circle,draw=blue!50,fill=blue!20] {};
\node at (1+4,0-10.3) {1};
\node at (2+4,0-10.3) [circle,draw=blue!50,fill=blue!20] {};
\node at (2+4,0-10.3) {1};
\node at (1.5+4,0-10.3) [circle,draw=black,fill=black] {};
\node at (1.75+4,-0.7-10.3){3.c};

\end{tikzpicture}
\end{center}
\caption{Possible configurations in a cubic graph (bold vertices: Uncolored vertices, vertices with number $i$: Vertices of color $i$).}
\label{figc}
\end{figure}
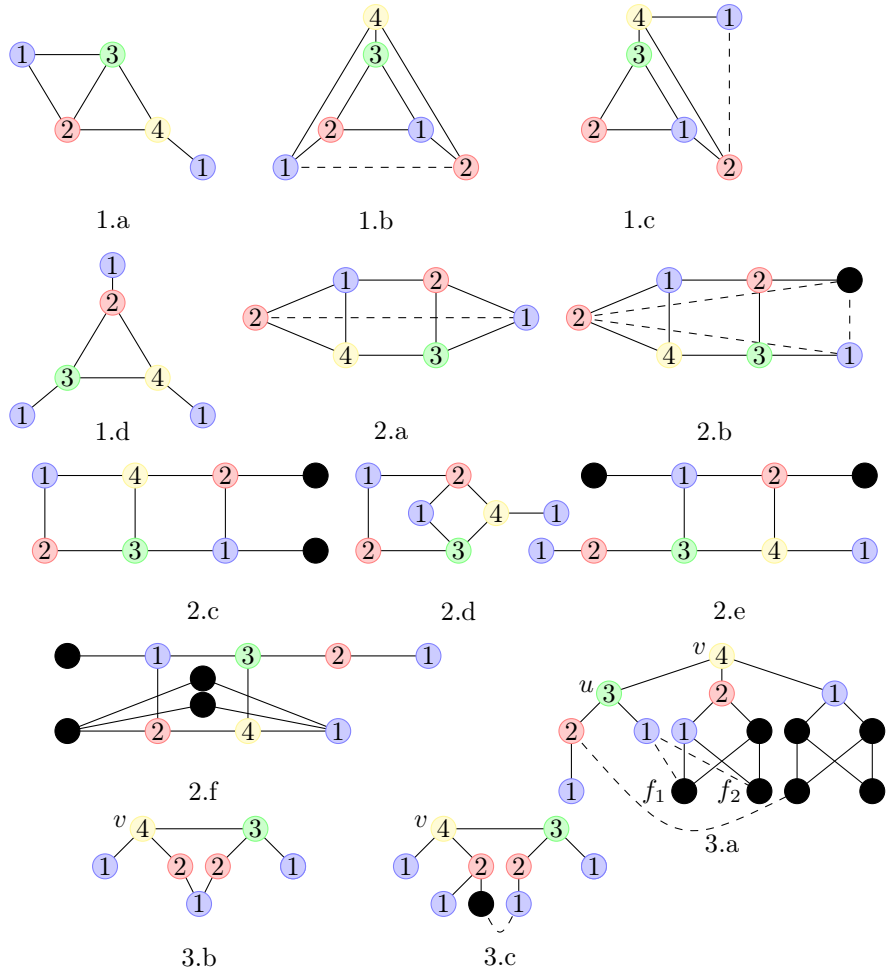
\begin{cor}
A cubic graph $G$ does not contain any induced minimal subcubic 4-atom if and only if every vertex is an $(i,3)$-twin-vertex, for some $i$, $0\le i \le 2$.
\end{cor}
\begin{cor}
Let $G$ be a cubic graph. If $G$ is without induced $C_4$, then $\Gamma(G)=4$.
\label{indc3}
\end{cor}
\begin{proof}
As every graph $G$ with $\Gamma(G)<4$ is composed of copies of $K_{2,3}$ or $K^{*}_{3,3}$, the graph $G$ always contains a square if $\Gamma(G)<4$.
\end{proof}
For a fixed integer $t$, the largest $(t+1)$-atom has order $2^{t}$. Thus, for a graph $G$ of maximum degree $t$, there exists an $O(n^{2^{t}})$-time algorithm to determine if $\Gamma(G)<t+1$ (which verifies if the graph contains an induced $(t+1)$-atom). For a cubic graph, we obtain an $O(n^{8})$-time algorithm, whereas our characterization yields a linear-time algorithm. 

\begin{obs}
Let $G$ be a cubic graph of order $n$. There exists an $O(n)$-time algorithm\footnote{Independently of our work, Yahiaoui et al. \cite{YA2012} have established a different algorithm to determine if the Grundy number of a cubic graph is 4.} to determine the Grundy number of $G$.
\end{obs}
\begin{proof}
Suppose we have a cubic graph $G$ with its adjacency list. Verifying if $G$ is $K_{3,3}$ can be done in constant time. We suppose now that $G$ is not $K_{3,3}$. For each vertex $v$, the algorithm verifies that $v$ is an $(i,3)$-twin-vertex, for some $i$, $0\le i \le 2$. If the condition is true for all vertices, then $\Gamma(G)=3$, else $\Gamma(G)=4$. To determine if a vertex $v$ is a $(0,3)$-twin-vertex, it suffices to verify that there is a common vertex other than $v$ in the adjacency lists of the neighbors of $v$. To determine if a vertex $v$ is a $(1,3)$-twin-vertex, it suffices to verify that there are two neighbors of $v$ which have the same adjacency list. To determine if a vertex $v$ is a $(2,3)$-twin-vertex, it suffices to verify that the neighborhood of $v$ is independent and that every neighbor is a $(1,3)$-twin-vertex. Hence, checking if a vertex is an $(i,3)$-twin vertex can be done in constant time, so the algorithms runs in linear time.
\end{proof}
\begin{prop}
If $G$ is a connected cubic graph and $G\neq K_{3,3}$, then $\partial \Gamma(G)=4$.
\end{prop}
\begin{proof}
Let $G$ be a cubic connected graph. Note that if $\Gamma(G)=4$ then $\partial \Gamma(G)=4$.
Every graph $G$ with $\Gamma(G)<4$ is composed of copies of $K_{2,3}$ or $K^{*}_{3,3}$. If $G$ contains more than two copies (so it is different from $K_{3,3}$), then a vertex can be colored 4 in the first copy and a vertex can be colored 3 in the second copy. Hence, $\partial \Gamma(G)=4$.
\end{proof}
Only $K_{3,3}$ and three other cubic graphs have $b$-chromatic number at most 3 \cite{KL2010}. Thus, our result is coherent with the results on the $b$-chromatic number.
Shi et al.~\cite{SH2005} proved that there exists a smallest integer $N_r$ such that every $r$-regular graph $G$ with more than $N_r$ vertices has $\partial\Gamma(G)=r+1$. Observe that we have $N_2=4$ and $N_3=6$. It is an open question to determine $N_r$ for $r\ge 4$. However, using results on $b$-chromatic number \cite{CA2011}, we have $N_r\le2 r^3-r^{2}+r$.
\section{Properties on the Grundy number of $r$-regular graphs}
\begin{de}
Let $r\ge2$ be an integer.
We define recursively the family of graphs $\mathcal{G}^{*}_r$ as follows:
\begin{enumerate}
\item $K_{r-k,k+2}  \in \mathcal{G}^{*}_r$, for any $k$, $0 \le k \le (r-2)/2$;
\item the disjoint union of two elements of $\mathcal{G}^{*}_{r}$ is in $\mathcal{G}^{*}_r$;
\item if $G$ is a graph in $\mathcal{G}^{*}_r$, then the graph $H$ obtained from $G$ by adding an edge between two vertices of degree at most $r-1$ is also in $\mathcal{G}^{*}_r$;
\item if $G$ is a graph in $\mathcal{G}^{*}_r$, then the graph $H$ obtained from $G$ by adding a new vertex adjacent to $r$ vertices of degree at most $r-1$ is in $\mathcal{G}^{*}_r$.
\end{enumerate}
The family $\mathcal{G}_r$ is the subfamily of $r$-regular graphs in $\mathcal{G}^{*}_r$.
\end{de}

\begin{prop}
Let $G$ be an $r$-regular graph. If $G \in \mathcal{G}_r$, then $\Gamma(G)<r+1$.
\end{prop}
\begin{proof}
By $I_{r-k}$ and $I_{k+2}$, with $|I_{r-k}|=r-k$ and $|I_{k+2}|=k+2$, we denote the two sets of vertices in the bipartition of an induced subgraph $K_{r-k,k+2}$ in $G$.
Firstly, suppose there exists a vertex $u$ in an induced subgraph $K_{r-k,k+2}$ colored $r+1$. Without loss of generality, suppose $u$ is in $I_{r-k}$. The $r$ neighbors of $u$ should have colors from 1 to $r$. Among the neighbors of $u$, $k+2$ neighbors are in $I_{k+2}$.
Let $v$ be the neighbor of $u$ in $I_{k+2}$ with the largest color in $I_{k+2}$. The vertex $v$ has color at least $k+2$. Hence, there exists an integer $s\ge 0$ such that the color of $v$ is $k+2+s$.
Note that there are $s$ vertices in $N(u)\setminus I_{k+2}$ which have colors at most $k+2+s$.
The colors of the $s$ vertices are the only one possible remaining colors at most $k+2+s$ in $I_{r-k}$.
Hence, as there are $k$ vertices in $N(v)\setminus I_{r-k}$, the neighbors of $v$ can only have at most $k+s$ different colors at most $k+2+s$.
Therefore, we have a contradiction and $u$ cannot have color $r+1$.
Secondly, suppose there exists a vertex $u$ added by Point 4 which has color $r+1$. As a neighbor of $u$ in an induced $K_{r-k,k+2}$ should be colored $r$, the argument is completely similar to the previous one.
\end{proof}

\begin{cor}
Let $G$ be a 4-regular graph. If $G\in \mathcal{G}_4$, then $\Gamma(G)<5$.
\end{cor}
The reader can believe that the family of 4-regular graphs with $\Gamma(G)< 5$ contains only the family $\mathcal{G}_4$. However, there exist graphs with Grundy number $r$ which are not inside this family. For example, the power graph (the graph where every pair of vertices at pairwise distance 2 become adjacent) of the 7-cycle $C_7^2$ satisfies $\Gamma(C_7^2)<5$ and is not in $\mathcal{G}_4$.

The next proposition shows that unlike the $b$-chromatic number, $r$-regular graphs of order arbitrarily large with Grundy number $k$ can be constructed for any $r$ and any $k$, $3\le k\le r+1$.
\begin{prop}
Let $r\ge4$ and $3\le k\le r+1$ be integers. There exists an infinite family $\mathcal{H}$ of connected $r$-regular graphs such that for all $G$ in  $\mathcal{H}$, $\Gamma(G)=k$.
\end{prop}
\begin{proof}
Let $i\ge 2$ be a positive integer and $r_1$, $\ldots$, $r_{k-1}$ be a sequence of positive integers such that $r=r_1+\ldots+r_{k-1}$. We construct a graph $G_{r,k,i}$ as follows:
Take $2i$ copies of $K_{r_1,\ldots,r_{k-1}}$. Let $H_{j-1}$ be the copy number $j$ of $K_{r_1,\ldots,r_{k-1}}$ and $H_{j,r_l}$ be the independent $r_l$-set in $H_j$.
If $j\equiv 1\pmod{2}$, do the graph join of $H_{j\pmod{2i},r_1}$ and $H_{j-1\pmod{2i},r_1}$ and for an integer $l$, $1<l<k$, do the graph join of $H_{j\pmod{2i},r_l}$ and $H_{j+1\pmod{2i},r_l}$.
The $r$-regular graph obtained is the graph $G_{r,k,i}$. Figure~\ref{figmr} gives $G_{r,k,i}$, for $k=4$ and $i\ge2$. Note that $H_{j,r_i}$ is an independent module. Thus, every vertex is a $(0,k)$-twin-vertex. By Proposition \ref{0twin}, $\Gamma(G_{r,k,i})\le k$.\newline
For an integer $l$, $1<l<k$, color one vertex $l-1$ in $H_{1,r_l}$ and $H_{2,r_l}$.
Afterwards, color one vertex $k-1$ in $H_{1,r_1}$ and one vertex $k$ in $H_{2,r_1}$. The given coloring is a Grundy partial $k$-coloring of $G_{r,k,i}$ for $i\ge2$.
Therefore, $\Gamma(G_{r,k,i})=k$, for $i\ge2$.
\end{proof}
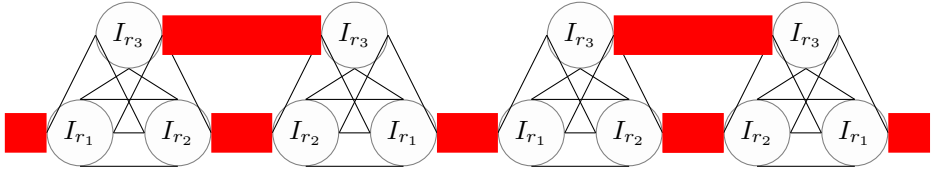
\begin{figure}[t]
\begin{center}
\begin{tikzpicture}
\node at (0-3,0) (t1)[circle,draw=black!50,fill=black!01] {$I_{r_{1}}$};
\node at  (1.3-3,0) (t2)[circle,draw=black!50,fill=black!01] {$I_{r_{2}}$};
\node at (0.65-3,1.3)(t3)[circle,draw=black!50,fill=black!01] {$I_{r_{3}}$};
\node at  (0,0) (t4)[circle,draw=black!50,fill=black!01] {$I_{r_{2}}$};
\node at  (1.3,0) (t5)[circle,draw=black!50,fill=black!01] {$I_{r_{1}}$};
\node at  (0.65,1.3) (t6)[circle,draw=black!50,fill=black!01] {$I_{r_{3}}$};
\node at  (0+3,0) (t7)[circle,draw=black!50,fill=black!01] {$I_{r_{1}}$};
\node at  (1.3+3,0) (t8)[circle,draw=black!50,fill=black!01] {$I_{r_{2}}$};
\node at  (0.65+3,1.3) (t9)[circle,draw=black!50,fill=black!01] {$I_{r_{3}}$};
\node at  (0+6,0) (t10)[circle,draw=black!50,fill=black!01] {$I_{r_{2}}$};
\node at (1.3+6,0) (t11)[circle,draw=black!50,fill=black!01] {$I_{r_{1}}$};
\node at  (0.65+6,1.3) (t12)[circle,draw=black!50,fill=black!01] {$I_{r_{3}}$};
\draw (t1.east) -- (t2.west);
\draw (t1.south) -- (t2.south);
\draw (t1.north) -- (t2.north);
\draw (t1.west) -- (t3.west);
\draw (t1.north) -- (t3.south);
\draw (t1.east) -- (t3.east);
\draw (t2.west) -- (t3.west);
\draw (t2.north) -- (t3.south);
\draw (t2.east) -- (t3.east);
\draw (t4.east) -- (t5.west);
\draw (t4.south) -- (t5.south);
\draw (t4.north) -- (t5.north);
\draw (t4.west) -- (t6.west);
\draw (t4.north) -- (t6.south);
\draw (t4.east) -- (t6.east);
\draw (t5.west) -- (t6.west);
\draw (t5.north) -- (t6.south);
\draw (t5.east) -- (t6.east);
\draw (t7.east) -- (t8.west);
\draw (t7.south) -- (t8.south);
\draw (t7.north) -- (t8.north);
\draw (t7.west) -- (t9.west);
\draw (t7.north) -- (t9.south);
\draw (t7.east) -- (t9.east);
\draw (t8.west) -- (t9.west);
\draw (t8.north) -- (t9.south);
\draw (t8.east) -- (t9.east);
\draw (t10.east) -- (t11.west);
\draw (t10.south) -- (t11.south);
\draw (t10.north) -- (t11.north);
\draw (t10.west) -- (t12.west);
\draw (t10.north) -- (t12.south);
\draw (t10.east) -- (t12.east);
\draw (t11.west) -- (t12.west);
\draw (t11.north) -- (t12.south);
\draw (t11.east) -- (t12.east);

\draw[line width=15pt,color=red] (t1.west) -- (-4,0);
\draw[line width=15pt,color=red] (t2.east) -- (t4.west);
\draw[line width=15pt,color=red] (t3.east) -- (t6.west);
\draw[line width=15pt,color=red] (t5.east) -- (t7.west);
\draw[line width=15pt,color=red] (t8.east) -- (t10.west);
\draw[line width=15pt,color=red] (t9.east) -- (t12.west);
\draw[line width=15pt,color=red] (t8.east) -- (t10.west);
\draw[line width=15pt,color=red] (t11.east) -- (8.3,0);
\end{tikzpicture}
\end{center}
\caption{The Graph $G_{r,4,i}$, $i\ge2$, $r=r_1+r_2+r_3$.}
\label{figmr}
\end{figure}

\section{Grundy number of 4-regular graphs without induced $C_4$}
The following lemmas will be useful to prove the second main theorem of this paper: The family of 4-regular graphs without induced $C_4$ contains only graphs with Grundy number 5.
\begin{lem}
Let $G$ be a 4-regular graph without induced $C_4$. If $G$ contains (an induced) $K_4$ then $\Gamma(G)=5$.
\label{k4}
\end{lem}
\begin{proof}
Note that if $G=K_5$, we have $\Gamma(G)=5$. If $G$ is not $K_5$ then every pair of neighbors of vertices of $K_4$ cannot be adjacent ($G$ would contain a $C_4$). Giving the color 1 to each neighbor of the vertices of $K_4$ and colors 2, 3, 4, 5 to the vertices of $K_4$, we obtain a Grundy partial 5-coloring of $G$.
\end{proof}
\begin{lem}
Let $G$ be a 4-regular graph without induced $C_4$ and let $W$ be the graph from Figure~\ref{figm31}. If $G$ contains an induced $W$ then $\Gamma(G)=5$.
\end{lem}
\begin{proof}
The names of the vertices of $W$ come from Figure~\ref{figm31}. Depending on the different cases that could happen, Grundy partial 5-colorings of $G$ will be given using their references on Figure~\ref{figm31}.
Let $D_1$ be the set of vertices at distance 1 from vertices of $W$ in $G-W$.
Suppose that two vertices of $W$ have a common neighbor in $D_1$. This two vertices could only be $u_4$ and $u_5$ or $u_3$ and $u_5$ (or $u_1$ and $u_4$, by symmetry).
In the case that $u_4$ and $u_5$ have a common neighbor in $D_1$, colors will be given to neighbors of $u_3$ in $D_1$, depending if they are adjacent (Figure 5.1.a) or not (Figure 5.1.b).
In the case that $u_3$ and $u_5$ have a common neighbor $w$ in $D_1$, $w$ can be adjacent with a neighbor of $u_3$ in $D_1$ (Figure 5.2.a) or not (Figure 5.2.b).
Suppose now that no vertices in $W$ have a common neighbor in $D_1$. Let $w_1$ and $w_2$ be the neighbors of $u_3$ in $D_1$. We first consider that $w_1$ and $w_2$ are adjacent (Figure 5.3.a). Secondly, we consider that $w_1$ and $w_2$ are not adjacent and that $u_5$, $u_3$ and $w_1$ are in an induced $C_5$ (Figure 5.3.b). Finally, we consider that the previous configurations are impossible (Figure 5.3.c).
\end{proof}
\begin{figure}[t]
\begin{center}
\begin{tikzpicture}
\draw (0.8*0.8,0) -- (-0.8*0.8,0);
\draw (-0.8*0.8,0) -- (-0.4*0.8,0.8*0.8);
\draw (0.8*0.8,0) -- (0.4*0.8,0.8*0.8);
\draw (0,0) -- (0.4*0.8,0.8*0.8);
\draw (0,0) -- (-0.4*0.8,0.8*0.8);
\draw (0.4*0.8,0.8*0.8) -- (-0.4*0.8,0.8*0.8);
\node at (0,-0.8){The graph $W$.};
\node at (-0.8,-0.2){$u_1$};
\node at (0,-0.2){$u_2$};
\node at (0.8,-0.2){$u_3$};
\node at (-0.4,0.8){$u_4$};
\node at (0.4,0.8){$u_5$};

\draw (3.4*0.8+1,0) -- (1.2*0.8+1,0);
\draw (1.2*0.8+1,0) -- (1.6*0.8+1,0.8*0.8);
\draw (2.8*0.8+1,0) -- (2.4*0.8+1,0.8*0.8);
\draw (2*0.8+1,0) -- (2.4*0.8+1,0.8*0.8);
\draw (2*0.8+1,0) -- (1.6*0.8+1,0.8*0.8);
\draw (2.4*0.8+1,0.8*0.8) -- (1.6*0.8+1,0.8*0.8);
\draw (2.4*0.8+1,0.8*0.8) -- (2*0.8+1,1.6*0.8);
\draw (1.6*0.8+1,0.8*0.8) -- (2*0.8+1,1.6*0.8);
\draw (3.4*0.8+1,0) -- (3.1*0.8+1,0.6*0.8);
\draw (2.8*0.8+1,0) -- (3.1*0.8+1,0.6*0.8);

\node at (2.8*0.8+1,0) [circle,draw=green!50,fill=green!20] {};
\node at (2.8*0.8+1,0) {3};
\node at (2.4*0.8+1,0.8*0.8) [circle,draw=yellow!50,fill=yellow!20] {};
\node at (2.4*0.8+1,0.8*0.8){4};
\node at (2*0.8+1,0) [circle,draw=black!50,fill=black!20] {};
\node at (2*0.8+1,0) {5};
\node at (1.2*0.8+1,0) [circle,draw=red!50,fill=red!20] {};
\node at (1.2*0.8+1,0) {2};
\node at (1.6*0.8+1,0.8*0.8) [circle,draw=blue!50,fill=blue!20] {};
\node at (1.6*0.8+1,0.8*0.8) {1};
\node at (2*0.8+1,1.6*0.8) [circle,draw=red!50,fill=red!20] {};
\node at (2*0.8+1,1.6*0.8) {2};
\node at (3.1*0.8+1,0.6*0.8) [circle,draw=blue!50,fill=blue!20] {};
\node at (3.1*0.8+1,0.6*0.8) {1};
\node at (3.4*0.8+1,0) [circle,draw=red!50,fill=red!20] {};
\node at (3.4*0.8+1,0) {2};
\node at (2*0.8+1,-0.8){1.a};

\draw (3.4*0.8+3.7,0) -- (1.2*0.8+3.7,0);
\draw (1.2*0.8+3.7,0) -- (1.6*0.8+3.7,0.8*0.8);
\draw (2.8*0.8+3.7,0) -- (2.4*0.8+3.7,0.8*0.8);
\draw (2*0.8+3.7,0) -- (2.4*0.8+3.7,0.8*0.8);
\draw (2*0.8+3.7,0) -- (1.6*0.8+3.7,0.8*0.8);
\draw (2.4*0.8+3.7,0.8*0.8) -- (1.6*0.8+3.7,0.8*0.8);
\draw (2.4*0.8+3.7,0.8*0.8) -- (2*0.8+3.7,1.6*0.8);
\draw (1.6*0.8+3.7,0.8*0.8) -- (2*0.8+3.7,1.6*0.8);
\draw (2.8*0.8+3.7,0) -- (3.1*0.8+3.7,0.6*0.8);
\draw (3.4*0.8+3.7,1.2*0.8)-- (3.1*0.8+3.7,0.6*0.8);

\node at (2.8*0.8+3.7,0) [circle,draw=green!50,fill=green!20] {};
\node at (2.8*0.8+3.7,0) {3};
\node at (2.4*0.8+3.7,0.8*0.8) [circle,draw=yellow!50,fill=yellow!20] {};
\node at (2.4*0.8+3.7,0.8*0.8){4};
\node at (2*0.8+3.7,0) [circle,draw=black!50,fill=black!20] {};
\node at (2*0.8+3.7,0) {5};
\node at (1.2*0.8+3.7,0) [circle,draw=red!50,fill=red!20] {};
\node at (1.2*0.8+3.7,0) {2};
\node at (1.6*0.8+3.7,0.8*0.8) [circle,draw=blue!50,fill=blue!20] {};
\node at (1.6*0.8+3.7,0.8*0.8) {1};
\node at (2*0.8+3.7,1.6*0.8) [circle,draw=red!50,fill=red!20] {};
\node at (2*0.8+3.7,1.6*0.8) {2};
\node at (3.4*0.8+3.7,0)[circle,draw=blue!50,fill=blue!20] {};
\node at (3.4*0.8+3.7,0) {1};
\node at (3.1*0.8+3.7,0.6*0.8) [circle,draw=red!50,fill=red!20] {};
\node at (3.1*0.8+3.7,0.6*0.8) {2};
\node at (3.4*0.8+3.7,1.2*0.8) [circle,draw=blue!50,fill=blue!20] {};
\node at (3.4*0.8+3.7,1.2*0.8) {1};
\node at (2*0.8+3.7,-0.8){1.b};

\draw (3.6*0.8+6.5,0) -- (1.2*0.8+6.5,0);
\draw (1.2*0.8+6.5,0) -- (1.6*0.8+6.5,0.8*0.8);
\draw (2.8*0.8+6.5,0) -- (2.4*0.8+6.5,0.8*0.8);
\draw (2*0.8+6.5,0) -- (2.4*0.8+6.5,0.8*0.8);
\draw (2*0.8+6.5,0) -- (1.6*0.8+6.5,0.8*0.8);
\draw (3.2*0.8+6.5,0.8*0.8) -- (1.6*0.8+6.5,0.8*0.8);
\draw (3.6*0.8+6.5,0) -- (3.2*0.8+6.5,0.8*0.8);
\draw (2.8*0.8+6.5,0) -- (3.2*0.8+6.5,0.8*0.8);

\node at (2.8*0.8+6.5,0) [circle,draw=green!50,fill=green!20] {};
\node at (2.8*0.8+6.5,0) {3};
\node at (2.4*0.8+6.5,0.8*0.8) [circle,draw=yellow!50,fill=yellow!20] {};
\node at (2.4*0.8+6.5,0.8*0.8){4};
\node at (2*0.8+6.5,0) [circle,draw=black!50,fill=black!20] {};
\node at (2*0.8+6.5,0) {5};
\node at (1.2*0.8+6.5,0) [circle,draw=red!50,fill=red!20] {};
\node at (1.2*0.8+6.5,0) {2};
\node at (1.6*0.8+6.5,0.8*0.8) [circle,draw=blue!50,fill=blue!20] {};
\node at (1.6*0.8+6.5,0.8*0.8) {1};
\node at (3.2*0.8+6.5,0.8*0.8) [circle,draw=red!50,fill=red!20] {};
\node at (3.2*0.8+6.5,0.8*0.8) {2};
\node at (3.6*0.8+6.5,0) [circle,draw=blue!50,fill=blue!20] {};
\node at (3.6*0.8+6.5,0) {1};
\node at (8.5,-0.8){2.a};

\draw (3.4*0.8-2,-2.2) -- (1.2*0.8-2,-2.2);
\draw (1.2*0.8-2,-2.2) -- (1.6*0.8-2,0.8*0.8-2.2);
\draw (2.8*0.8-2,0-2.2) -- (2.4*0.8-2,0.8*0.8-2.2);
\draw (2*0.8-2,0-2.2) -- (2.4*0.8-2,0.8*0.8-2.2);
\draw (2*0.8-2,0-2.2) -- (1.6*0.8-2,0.8*0.8-2.2);
\draw (4*0.8-2,0.8*0.8-2.2) -- (1.6*0.8-2,0.8*0.8-2.2);
\draw (2.8*0.8-2,-2.2) -- (3.2*0.8-2,0.8*0.8-2.2);
\draw[dashed]  (1.2*0.8-2,-2.2) .. controls  (-0.8*0.8-1,-0.8) and  (1.2*0.8,-0.8) .. (4*0.8-2,0.8*0.8-2.2);

\node at (2.8*0.8-2,0-2.2) [circle,draw=green!50,fill=green!20] {};
\node at (2.8*0.8-2,0-2.2) {3};
\node at (2.4*0.8-2,0.8*0.8-2.2) [circle,draw=yellow!50,fill=yellow!20] {};
\node at (2.4*0.8-2,0.8*0.8-2.2){4};
\node at (2*0.8-2,0-2.2) [circle,draw=black!50,fill=black!20] {};
\node at (2*0.8-2,0-2.2) {5};
\node at (1.2*0.8-2,0-2.2) [circle,draw=red!50,fill=red!20] {};
\node at (1.2*0.8-2,0-2.2) {2};
\node at (1.6*0.8-2,0.8*0.8-2.2) [circle,draw=blue!50,fill=blue!20] {};
\node at (1.6*0.8-2,0.8*0.8-2.2) {1};
\node at (3.2*0.8-2,0.8*0.8-2.2) [circle,draw=red!50,fill=red!20] {};
\node at (3.2*0.8-2,0.8*0.8-2.2) {2};
\node at (3.6*0.8-2,0-2.2) [circle,draw=blue!50,fill=blue!20] {};
\node at (3.6*0.8-2,0-2.2) {1};
\node at (4*0.8-2,0.8*0.8-2.2) [circle,draw=blue!50,fill=blue!20] {};
\node at (4*0.8-2,0.8*0.8-2.2) {1};
\node at (0,-2.8){2.b};

\draw (1.4*0.8+2.5,-2.2) -- (-0.8*0.8+2.5,-2.2);
\draw (-0.8*0.8+2.5,0-2.2) -- (-0.4*0.8+2.5,0.8*0.8-2.2);
\draw (0.8*0.8+2.5,0-2.2) -- (0.4*0.8+2.5,0.8*0.8-2.2);
\draw (0+2.5,0-2.2) -- (0.4*0.8+2.5,0.8*0.8-2.2);
\draw (0+2.5,0-2.2) -- (-0.4*0.8+2.5,0.8*0.8-2.2);
\draw (0.4*0.8+2.5,0.8*0.8-2.2) -- (-0.4*0.8+2.5,0.8*0.8-2.2);
\draw (0.4*0.8+2.5,0.8*0.8-2.2) -- (1*0.8+2.5,0.8*0.8-1.9);
\draw (1.4*0.8+2.5,-2.2) -- (1.1*0.8+2.5,0.8*0.6-2.2);
\draw (0.8*0.8+2.5,-2.2) -- (1.1*0.8+2.5,0.8*0.6-2.2);

\node at (0+2.5,-2.2) [circle,draw=green!50,fill=green!20] {};
\node at (0+2.5,-2.2) {3};
\node at (0.4*0.8+2.5,0.8*0.8-2.2) [circle,draw=yellow!50,fill=yellow!20] {};
\node at (0.4*0.8+2.5,0.8*0.8-2.2) {4};
\node at (0.8*0.8+2.5,-2.2) [circle,draw=black!50,fill=black!20] {};
\node at (0.8*0.8+2.5,-2.2) {5};
\node at (1.4*0.8+2.5,-2.2) [circle,draw=red!50,fill=red!20] {};
\node at (1.4*0.8+2.5,-2.2) {2};
\node at (-0.4*0.8+2.5,0.8*0.8-2.2) [circle,draw=red!50,fill=red!20] {};
\node at (-0.4*0.8+2.5,0.8*0.8-2.2) {2};
\node at  (-0.8*0.8+2.5,-2.2) [circle,draw=blue!50,fill=blue!20] {};
\node at  (-0.8*0.8+2.5,-2.2) {1};
\node at (1*0.8+2.5,0.8*0.8-1.9) [circle,draw=blue!50,fill=blue!20] {};
\node at (1*0.8+2.5,0.8*0.8-1.9) {1};
\node at (1.1*0.8+2.5,0.8*0.6-2.2) [circle,draw=blue!50,fill=blue!20] {};
\node at (1.1*0.8+2.5,0.8*0.6-2.2) {1};
\node at (0.2+2.5,-2.8){3.a};

\draw (1.6*0.8+5.1,-2.2) -- (-0.8*0.8+5.1,-2.2);
\draw (-0.8*0.8+5.1,0-2.2) -- (-0.4*0.8+5.1,0.8*0.8-2.2);
\draw (0.8*0.8+5.1,0-2.2) -- (0.4*0.8+5.1,0.8*0.8-2.2);
\draw (0+5.1,0-2.2) -- (0.4*0.8+5.1,0.8*0.8-2.2);
\draw (0+5.1,0-2.2) -- (-0.4*0.8+5.1,0.8*0.8-2.2);
\draw (2*0.8+5.1,0.8*0.8-2.2) -- (-0.4*0.8+5.1,0.8*0.8-2.2);
\draw (1.6*0.8+5.1,-2.2) -- (2*0.8+5.1,0.8*0.8-2.2);
\draw (0.8*0.8+5.1,-2.2) -- (1.4*0.8+5.1,-2.6);
\draw[dashed] (-0.8*0.8+5.1,-2.2) .. controls  (-0.8*0.8+4.8,-0.8) and  (1.2*0.8+5.1,-0.8) .. (2*0.8+5.1,0.8*0.8-2.2);

\node at (0+5.1,-2.2) [circle,draw=green!50,fill=green!20] {};
\node at (0+5.1,-2.2) {3};
\node at (0.4*0.8+5.1,0.8*0.8-2.2) [circle,draw=yellow!50,fill=yellow!20] {};
\node at (0.4*0.8+5.1,0.8*0.8-2.2) {4};
\node at (0.8*0.8+5.1,-2.2) [circle,draw=black!50,fill=black!20] {};
\node at (0.8*0.8+5.1,-2.2) {5};
\node at (1.6*0.8+5.1,-2.2) [circle,draw=red!50,fill=red!20] {};
\node at (1.6*0.8+5.1,-2.2) {2};
\node at (-0.8*0.8+5.1,-2.2) [circle,draw=red!50,fill=red!20] {};
\node at (-0.8*0.8+5.1,-2.2) {2};
\node at (1.2*0.8+5.1,0.8*0.8-2.2) [circle,draw=red!50,fill=red!20] {};
\node at (1.2*0.8+5.1,0.8*0.8-2.2) {2};
\node at  (-0.4*0.8+5.1,0.8*0.8-2.2) [circle,draw=blue!50,fill=blue!20] {};
\node at  (-0.4*0.8+5.1,0.8*0.8-2.2) {1};
\node at (2*0.8+5.1,0.8*0.8-2.2) [circle,draw=blue!50,fill=blue!20] {};
\node at (2*0.8+5.1,0.8*0.8-2.2) {1};
\node at (1.4*0.8+5.1,-2.6) [circle,draw=blue!50,fill=blue!20] {};
\node at (1.4*0.8+5.1,-2.6) {1};
\node at (0.4+5.1,-2.8){3.b};

\draw (2.4*0.8+7.9,-2.2) -- (-0.8*0.8+7.9,-2.2);
\draw (-0.8*0.8+7.9,0-2.2) -- (-0.4*0.8+7.9,0.8*0.8-2.2);
\draw (0.8*0.8+7.9,0-2.2) -- (0.4*0.8+7.9,0.8*0.8-2.2);
\draw (0+7.9,0-2.2) -- (0.4*0.8+7.9,0.8*0.8-2.2);
\draw (0+7.9,0-2.2) -- (-0.4*0.8+7.9,0.8*0.8-2.2);
\draw (1.2*0.8+7.9,0.8*0.8-2.2) -- (-0.4*0.8+7.9,0.8*0.8-2.2);
\draw (0.8*0.8+7.9,-2.2) -- (1.4*0.8+7.9,-2.6);
\draw (1.6*0.8+7.9,-2.2) -- (2.2*0.8+7.9,-1.8);
\draw[dashed] (-0.8*0.8+7.9,-2.2) .. controls  (-0.8*0.8+7.6,-0.8) and  (1.2*0.8+7.8,-0.8) .. (2.2*0.8+7.9,-1.8);

\node at (0+7.9,-2.2) [circle,draw=green!50,fill=green!20] {};
\node at (0+7.9,-2.2) {3};
\node at (0.4*0.8+7.9,0.8*0.8-2.2) [circle,draw=yellow!50,fill=yellow!20] {};
\node at (0.4*0.8+7.9,0.8*0.8-2.2) {4};
\node at (0.8*0.8+7.9,-2.2) [circle,draw=black!50,fill=black!20] {};
\node at (0.8*0.8+7.9,-2.2) {5};
\node at (1.6*0.8+7.9,-2.2) [circle,draw=red!50,fill=red!20] {};
\node at (1.6*0.8+7.9,-2.2) {2};
\node at (1.2*0.8+7.9,0.8*0.8-2.2) [circle,draw=blue!50,fill=blue!20] {};
\node at (1.2*0.8+7.9,0.8*0.8-2.2){1};
\node at (-0.4*0.8+7.9,0.8*0.8-2.2) [circle,draw=red!50,fill=red!20] {};
\node at (-0.4*0.8+7.9,0.8*0.8-2.2) {2};
\node at (-0.8*0.8+7.9,-2.2)  [circle,draw=blue!50,fill=blue!20] {};
\node at (-0.8*0.8+7.9,-2.2) {1};
\node at (2.4*0.8+7.9,-2.2) [circle,draw=blue!50,fill=blue!20] {};
\node at (2.4*0.8+7.9,-2.2){1};
\node at (1.4*0.8+7.9,-2.6)[circle,draw=blue!50,fill=blue!20] {};
\node at (1.4*0.8+7.9,-2.6){1};
\node at (2.2*0.8+7.9,-1.8) [circle,draw=black,fill=black] {};
\node at (0.8*0.8+7.9,-2.8){3.c};

\end{tikzpicture}
\end{center}
\caption{Possible configurations when $G$ contains an induced $W$.}
\label{figm31}
\end{figure}
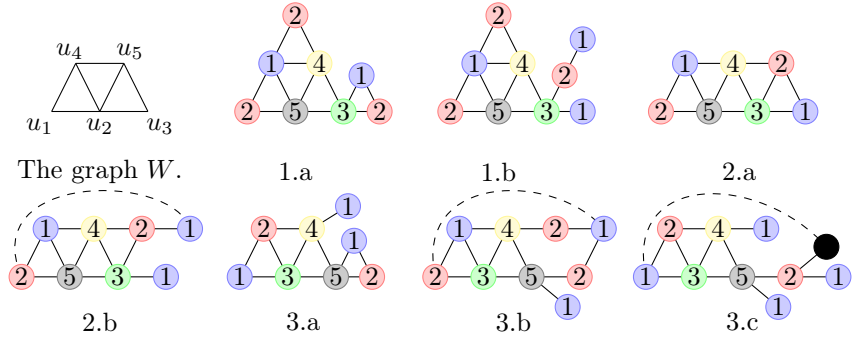

\begin{prop}
Let $G$ be a 4-regular graph without induced $C_4$. If $G$ contains $C_3$ then $\Gamma(G)=5$.
\label{cycle3}
\end{prop}
\begin{figure}[t]
\begin{center}
\begin{tikzpicture}
\draw (0-0.5,0) -- (0-0.5,0.8);
\draw (-0.8-0.5,0.8) -- (0-0.5,0);
\draw (-0.8-0.5,0) -- (0-0.5,0.8);
\draw (-0.8-0.5,0) -- (-0.8-0.5,0.8);
\draw (0.8-0.5,0) -- (0-0.5,0);
\draw (0.8-0.5,0) -- (0-0.5,0.8);
\draw (0.8-0.5,0.8) -- (0-0.5,0);
\draw (0.8-0.5,0.8) -- (0-0.5,0.8);
\draw (0.8-0.5,0.8) -- (1.2-0.5,1.2);
\draw (0.8-0.5,0.8) -- (1.2-0.5,0.4);
\draw[dashed] (1.2-0.5,1.2) -- (1.2-0.5,0.4);
\draw[dashed] (1.2-0.5,1.2) .. controls (-0.7,1.6) .. (-0.8-0.5,0.8);
\draw[dashed] (1.2-0.5,1.2) .. controls (-0.7,1.4) .. (-0.8-0.5,0);
\draw[dashed] (1.2-0.5,0.4) .. controls (0.9-0.5,-0.5) .. (-0.8-0.5,0);

\node at (0-0.5,0) [circle,draw=green!50,fill=green!20] {};
\node at (0-0.5,0) {3};
\node at (0-0.5,0.8) [circle,draw=yellow!50,fill=yellow!20] {};
\node at (0-0.5,0.8) {4};
\node at (-0.4-0.5,0.4) [circle,draw=black!50,fill=black!20] {};
\node at (-0.4-0.5,0.4) {5};
\node at (-0.8-0.5,0.8) [circle,draw=blue!50,fill=blue!20] {};
\node at (-0.8-0.5,0.8) {1};
\node at (-0.8-0.5,0) [circle,draw=red!50,fill=red!20] {};
\node at (-0.8-0.5,0) {2};
\node at (0.8-0.5,0) [circle,draw=blue!50,fill=blue!20] {};
\node at (0.8-0.5,0) {1};
\node at (0.8-0.5,0.8) [circle,draw=red!50,fill=red!20] {};
\node at (0.8-0.5,0.8) {2};
\node at (1.2-0.5,0.4) [circle,draw=blue!50,fill=blue!20] {};
\node at (1.2-0.5,0.4) {1};
\node at (1.2-0.5,1.2) [circle,draw=black,fill=black] {};
\node at (0-0.5,-0.4) {$u_2$};
\node at (0-0.5,1.025) {$u_1$};
\node at (0-0.5,-0.9){1.a};

\draw (0+2.3,-0.4) -- (0+2.3,1.2);
\draw (-0.8+2.3,0.8) -- (0+2.3,0);
\draw (-0.8+2.3,0) -- (0+2.3,0.8);
\draw (-0.8+2.3,0) -- (-0.8+2.3,0.8);
\draw (0.4+2.3,0.4) -- (0+2.3,0);
\draw (0.4+2.3,0.4) -- (0+2.3,0.8);
\draw (0.4+2.3,0.4) -- (0.8+2.3,0);
\draw (0.4+2.3,0.4) -- (0.8+2.3,0.8);
\draw[dashed] (0.8+2.3,0) -- (0.8+2.3,0.8);
\draw[dashed] (0.8+2.3,0) .. controls (2.3+0.2,1.8) .. (-0.8+2.3,0.8);
\draw[dashed] (0.8+2.3,0) .. controls (2.3,-0.9) .. (-0.8+2.3,0);
\draw[dashed] (0.8+2.3,0.8) .. controls (2.3,2) .. (-0.8+2.3,0.8);
\node at (0+2.3,0) [circle,draw=green!50,fill=green!20] {};
\node at (0+2.3,0) {3};
\node at (0+2.3,0.8) [circle,draw=yellow!50,fill=yellow!20] {};
\node at (0+2.3,0.8) {4};
\node at (-0.4+2.3,0.4) [circle,draw=black!50,fill=black!20] {};
\node at (-0.4+2.3,0.4) {5};
\node at (-0.8+2.3,0) [circle,draw=blue!50,fill=blue!20] {};
\node at (-0.8+2.3,0) {1};
\node at (-0.8+2.3,0.8) [circle,draw=red!50,fill=red!20] {};
\node at (-0.8+2.3,0.8) {2};
\node at (0.4+2.3,0.4) [circle,draw=red!50,fill=red!20] {};
\node at (0.4+2.3,0.4) {2};
\node at (0+2.3,-0.4) [circle,draw=blue!50,fill=blue!20] {};
\node at (0+2.3,-0.4) {1};
\node at (0+2.3,1.2) [circle,draw=blue!50,fill=blue!20] {};
\node at (0+2.3,1.2) {1};
\node at (0.8+2.3,0.8) [circle,draw=blue!50,fill=blue!20] {};
\node at (0.8+2.3,0.8) {1};
\node at (0.8+2.3,0) [circle,draw=black,fill=black] {};
\node at (0+2.3,-0.9){1.b};

\draw (0+4.1,0) -- (0+4.1,0.8);
\draw (-0.4+4.1,0.4) -- (0+4.1,0);
\draw (-0.4+4.1,0.4) -- (0+4.1,0.8);
\draw (1.6+4.1,0) -- (0+4.1,0);
\draw (0.8+4.1,0) -- (0+4.1,0.8);
\draw (0.8+4.1,0.8) -- (0+4.1,0);
\draw (1.6+4.1,0.8) -- (0+4.1,0.8);
\draw (1.6+4.1,0.8) -- (1.6+4.1,0);
\draw (0.8+4.1,0.8) -- (0.8+4.1,1.4);
\node at (0.8+4.1,0.8) [circle,draw=green!50,fill=green!20] {};
\node at (0.8+4.1,0.8) {3};
\node at (0+4.1,0.8) [circle,draw=yellow!50,fill=yellow!20] {};
\node at (0+4.1,0.8) {4};
\node at (0+4.1,0) [circle,draw=black!50,fill=black!20] {};
\node at (0+4.1,0) {5};
\node at (-0.4+4.1,0.4) [circle,draw=blue!50,fill=blue!20] {};
\node at (-0.4+4.1,0.4) {1};
\node at (0.8+4.1,0) [circle,draw=red!50,fill=red!20] {};
\node at (0.8+4.1,0) {2};
\node at (1.6+4.1,0) [circle,draw=blue!50,fill=blue!20] {};
\node at (1.6+4.1,0) {1};
\node at (1.6+4.1,0.8) [circle,draw=red!50,fill=red!20] {};
\node at(1.6+4.1,0.8) {2};
\node at (0.8+4.1,1.4) [circle,draw=blue!50,fill=blue!20] {};
\node at (0.8+4.1,1.4) {1};
\node at (0.8+4.1,-0.9){2.a};

\draw (0+6.7,0) -- (0+6.7,0.8);
\draw (-0.4+6.7,0.4) -- (0+6.7,0);
\draw (-0.4+6.7,0.4) -- (0+6.7,0.8);
\draw (1.6+6.7,0) -- (0+6.7,0);
\draw (0.8+6.7,0) -- (0+6.7,0.8);
\draw (0.8+6.7,0.8) -- (0+6.7,0);
\draw (1.6+6.7,0.8) -- (0+6.7,0.8);
\draw (0.8+6.7,0.8) -- (0.4+6.7,1.2);
\draw (1.6+6.7,0.8) -- (1.2+6.7,0.4);
\draw (1.6+6.7,0.8) -- (2+6.7,0.4);
\draw (1.2+6.7,0.4) -- (2+6.7,0.4);
\draw (1.6+6.7,0) -- (1.2+6.7,0.4);
\draw (1.6+6.7,0) -- (2+6.7,0.4);
\draw (0.8+6.7,0) -- (0.4+6.7,-0.4);
\draw[dashed] (-0.4+6.7,0.4) -- (1.2+6.7,0.4);

\node at (0.8+6.7,0.8) [circle,draw=green!50,fill=green!20] {};
\node at (0.8+6.7,0.8) {3};
\node at (0+6.7,0.8) [circle,draw=yellow!50,fill=yellow!20] {};
\node at (0+6.7,0.8) {4};
\node at (0+6.7,0) [circle,draw=black!50,fill=black!20] {};
\node at (0+6.7,0) {5};
\node at (-0.4+6.7,0.4) [circle,draw=blue!50,fill=blue!20] {};
\node at (-0.4+6.7,0.4) {1};
\node at (0.8+6.7,0) [circle,draw=red!50,fill=red!20] {};
\node at (0.8+6.7,0) {2};
\node at (1.6+6.7,0) [circle,draw=black,fill=black] {};
\node at (1.6+6.7,0.8) [circle,draw=red!50,fill=red!20] {};
\node at (1.6+6.7,0.8) {2};
\node at (2+6.7,0.4) [circle,draw=blue!50,fill=blue!20] {};
\node at (2+6.7,0.4) {1};
\node at (0.4+6.7,-0.4) [circle,draw=blue!50,fill=blue!20] {};
\node at (0.4+6.7,-0.4) {1};
\node at (0.4+6.7,1.2) [circle,draw=blue!50,fill=blue!20] {};
\node at (0.4+6.7,1.2) {1};
\node at (1.2+6.7,0.4) [circle,draw=black,fill=black] {};
\node at (1.2+6.5,-0.9){2.b};

\draw (0-1.3,0-2.5) -- (0-1.3,0.8-2.5);
\draw (-0.4-1.3,0.4-2.5) -- (0-1.3,0-2.5);
\draw (-0.4-1.3,0.4-2.5) -- (0-1.3,0.8-2.5);
\draw (1.6-1.3,0-2.5) -- (0-1.3,0-2.5);
\draw (0.8-1.3,0-2.5) -- (0+-1.3,0.8-2.5);
\draw (0.8-1.3,0.8-2.5) -- (0-1.3,0-2.5);
\draw (1.6-1.3,0.8-2.5) -- (0-1.3,0.8-2.5);
\draw (0.8-1.3,0.8-2.5) -- (0.4-1.3,1.2-2.5);
\draw[dashed] (-0.4-1.3,0.4-2.5) -- (1.2-1.3,0.4-2.5);
\draw (1.6-1.3,0.8-2.5) -- (1.2-1.3,0.4-2.5);
\draw (1.6-1.3,0.8-2.5) -- (1.2-1.3,1.2-2.5);

\node at (0.8-1.3,0.8-2.5) [circle,draw=green!50,fill=green!20] {};
\node at (0.8-1.3,0.8-2.5) {3};
\node at (0-1.3,0.8-2.5) [circle,draw=yellow!50,fill=yellow!20] {};
\node at (0-1.3,0.8-2.5) {4};
\node at (0-1.3,0-2.5) [circle,draw=black!50,fill=black!20] {};
\node at (0-1.3,0-2.5) {5};
\node at (0.8-1.3,0-2.5) [circle,draw=red!50,fill=red!20] {};
\node at (0.8-1.3,0-2.5) {2};
\node at (1.6-1.3,0.8-2.5) [circle,draw=red!50,fill=red!20] {};
\node at (1.6-1.3,0.8-2.5) {2};
\node at (-0.4-1.3,0.4-2.5) [circle,draw=blue!50,fill=blue!20] {};
\node at (-0.4-1.3,0.4-2.5) {1};
\node at (1.2-1.3,1.2-2.5) [circle,draw=blue!50,fill=blue!20] {};
\node at (1.2-1.3,1.2-2.5) {1};
\node at (0.4-1.3,1.2-2.5) [circle,draw=blue!50,fill=blue!20] {};
\node at (0.4-1.3,1.2-2.5) {1};
\node at (1.6-1.3,0-2.5) [circle,draw=blue!50,fill=blue!20] {};
\node at (1.6-1.3,0-2.5) {1};
\node at (1.2-1.3,0.4-2.5) [circle,draw=black,fill=black] {};
\node at (0.4-1.3,-0.9-2.5){2.c};

\draw (0+1.9,0-2.5) -- (0+1.9,0.8-2.5);
\draw (1.6+1.9,0-2.5) -- (-0.8+1.9,0-2.5);
\draw (0.8+1.9,0-2.5) -- (0+1.9,0.8-2.5);
\draw (0.8+1.9,0.8-2.5) -- (0+1.9,0-2.5);
\draw (1.6+1.9,0.8-2.5) -- (-0.8+1.9,0.8-2.5);
\draw (0.8+1.9,0.8-2.5) -- (0.4+1.9,1.2-2.5);
\draw (1.6+1.9,0.8-2.5) -- (1.6+1.9,0-2.5);
\node at (0.8+1.9,0.8-2.5) [circle,draw=green!50,fill=green!20] {};
\node at (0.8+1.9,0.8-2.5) {3};
\node at (0+1.9,0.8-2.5) [circle,draw=yellow!50,fill=yellow!20] {};
\node at (0+1.9,0.8-2.5) {4};
\node at (0+1.9,0-2.5) [circle,draw=black!50,fill=black!20] {};
\node at (0+1.9,0-2.5) {5};
\node at (0.8+1.9,0-2.5) [circle,draw=red!50,fill=red!20] {};
\node at (0.8+1.9,0-2.5) {2};
\node at (-0.8+1.9,0-2.5) [circle,draw=blue!50,fill=blue!20] {};
\node at (-0.8+1.9,0-2.5) {1};
\node at (-0.8+1.9,0.8-2.5) [circle,draw=blue!50,fill=blue!20] {};
\node at (-0.8+1.9,0.8-2.5) {1};
\node at (1.6+1.9,0-2.5) [circle,draw=blue!50,fill=blue!20] {};
\node at (1.6+1.9,0-2.5) {1};
\node at (1.6+1.9,0.8-2.5) [circle,draw=red!50,fill=red!20] {};
\node at (1.6+1.9,0.8-2.5) {2};
\node at (0.4+1.9,1.2-2.5) [circle,draw=blue!50,fill=blue!20] {};
\node at (0.4+1.9,1.2-2.5) {1};
\node at (0.4+1.9,-0.9-2.5){3.a};

\draw (0+4,-0.6-2.5) -- (0+4,0.8-2.5);
\draw (0.8+4,0-2.5) -- (0+4,0-2.5);
\draw (0.8+4,0-2.5) -- (0+4,0.8-2.5);
\draw (0.8+4,0.8-2.5) -- (0+4,0-2.5);
\draw (0.8+4,0.8-2.5) -- (0+4,0.8-2.5);
\draw (0.8+4,0.8-2.5) -- (1.6+4,0.4-2.5);
\draw (1.6+4,0.4-2.5) -- (1.2+4,0.8-2.5);
\draw (1.6+4,0.4-2.5) -- (2+4,0.8-2.5);
\draw (1.2+4,0.8-2.5) -- (2+4,0.8-2.5);
\draw (1.6+4,1.2-2.5) -- (1.2+4,0.8-2.5);
\draw (1.6+4,1.2-2.5) -- (2+4,0.8-2.5);
\draw (0+4,0.8-2.5) -- (1.6+4,1.2-2.5);
\draw (1.6+4,-0.4-2.5) -- (1.6+4,-0.8-2.5);
\draw (1.6+4,-0.4-2.5) -- (1.6+4,0-2.5);
\draw (1.6+4,-0.4-2.5) -- (2+4,-0.4-2.5);
\draw (0.8+4,0.8-2.5) -- (0.8+4,0.4-2.5);
\draw (0.8+4,0-2.5) -- (1.6+4,-0.4-2.5);
\draw (0.8+4,0-2.5) -- (0.5+4,-0.4-2.5);
\draw[dashed] (0+4,-0.6-2.5) -- (1.6+4,-0.8-2.5);
\draw[dashed] (1.6+4,1.2-2.5) .. controls (2.5+4,1-2.5) .. (2+4,-0.4-2.5);
\draw[dashed] (0.8+4,0.4-2.5) -- (1.6+4,0-2.5);

\node at (0.8+4,0-2.5) [circle,draw=green!50,fill=green!20] {};
\node at (0.8+4,0-2.5) {3};
\node at (0+4,0.8-2.5) [circle,draw=yellow!50,fill=yellow!20] {};
\node at (0+4,0.8-2.5) {4};
\node at (0+4,0-2.5) [circle,draw=black!50,fill=black!20] {};
\node at (0+4,0-2.5) {5};
\node at (0.8+4,0.8-2.5) [circle,draw=red!50,fill=red!20] {};
\node at (0.8+4,0.8-2.5) {2};
\node at (1.6+4,0-2.5) [circle,draw=blue!50,fill=blue!20] {};
\node at (1.6+4,0-2.5) {1};
\node at (1.6+4,1.2-2.5) [circle,draw=blue!50,fill=blue!20] {};
\node at (1.6+4,1.2-2.5) {1};
\node at (1.6+4,0.4-2.5) [circle,draw=blue!50,fill=blue!20] {};
\node at (1.6+4,0.4-2.5) {1};
\node at (0+4,-0.6-2.5) [circle,draw=blue!50,fill=blue!20] {};
\node at (0+4,-0.6-2.5) {1};
\node at (1.6+4,-0.4-2.5) [circle,draw=red!50,fill=red!20] {};
\node at (1.6+4,-0.4-2.5) {2};
\node at (0.5+4,-0.4-2.5) [circle,draw=blue!50,fill=blue!20] {};
\node at (0.5+4,-0.4-2.5) {1};
\node at (0.8+4,0.4-2.5) [circle,draw=black,fill=black] {};
\node at (2+4,0.8-2.5) [circle,draw=black,fill=black] {};
\node at (1.2+4,0.8-2.5) [circle,draw=black,fill=black] {};
\node at (1.6+4,-0.8-2.5) [circle,draw=black,fill=black] {};
\node at (2+4,-0.4-2.5) [circle,draw=black,fill=black] {};
\node at (0.8+4,-0.9-2.5){3.b};

\draw (0+6.6,-0.6-2.5) -- (0+6.6,1.4-2.5);
\draw (1.6+6.6,0-2.5) -- (0+6.6,0-2.5);
\draw (0.8+6.6,0-2.5) -- (0+6.6,0.8-2.5);
\draw (0.8+6.6,0.8-2.5) -- (0+6.6,0-2.5);
\draw (1.6+6.6,0.8-2.5) -- (0+6.6,0.8-2.5);
\draw (0.8+6.6,0.8-2.5) -- (0.4+6.6,1.2-2.5);
\draw (1.6+6.6,0.8-2.5) -- (1.2+6.6,0.4-2.5);
\draw (1.6+6.6,0.8-2.5) -- (2+6.6,0.4-2.5);
\draw (1.6+6.6,0-2.5) -- (1.2+6.6,0.4-2.5);
\draw (1.6+6.6,0-2.5) -- (2+6.6,0.4-2.5);
\draw (1.2+6.6,0.4-2.5) -- (2+6.6,0.4-2.5);
\draw (1.6+6.6,0.8-2.5) -- (1.2+6.6,1.2-2.5);
\node at (0.8+6.6,0.8-2.5) [circle,draw=green!50,fill=green!20] {};
\node at (0.8+6.6,0.8-2.5) {3};
\node at (0+6.6,0.8-2.5) [circle,draw=yellow!50,fill=yellow!20] {};
\node at (0+6.6,0.8-2.5) {4};
\node at (0+6.6,0-2.5) [circle,draw=black!50,fill=black!20] {};
\node at (0+6.6,0-2.5) {5};
\node at (0.8+6.6,0-2.5) [circle,draw=red!50,fill=red!20] {};
\node at (0.8+6.6,0-2.5) {2};
\node at (1.6+6.6,0.8-2.5) [circle,draw=red!50,fill=red!20] {};
\node at (1.6+6.6,0.8-2.5) {2};
\node at (1.6+6.6,0-2.5) [circle,draw=blue!50,fill=blue!20] {};
\node at (1.6+6.6,0-2.5) {1};
\node at (0+6.6,-0.6-2.5) [circle,draw=blue!50,fill=blue!20] {};
\node at (0+6.6,-0.6-2.5) {1};
\node at (0+6.6,1.4-2.5) [circle,draw=blue!50,fill=blue!20] {};
\node at (0+6.6,1.4-2.5) {1};
\node at (0.4+6.6,1.2-2.5) [circle,draw=blue!50,fill=blue!20] {};
\node at (0.4+6.6,1.2-2.5) {1};
\node at (1.2+6.6,1.2-2.5) [circle,draw=blue!50,fill=blue!20] {};
\node at (1.2+6.6,1.2-2.5) {1};
\node at (1.2+6.6,0.4-2.5) [circle,draw=black,fill=black] {};
\node at (2+6.6,0.4-2.5) [circle,draw=black,fill=black] {};
\node at (0.8+6.6,-0.9-2.5){3.c};

\draw (0-1.7,-0.6-5) -- (0-1.7,1.4-5);
\draw (1.6-1.7,0-5) -- (0-1.7,0-5);
\draw (0.8-1.7,0-5) -- (0-1.7,0.8-5);
\draw (0.8-1.7,0.8-5) -- (0-1.7,0-5);
\draw (2.2-1.7,0.8-5) -- (0-1.7,0.8-5);
\draw (0.8-1.7,0.8-5) -- (0.4-1.7,1.2-5);
\draw (1.6-1.7,0.8-5) -- (2.2-1.7,1.3-5);
\draw (1.6-1.7,0.8-5) -- (2.2-1.7,0.3-5);
\node at (0.8-1.7,0.8-5) [circle,draw=green!50,fill=green!20] {};
\node at (0.8-1.7,0.8-5) {3};
\node at (0-1.7,0.8-5) [circle,draw=yellow!50,fill=yellow!20] {};
\node at (0-1.7,0.8-5) {4};
\node at (0-1.7,0-5) [circle,draw=black!50,fill=black!20] {};
\node at (0-1.7,0-5) {5};
\node at (0.8-1.7,0-5) [circle,draw=red!50,fill=red!20] {};
\node at (0.8-1.7,0-5) {2};
\node at (1.6-1.7,0.8-5) [circle,draw=red!50,fill=red!20] {};
\node at (1.6-1.7,0.8-5) {2};
\node at (1.6-1.7,0-5) [circle,draw=blue!50,fill=blue!20] {};
\node at (1.6-1.7,0-5) {1};
\node at (0-1.7,-0.6-5) [circle,draw=blue!50,fill=blue!20] {};
\node at (0-1.7,-0.6-5) {1};
\node at (0-1.7,1.4-5) [circle,draw=blue!50,fill=blue!20] {};
\node at (0-1.7,1.4-5) {1};
\node at (0.4-1.7,1.2-5) [circle,draw=blue!50,fill=blue!20] {};
\node at (0.4-1.7,1.2-5) {1};
\node at (2.2-1.7,0.8-5) [circle,draw=blue!50,fill=blue!20] {};
\node at (2.2-1.7,0.8-5) {1};
\node at (2.2-1.7,1.3-5) [circle,draw=black,fill=black] {};
\node at (2.2-1.7,0.3-5) [circle,draw=black,fill=black] {};
\node at (1-2.3,-6.4){3.d};

\draw (-0.8+2,0-5) -- (-0.4+2,0.8-5);
\draw (-0.8+2,0-5) -- (-1+2,0-0.6-5);
\draw (-0.8+2,0-5) -- (-0.6+2,0.-0.6-5);
\draw (-0.4+2,0.8-5) -- (-0.8+2,1.1-5);
\draw (0+2,0-5) -- (0+2,-0.6-5);
\draw (-0.8+2,0-5) -- (0.8+2,0-5);
\draw (-0.4+2,0.8-5) -- (0+2,0-5);
\draw (-0.4+2,0.8-5) -- (1.2+2,0.8-5);
\draw (0.8+2,0-5) -- (1.2+2,0.8-5);
\draw (-0.6+2,-0.6-5) .. controls (-0+2,-1-5) and (1.2+2,-1-5) .. (1.2+2,0.8-5);
\node at (1.2+2,0.8-5) [circle,draw=blue!50,fill=blue!20] {};
\node at (1.2+2,0.8-5){1};
\node at (-1+2,0.-0.6-5) [circle,draw=blue!50,fill=blue!20] {};
\node at (-1+2,0.-0.6-5){1};
\node at (-0.8+2,1.1-5) [circle,draw=blue!50,fill=blue!20] {};
\node at (-0.8+2,1.1-5){1};
\node at (0+2,-0.6-5) [circle,draw=blue!50,fill=blue!20] {};
\node at (0+2,-0.6-5){1};
\node at (-0.4+2,0.8-5) [circle,draw=yellow!50,fill=yellow!20] {};
\node at (-0.4+2,0.8-5) {4};
\node at (0+2,0-5) [circle,draw=black!50,fill=black!20] {};
\node at (0+2,0-5) {5};
\node at (-0.6+2,-0.6-5)  [circle,draw=red!50,fill=red!20] {};
\node at (-0.6+2,-0.6-5) {2};
\node at (0.8+2,0-5) [circle,draw=red!50,fill=red!20] {};
\node at (0.8+2,0-5) {2};
\node at (0.4+2,0.8-5) [circle,draw=red!50,fill=red!20] {};
\node at (0.4+2,0.8-5) {2};
\node at (-0.8+2,0-5) [circle,draw=green!50,fill=green!20] {};
\node at (-0.8+2,0-5) {3};

\node at (-0.4+2,1.1-5) {$u_1$};
\node at (0.2+2,0.3-5) {$u_2$};
\node at (-1.2+2,0-5) {$u_3$};
\node at (2.1,-5-1.4) {4.a};

\draw (-0.8+4.6,0-5) -- (-0.4+4.6,0.8-5);
\draw (-0.8+4.6,0-5) -- (-0.8+4.6,0-1-5);
\draw (-0.8+4.6,-1-5) -- (0.8+4.6,0-1-5);
\draw (-0.4+4.6,0.8-5) -- (1.2+4.6,0.8-5);
\draw (0+4.6,0-5) -- (0+4.6,-0.6-5);
\draw (-0.8+4.6,0-5) -- (0.8+4.6,0-5);
\draw (-0.4+4.6,0.8-5) -- (0+4.6,0-5);
\draw (0.8+4.6,0-5) -- (1.2+4.6,0.8-5);
\draw (0.8+4.6,0-5) -- (0.8+4.6,-1-5);
\draw (0.8+4.6,0-5) -- (0.4+4.6,0.4-5);
\draw (0.8+4.6,-1-5) -- (1.2+4.6,0.8-5);
\node at (0+4.6,0-5) [circle,draw=black!50,fill=black!20] {};
\node at (0+4.6,0-5) {5};
\node at (-0.8+4.6,0-5) [circle,draw=green!50,fill=green!20] {};
\node at (-0.8+4.6,0-5) {3};
\node at (0.8+4.6,0-5) [circle,draw=yellow!50,fill=yellow!20] {};
\node at (0.8+4.6,0-5) {4};
\node at (0.8+4.6,-1-5) [circle,draw=green!50,fill=green!20] {};
\node at (0.8+4.6,-1-5) {3};
\node at (-0.4+4.6,0.8-5) [circle,draw=red!50,fill=red!20] {};
\node at (-0.4+4.6,0.8-5) {2};
\node at (1.2+4.6,0.8-5) [circle,draw=red!50,fill=red!20] {};
\node at (1.2+4.6,0.8-5){2};
\node at (0.4+4.6,0.8-5) [circle,draw=blue!50,fill=blue!20] {};
\node at (0.4+4.6,0.8-5) {1};
\node at (-0.8+4.6,0-1-5) [circle,draw=blue!50,fill=blue!20] {};
\node at (-0.8+4.6,0-1-5) {1};
\node at (0+4.6,-0.6-5) [circle,draw=blue!50,fill=blue!20] {};
\node at (0+4.6,-0.6-5) {1};
\node at (0.4+4.6,0.4-5) [circle,draw=blue!50,fill=blue!20] {};
\node at (0.4+4.6,0.4-5) {1};
\node at (0.2+4.6,-5-1.4) {4.b};

\draw (-0.8+7.2,0-5) -- (-0.4+7.2,0.8-5);
\draw (-0.8+7.2,0-5) -- (-0.8+7.2,0-1-5);
\draw (-0.8+7.2,-1-5) -- (0.8+7.2,0-1-5);
\draw (-0.4+7.2,0.8-5) -- (1.2+7.2,0.8-5);
\draw (0+7.2,0-5) -- (0+7.2,-0.6-5);
\draw (-0.8+7.2,0-5) -- (0.8+7.2,0-5);
\draw (-0.4+7.2,0.8-5) -- (0+7.2,0-5);
\draw (0.8+7.2,0-5) -- (1.2+7.2,0.8-5);
\draw (0.8+7.2,0-5) -- (0.8+7.2,-1-5);
\draw (-0.8+7.2,0-5) -- (-1.1+7.2,-0.6-5);
\draw (-0.4+7.2,0.8-5) -- (-0.8+7.2,1.2-5);

\node at (-0.4+7.2,0.8-5) [circle,draw=yellow!50,fill=yellow!20] {};
\node at (-0.4+7.2,0.8-5) {4};
\node at (0+7.2,0-5) [circle,draw=black!50,fill=black!20] {};
\node at (0+7.2,0-5) {5};
\node at (0.8+7.2,0-5) [circle,draw=red!50,fill=red!20] {};
\node at (0.8+7.2,0-5) {2};
\node at (0.4+7.2,0.8-5) [circle,draw=red!50,fill=red!20] {};
\node at (0.4+7.2,0.8-5) {2};
\node at (-0.8+7.2,-1-5) [circle,draw=red!50,fill=red!20] {};
\node at (-0.8+7.2,-1-5) {2};
\node at (-0.8+7.2,0-5) [circle,draw=green!50,fill=green!20] {};
\node at (-0.8+7.2,0-5) {3};
\node at (1.2+7.2,0.8-5) [circle,draw=blue!50,fill=blue!20] {};
\node at (1.2+7.2,0.8-5) {1};
\node at (0.8+7.2,-1-5) [circle,draw=blue!50,fill=blue!20] {};
\node at (0.8+7.2,-1-5) {1};
\node at (0+7.2,-0.6-5) [circle,draw=blue!50,fill=blue!20] {};
\node at (0+7.2,-0.6-5) {1};
\node at (-1.1+7.2,-0.6-5) [circle,draw=blue!50,fill=blue!20] {};
\node at (-1.1+7.2,-0.6-5) {1};
\node at (-0.8+7.2,1.2-5) [circle,draw=blue!50,fill=blue!20] {};
\node at (-0.8+7.2,1.2-5) {1};
\node at (0.2+7.2,-5-1.4) {4.c};

\draw (-0.8+0.3,0-7.3) -- (-0.4+0.3,0.8-7.3);
\draw (-0.8+0.3,0-7.3) -- (-0.6+0.3,-0.5-7.3);
\draw (-0.4+0.3,0.8-7.3) -- (1.2+0.3,0.8-7.3);
\draw (0+0.3,0-7.3) -- (0+0.3,-0.6-7.3);
\draw (0.8+0.3,0-7.3) -- (0.8+0.3,-0.6-7.3);
\draw (-0.8+0.3,0-7.3) -- (0.8+0.3,0-7.3);
\draw (-0.4+0.3,0.8-7.3) -- (0+0.3,0-7.3);
\draw (0.8+0.3,0-7.3) -- (1.2+0.3,0.8-7.3);
\draw (-0.8+0.3,0-7.3) -- (-1.1+0.3,-0.6-7.3);
\draw (-0.4+0.3,0.8-7.3) -- (-0.8+0.3,1.2-7.3);
\draw (0.4+0.3,0.8-7.3) -- (0+0.3,1.2-7.3);
\draw (0.4+0.3,0.8-7.3) -- (0.8+0.3,1.2-7.3);
\draw (-1+0.3,-1-7.3) -- (-0.6+0.3,-0.5-7.3);
\draw (-0.6+0.3,-1-7.3) -- (-0.6+0.3,-0.5-7.3);
\draw (-0.2+0.3,-1-7.3) -- (-0.6+0.3,-0.5-7.3);
\draw[dashed] (-1+0.3,-1-7.3) .. controls (-1.5+0.3,-1-7.3) .. (-0.8+0.3,1.2-7.3);
\draw[dashed] (-0.2+0.3,-1-7.3) .. controls (0.4+0.3,-1-7.3) .. (0+0.3,-0.6-7.3);
\draw[dashed] (0.8+0.3,-0.6-7.3) .. controls (1.2+0.3,-0.4-7.3) ..  (0.8+0.3,1.2-7.3);

\node at (-0.4+0.3,0.8-7.3) [circle,draw=yellow!50,fill=yellow!20] {};
\node at (-0.4+0.3,0.8-7.3) {4};
\node at (0+0.3,0-7.3) [circle,draw=black!50,fill=black!20] {};
\node at (0+0.3,0-7.3) {5};
\node at (0.8+0.3,0-7.3) [circle,draw=red!50,fill=red!20] {};
\node at (0.8+0.3,0-7.3) {2};
\node at (0.4+0.3,0.8-7.3) [circle,draw=red!50,fill=red!20] {};
\node at (0.4+0.3,0.8-7.3) {2};
\node at (-0.8+0.3,0-7.3) [circle,draw=green!50,fill=green!20] {};
\node at (-0.8+0.3,0-7.3) {3};
\node at (-0.6+0.3,-0.5-7.3) [circle,draw=red!50,fill=red!20] {};
\node at (-0.6+0.3,-0.5-7.3) {2};
\node at (0+0.3,1.2-7.3) [circle,draw=blue!50,fill=blue!20] {};
\node at (0+0.3,1.2-7.3) {1};
\node at (-0.8+0.3,1.2-7.3) [circle,draw=blue!50,fill=blue!20] {};
\node at (-0.8+0.3,1.2-7.3){1};
\node at (-1.1+0.3,-0.6-7.3) [circle,draw=blue!50,fill=blue!20] {};
\node at (-1.1+0.3,-0.6-7.3) {1};
\node at (0+0.3,-0.6-7.3) [circle,draw=blue!50,fill=blue!20] {};
\node at (0+0.3,-0.6-7.3) {1};
\node at (0.8+0.3,-0.6-7.3) [circle,draw=blue!50,fill=blue!20] {};
\node at (0.8+0.3,-0.6-7.3) {1};
\node at (-0.6+0.3,-1-7.3) [circle,draw=blue!50,fill=blue!20] {};
\node at (-0.6+0.3,-1-7.3) {1};
\node at (1.2+0.3,0.8-7.3) [circle,draw=black,fill=black] {};
\node at (-0.2+0.3,-1-7.3) [circle,draw=black,fill=black] {};
\node at (-1+0.3,-1-7.3) [circle,draw=black,fill=black] {};
\node at (0.8+0.3,1.2-7.3) [circle,draw=black,fill=black] {};
\node at (0.2,-0.5-8.2) {4.d};

\draw (-1.2+3.7,0-8.2) -- (1.2+3.7,0-8.2);
\draw (0.4+3.7,0-8.2) -- (0+3.7,0.8-8.2);
\draw (-0.4+3.7,0-8.2) -- (0+3.7,0.8-8.2);
\draw (-0.4+3.7,0-8.2) -- (-1+3.7,0.4-8.2);
\draw (-1+3.7,0.4-8.2) -- (-1.8+3.7,0-8.2);
\draw (-1+3.7,0.4-8.2) -- (-1.8+3.7,0.4-8.2);
\draw (-1+3.7,0.4-8.2) -- (-1.8+3.7,0.8-8.2);
\draw (0+3.7,0.8-8.2) -- (0+3.7,1.2-8.2);
\draw (0+3.7,1.2-8.2) -- (-0.8+3.7,0.8-8.2);
\draw (0+3.7,1.2-8.2) -- (-0.8+3.7,1.2-8.2);
\draw (0+3.7,1.2-8.2) -- (-0.8+3.7,1.6-8.2);
\draw (0.4+3.7,0-8.2) -- (1+3.7,0.4-8.2);
\draw (1+3.7,0.4-8.2) -- (1.8+3.7,0-8.2);
\draw (1+3.7,0.4-8.2) -- (1.8+3.7,0.4-8.2);
\draw (1+3.7,0.4-8.2) -- (1.8+3.7,0.8-8.2);
\draw (0+3.7,0.8-8.2) -- (0.6+3.7,1.2-8.2);
\draw[dashed] (-1.8+3.7,0-8.2) .. controls (0+3.7,-0.4-8.2) .. (1.8+3.7,0-8.2);
\draw[dashed] (-1.8+3.7,0.8-8.2) -- (-0.8+3.7,0.8-8.2);
\draw[dashed] (-0.8+3.7,1.6-8.2) .. controls (1.8+3.7,1.6-8.2) .. (1.8+3.7,0.8-8.2);
\node at (-0.4+3.7,0-8.2) [circle,draw=yellow!50,fill=yellow!20] {};
\node at (-0.4+3.7,0-8.2) {4};
\node at (0.4+3.7,0-8.2) [circle,draw=black!50,fill=black!20] {};
\node at (0.4+3.7,0-8.2) {5};
\node at (0+3.7,0.8-8.2) [circle,draw=green!50,fill=green!20] {};
\node at (0+3.7,0.8-8.2) {3};
\node at (0+3.7,1.2-8.2) [circle,draw=red!50,fill=red!20] {};
\node at (0+3.7,1.2-8.2) {2};
\node at (-1+3.7,0.4-8.2)[circle,draw=red!50,fill=red!20] {};
\node at (-1+3.7,0.4-8.2) {2};
\node at (1+3.7,0.4-8.2) [circle,draw=red!50,fill=red!20] {};
\node at (1+3.7,0.4-8.2) {2};
\node at (0.6+3.7,1.2-8.2) [circle,draw=blue!50,fill=blue!20] {};
\node at (0.6+3.7,1.2-8.2) {1};
\node at (1+3.7,0-8.2) [circle,draw=blue!50,fill=blue!20] {};
\node at (1+3.7,0-8.2) {1};
\node at (-1+3.7,0-8.2) [circle,draw=blue!50,fill=blue!20] {};
\node at (-1+3.7,0-8.2) {1};
\node at (-0.8+3.7,1.2-8.2) [circle,draw=blue!50,fill=blue!20] {};
\node at (-0.8+3.7,1.2-8.2) {1};
\node at (1.8+3.7,0.4-8.2) [circle,draw=blue!50,fill=blue!20] {};
\node at (1.8+3.7,0.4-8.2) {1};
\node at (-1.8+3.7,0.4-8.2) [circle,draw=blue!50,fill=blue!20] {};
\node at (-1.8+3.7,0.4-8.2) {1};
\node at (-1.8+3.7,0.8-8.2) [circle,draw=black,fill=black] {};
\node at (-1.8+3.7,0-8.2)[circle,draw=black,fill=black] {};
\node at (1.8+3.7,0.8-8.2) [circle,draw=black,fill=black] {};
\node at (1.8+3.7,0-8.2)[circle,draw=black,fill=black] {};
\node at (-0.8+3.7,0.8-8.2) [circle,draw=black,fill=black] {};
\node at (-0.8+3.7,1.6-8.2)[circle,draw=black,fill=black] {};
\node at (0+3.7,-0.5-8.2) {4.e};
\end{tikzpicture}
\end{center}
\caption{Possible configurations when $G$ an induced $C_3$.}
\label{figm32} 
\end{figure}
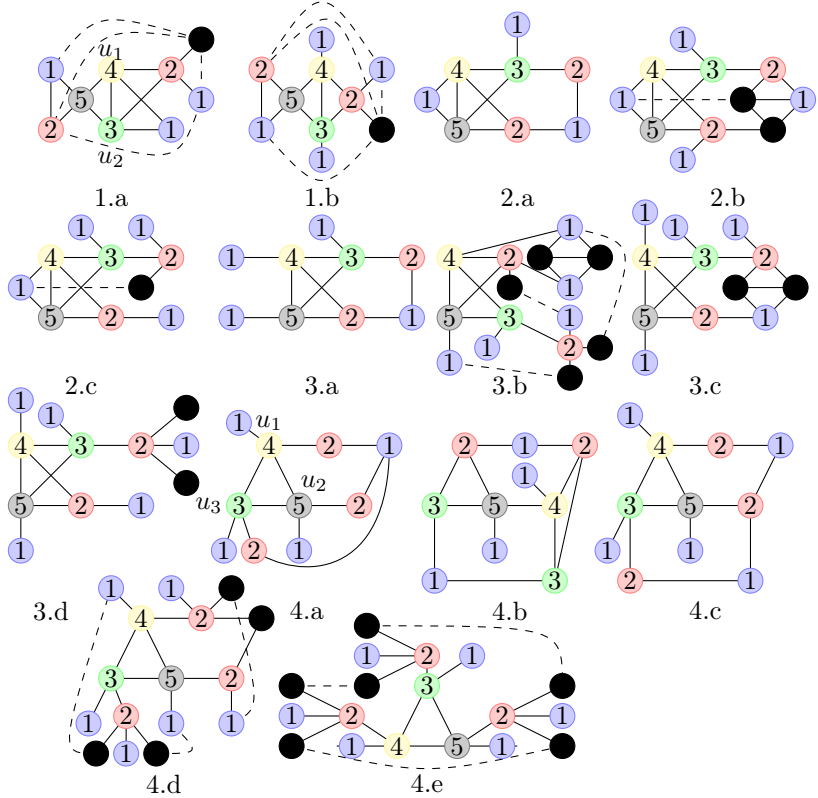
\begin{proof}
Depending on the different cases that could happen, a reference to the Grundy partial 5-coloring of $G$ in Figure~\ref{figm32} will be given.
Let $M_i$, $i=2$ or 3, be the graph of order $2+i$ containing two adjacent vertices $u_1$ and $u_2$ which have exactly $i$ common neighbors, $\{v_1,\ldots,v_i \}$, that form an independent set. 
Let $D_1$ be the set of vertices at distance 1 from an induced $M_i$ in $G-M_i$, for $2\le i\le3$.
\begin{description}
\item[Case 1:] Firstly, assume that $G$ contains an induced $M_3$ and a vertex of $M_3$ has its two neighbors in $D_1$ adjacent (Figure 6.1.a).
Secondly, assume that $G$ contains an induced $M_2$ and a vertex of $M_2$ has its two neighbors in $D_1$ adjacent (Figure 6.1.b). Note that these Grundy partial 5-colorings use the fact that $G$ cannot contain a $K_4$ by Lemma \ref{k4}.
\item[Case 2:] Assume that $G$ contains an induced $M_3$ excluding the previous configuration. There are three cases: $u_1$, $v_2$ and $v_3$ are in an induced $C_5$ (Figure 6.2.a), $u_1$, $v_2$ and $v_3$ are in an induced $C_6$ and not in an induced $C_5$ (Figure 6.2.b) and $u_1$, $v_2$ and $v_3$ are neither in an induced $C_5$ nor $C_6$ (Figure 6.2.c). 
\item[Case 3:] Suppose that $G$ contains an induced $M_2$ excluding the previous configurations. Firstly, we suppose that $u_1$, $v_1$ and $v_2$ are in an induced $C_5$ (Figure 6.3.a). Secondly, we suppose that $u_1$, $v_1$ are in an induced $C_5$ excluding the previous case (Figure 6.3.b). Thirdly, we suppose that $u_1$, $v_1$ and $v_2$ are in an induced $C_6$ and not in an induced $C_5$ (Figure 6.3.c) and finally neither in an induced $C_5$ nor $C_6$ (Figure 6.3.d).

Suppose that $G$ contains a 3-cycle $C$ and no induced $M_2$.
Let $u_1$, $u_2$ and $u_3$ be the vertices of $C$. Let $w_1$ and $w_2$ be the neighbors of $u_1$ outside $C$, let $w'_1$ and $w'_2$ be the neighbors of $u_2$ outside $C$ and let $w''_1$ and $w''_2$ be the neighbors of $u_3$ outside $C$.
\item[Case 4:] Firstly, suppose that $u_1$, $u_2$, $w_1$ and $w'_1$ are in a 5-cycle and a neighbor of $u_1$, say $w_1$, has a common neighbor with $w'_1$ (Figure 6.4.a).
Secondly, excluding the previous configuration, suppose that $u_1$, $u_2$, $w_1$ and $w'_1$ are in a 5-cycle; $w''_1$,  $v_1$, $u_1$ and $w_1$ are in another 5-cycle and $w_1$ is in a triangle (Figure 6.4.b). 
We suppose that $w_1$ is not in a triangle (Figure 6.4.c).
Thirdly, excluding the previous configurations, we obtain a Grundy partial 5-coloring if two vertices of $C$ are in a 5-cycle (Figure 6.4.d).
Fourthly, we suppose that two vertices of $C$ cannot be in a 5-cycle (Figure 6.4.e).
\end{description}
\end{proof}
In the following two lemmas, we consider a graph $G$ of girth $g=5$ and possibly containing an induced Petersen graph. Let $u_1$, $u_2$, $u_3$, $u_4$ and $u_5$ be the vertices in an induced $C_5$ (or in the the outer cycle of a Petersen graph, if any).  Let $v_1$, $v'_1$, $v_2$, $v'_2$, $v_3$, $v'_3$, $v_4$, $v'_4$,  $v_5$ and $v'_5$ be the remaining neighbors of respectively $u_1$, $u_2$, $u_3$, $u_4$ and $u_5$ (all different as $g=5$).
\begin{lem}
Let $G$ be a 4-regular graph with girth $g=5$. If $G$ contains the Petersen graph as induced subgraph then $\Gamma(G)=5$.
\label{lpet}
\end{lem}
\begin{proof}
Suppose that $v_1$, $v_2$, $v_3$, $v_4$ and $v_5$ form an induced $C_5$ (the inner cycle of the Petersen graph). Let $u'_2$ and $u'_5$ be the remaining neighbors of respectively $v_2$ and $v_5$.
Observe that $v'_1$ can be adjacent with no more than three vertices among  $v'_3$, $v'_4$, $u'_2$ and $u'_5$.
Firstly, suppose that $v'_1$ is not adjacent with $v'_3$ (or $v'_4$, without loss of generality since the configuration is symmetric). The left part of Figure~\ref{figm51} illustrates a Grundy partial 5-coloring of the graph $G$.
Secondly, assume that $v'_1$ is not adjacent with $u'_5$ (or $u'_2$, without loss of generality). The right part of Figure~\ref{figm51} illustrates a Grundy partial 5-coloring of the graph $G$.
\end{proof}
\begin{figure}[t]
\begin{center}
\begin{tikzpicture}[scale=1.1]
\draw (0,0) -- (3*2/3,0);
\draw (0,0) -- (-1*2/3,2*2/3);
\draw (3*2/3,0) -- (4*2/3,2*2/3);
\draw (-1*2/3,2*2/3) -- (1.5*2/3,4.4*2/3);
\draw (4*2/3, 2*2/3) -- (1.5*2/3,4.4*2/3);
\draw (1.5*2/3,3.2*2/3) -- (1.5*2/3,4.4*2/3);
\draw (-1*2/3,2*2/3) -- (0.2*2/3, 2.2*2/3);
\draw (4*2/3, 2*2/3) -- (2.8*2/3, 2.2*2/3);
\draw (0,0) -- (0.7*2/3,0.7*2/3);
\draw (3*2/3,0) -- (2.3*2/3,0.7*2/3);
\draw (0.7*2/3,0.7*2/3) -- (2.8*2/3,2.2*2/3);
\draw (2.3*2/3,0.7*2/3) -- (0.2*2/3,2.2*2/3);
\draw (0.7*2/3,0.7*2/3) -- (1.5*2/3,3.2*2/3);
\draw (2.3*2/3,0.7*2/3) -- (1.5*2/3,3.2*2/3);
\draw (1.5*2/3,4.4*2/3) -- (1.5*2/3,5.6*2/3) ;
\draw (-1*2/3,2*2/3) -- (1.5*2/3,4.4*2/3);
\draw  (0.2*2/3,2.2*2/3) -- (2.8*2/3,2.2*2/3);
\draw (0,0) -- (-1.2*2/3,0);
\draw (4.2*2/3,0) -- (3*2/3,0);
\draw (0.2*2/3,2.2*2/3) -- (-1.5*2/3,2.2*2/3+0.9);
\draw (2.8*2/3,2.2*2/3) -- (4.5*2/3,2.2*2/3+0.9);
\draw[dashed] (1.5*2/3,5.6*2/3) -- (-1.5*2/3,2.2*2/3+0.9);
\draw[dashed] (1.5*2/3,5.6*2/3) -- (4.5*2/3,2.2*2/3+0.9);
\draw[dashed] (1.5*2/3,5.6*2/3) -- (-1.2*2/3,0);
\draw (4*2/3,2*2/3) -- (4*2/3+0.8,2*2/3) ;
\node at (0,0) [circle,draw=red!50,fill=red!20] {};
\node at (0,0) {2};
\node at (-1*2/3,2*2/3) [circle,draw=green!50,fill=green!20] {};
\node at (-1*2/3,2*2/3) {3};
\node at (3*2/3,0) [circle,draw=green!50,fill=green!20] {};
\node at (3*2/3,0) {3};
\node at (4*2/3,2*2/3) [circle,draw=yellow!50,fill=yellow!20] {};
\node at (4*2/3,2*2/3){4};
\node at (1.5*2/3,4.4*2/3) [circle,draw=black!50,fill=black!20] {};
\node at (1.5*2/3,4.4*2/3) {5};
\node at (1.5*2/3,3.2*2/3) [circle,draw=red!50,fill=red!20] {};
\node at (1.5*2/3,3.2*2/3) {2};
\node at (0.7*2/3,0.7*2/3) [circle,draw=blue!50,fill=blue!20] {};
\node at (0.7*2/3,0.7*2/3) {1};
\node at (2.8*2/3,2.2*2/3) [circle,draw=red!50,fill=red!20] {};
\node at (2.8*2/3,2.2*2/3) {2};
\node at (0.2*2/3,2.2*2/3) [circle,draw=blue!50,fill=blue!20] {};
\node at (0.2*2/3,2.2*2/3) {1};
\node at (1.5*2/3,5.6*2/3) [circle,draw=blue!50,fill=blue!20] {};
\node at (1.5*2/3,5.6*2/3) {1};
\node at (4*2/3+0.8,2*2/3) [circle,draw=blue!50,fill=blue!20] {};
\node at (4*2/3+0.8,2*2/3) {1};
\node at (4.2*2/3,0) [circle,draw=blue!50,fill=blue!20] {};
\node at (4.2*2/3,0) {1};
\node at (-1.5*2/3,2.2*2/3+0.9)[ circle,draw=black,fill=black] {};
\node at (4.5*2/3,2.2*2/3+0.9)[ circle,draw=black,fill=black] {};
\node at (-1.2*2/3,0)[ circle,draw=black,fill=black] {};
\node at (2.3*2/3,0.7*2/3)[ circle,draw=black,fill=black] {};
\node at (-0.3,-0.3){$u_{3}$};
\node at (-1*2/3-0.4,2*2/3){$u_{2}$};
\node at (3*2/3+0.3,-0.3){$u_{4}$};
\node at (4*2/3+0.4,2*2/3-0.2){$u_{5}$};
\node at (1.5*2/3+0.4,4.4*2/3+0.1){$u_{1}$};
\node at (0.7*2/3-0.4,0.7*2/3+0.1){$v_{3}$};
\node at (2.3*2/3+0.4,0.7*2/3+0.1){$v_{4}$};
\node at (0.2*2/3,2.2*2/3+0.3){$v_{2}$};
\node at (2.8*2/3,2.2*2/3+0.3){$v_{5}$};
\node at (1.5*2/3+0.4,3.2*2/3){$v_{1}$};
\node at (1.5*2/3+0.4,5.6*2/3){$v'_{1}$};
\node at (-1.5*2/3-0.5,2.2*2/3+0.9){$u'_{2}$};
\node at (4.5*2/3+0.5,2.2*2/3+0.9){$u'_{5}$};
\node at (-1.2*2/3-0.3,-0.3){$v'_{3}$};
\node at (4.2*2/3+0.3,-0.3){$v'_{4}$};
\node at (4*2/3+1.2,2*2/3-0.2){$v'_{5}$};

\draw (0+6.3,0) -- (3*2/3+6.3,0);
\draw (0+6.3,0) -- (-1*2/3+6.3,2*2/3);
\draw (3*2/3+6.3,0) -- (4*2/3+6.3,2*2/3);
\draw (-1*2/3+6.3,2*2/3) -- (1.5*2/3+6.3,4.4*2/3);
\draw (4*2/3+6.3,2*2/3) -- (1.5*2/3+6.3,4.4*2/3);
\draw (1.5*2/3+6.3,3.2*2/3) -- (1.5*2/3+6.3,4.4*2/3);
\draw (-1*2/3+6.3,2*2/3) -- (0.2*2/3+6.3,2.2*2/3);
\draw (4*2/3+6.3,2*2/3) -- (2.8*2/3+6.3,2.2*2/3);
\draw (0+6.3,0) -- (0.7*2/3+6.3,0.7*2/3);
\draw (3*2/3+6.3,0) -- (2.3*2/3+6.3,0.7*2/3);
\draw (0.7*2/3+6.3,0.7*2/3) -- (2.8*2/3+6.3,2.2*2/3);
\draw (2.3*2/3+6.3,0.7*2/3) -- (0.2*2/3+6.3,2.2*2/3);
\draw (0.7*2/3+6.3,0.7*2/3) -- (1.5*2/3+6.3,3.2*2/3);
\draw (2.3*2/3+6.3,0.7*2/3) -- (1.5*2/3+6.3,3.2*2/3);
\draw (1.5*2/3+6.3,4.4*2/3) -- (1.5*2/3+6.3,5.6*2/3) ;
\draw (-1*2/3+6.3,2*2/3) -- (1.5*2/3+6.3,4.4*2/3);
\draw  (0.2*2/3+6.3,2.2*2/3) -- (2.8*2/3+6.3,2.2*2/3);
\draw (0+6.3,0) -- (-1.2*2/3+6.3,0);
\draw (4.2*2/3+6.3,0) -- (3*2/3+6.3,0);
\draw (0.2*2/3+6.3,2.2*2/3) -- (-1.5*2/3+6.3,2.2*2/3+0.9);
\draw (2.8*2/3+6.3,2.2*2/3) -- (4.5*2/3+6.3,2.2*2/3+0.9);
\draw[dashed] (1.5*2/3+6.3,5.6*2/3) -- (-1.5*2/3+6.3,2.2*2/3+0.9);
\draw[dashed] (1.5*2/3+6.3,5.6*2/3) -- (4.2*2/3+6.3,0);
\draw[dashed] (1.5*2/3+6.3,5.6*2/3) -- (-1.2*2/3+6.3,0);
\draw (4*2/3+6.3,2*2/3) -- (4*2/3+6.3+0.8,2*2/3) ;
\node at (0+6.3,0) [circle,draw=blue!50,fill=blue!20] {};
\node at (0+6.3,0) {1};
\node at (-1*2/3+6.3,2*2/3) [circle,draw=red!50,fill=red!20] {};
\node at (-1*2/3+6.3,2*2/3) {2};
\node at (3*2/3+6.3,0) [circle,draw=red!50,fill=red!20] {};
\node at (3*2/3+6.3,0) {2};
\node at (4*2/3+6.3,2*2/3) [circle,draw=yellow!50,fill=yellow!20] {};
\node at (4*2/3+6.3,2*2/3){4};
\node at (1.5*2/3+6.3,4.4*2/3) [circle,draw=black!50,fill=black!20] {};
\node at (1.5*2/3+6.3,4.4*2/3) {5};
\node at (1.5*2/3+6.3,3.2*2/3) [circle,draw=green!50,fill=green!20] {};
\node at (1.5*2/3+6.3,3.2*2/3) {3};
\node at (0.7*2/3+6.3,0.7*2/3) [circle,draw=red!50,fill=red!20] {};
\node at (0.7*2/3+6.3,0.7*2/3) {2};
\node at (2.8*2/3+6.3,2.2*2/3) [circle,draw=green!50,fill=green!20] {};
\node at (2.8*2/3+6.3,2.2*2/3) {3};
\node at (2.3*2/3+6.3,0.7*2/3) [circle,draw=blue!50,fill=blue!20] {};
\node at (2.3*2/3+6.3,0.7*2/3) {1};
\node at (1.5*2/3+6.3,5.6*2/3) [circle,draw=blue!50,fill=blue!20] {};
\node at (1.5*2/3+6.3,5.6*2/3) {1};
\node at (4.5*2/3+6.3,2.2*2/3+0.9) [circle,draw=blue!50,fill=blue!20] {};
\node at (4.5*2/3+6.3,2.2*2/3+0.9) {1};
\node at (4*2/3+6.3+0.8,2*2/3) [circle,draw=blue!50,fill=blue!20] {};
\node at (4*2/3+6.3+0.8,2*2/3) {1};
\node at (-1.5*2/3+6.3,2.2*2/3+0.9) [circle,draw=black,fill=black] {};
\node at (4.2*2/3+6.3,0) [circle,draw=black,fill=black] {};
\node at (-1.2*2/3+6.3,0) [circle,draw=black,fill=black] {};
\node at (0.2*2/3+6.3,2.2*2/3) [circle,draw=black,fill=black] {};
\node at (-0.3+6.3,-0.3){$u_{3}$};
\node at (-1*2/3-0.4+6.3,2*2/3){$u_{2}$};
\node at (3*2/3+0.3+6.3,-0.3){$u_{4}$};
\node at (4*2/3+0.4+6.3,2*2/3-0.2){$u_{5}$};
\node at (1.5*2/3+0.4+6.3,4.4*2/3+0.1){$u_{1}$};
\node at (0.7*2/3-0.4+6.3,0.7*2/3+0.1){$v_{3}$};
\node at (2.3*2/3+0.4+6.3,0.7*2/3+0.1){$v_{4}$};
\node at (0.2*2/3+6.3,2.2*2/3+0.3){$v_{2}$};
\node at (2.8*2/3+6.3,2.2*2/3+0.3){$v_{5}$};
\node at (1.5*2/3+6.3+0.4,3.2*2/3){$v_{1}$};
\node at (1.5*2/3+6.3+0.4,5.6*2/3){$v'_{1}$};
\node at (-1.5*2/3+6.3-0.5,2.2*2/3+0.9){$u'_{2}$};
\node at (4.5*2/3+6.3+0.5,2.2*2/3+0.9){$u'_{5}$};
\node at (-1.2*2/3-0.3+6.3,-0.3){$v'_{3}$};
\node at (4.2*2/3+0.3+6.3,-0.3){$v'_{4}$};
\node at (4*2/3+1.2+6.3,2*2/3-0.2){$v'_{5}$};
\end{tikzpicture}
\end{center}
\caption{Two Grundy partial 5-colorings of a subgraph containing an induced Petersen graph.}
\label{figm51}
\end{figure}
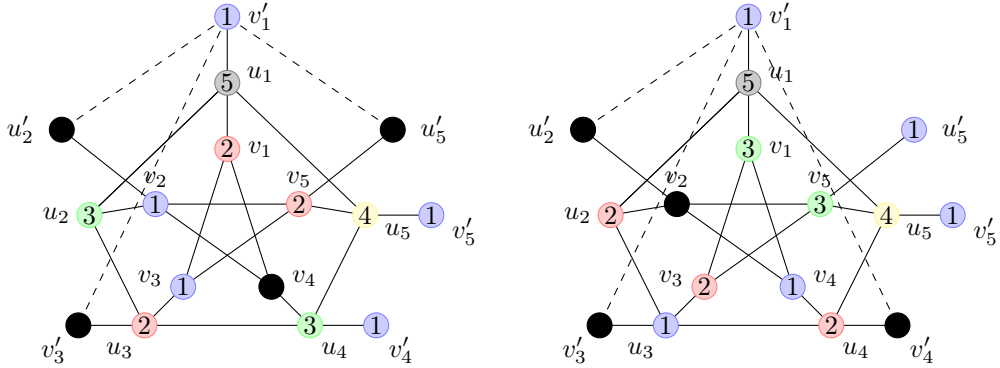
In a graph $G$, let a \textit{neighbor-connected} $C_n$ be an $n$-cycle $C$ such that the set of vertices of $G$ at distance 1 from $C$ is not independent.
\begin{lem}
Let $G$ be a 4-regular graph with girth $g=5$. If $G$ contains a neighbor-connected $C_5$ as induced subgraph, then $\Gamma(G)=5$.
\end{lem}
\begin{proof}
Let $C$ be a neighbor-connected $C_5$ in $G$.
By Lemma~\ref{lpet} we can suppose that the neighbors of the vertices of $C$ do not form an induced $C_5$ (otherwise a
Petersen would be an induced subgraph). Hence, we can assume that the neighbors of the vertices of $C$ form a subgraph of a $C_{10}$.
If there are two edges between the neighbors of the vertices of $C$, then Figure~\ref{figm52} illustrates Grundy partial 5-colorings of the graph $G$.
Suppose that two neighbors are adjacent, say $v_1$ and $v'_3$ and the graph $G$ does not contain the previous configuration. Note that $v'_3$ can be adjacent with $v'_1$ and $v'_5$.
Let $w_1$, $w_2$ and $w_3$ be the three neighbors of $v_2$ different from $u_2$.
We suppose that $w_1$ can be possibly adjacent with $v'_3$ and $w_2$ can be possibly adjacent with $v'_1$.
Figure~\ref{figm53} illustrates a Grundy partial 5-coloring of $G$ in this case.
In this figure, the vertex $w_3$ can be possibly adjacent with $v'_5$ or $v_4$, but in this case we can switch the color 1 from $v'_5$ to $v_5$ or from $v'_4$ to $v_4$.
\end{proof}
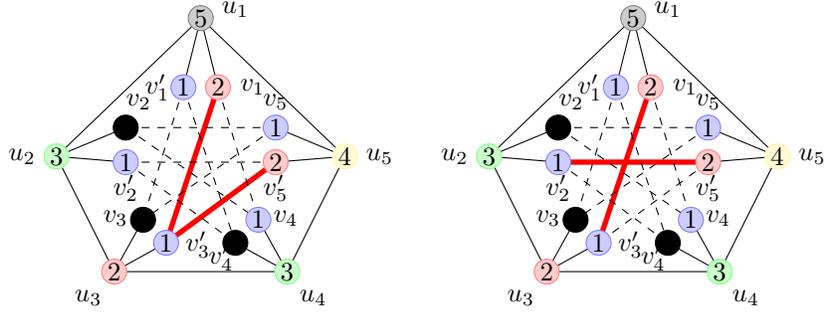
\begin{figure}[t]
\begin{center}
\begin{tikzpicture}[scale=1.15]
\draw (0,0) -- (3*2/3,0);
\draw (0,0) -- (-1*2/3,2*2/3);
\draw (3*2/3,0) -- (4*2/3,2*2/3);
\draw (-1*2/3,2*2/3) -- (1.5*2/3,4.4*2/3);
\draw (4*2/3,2*2/3) -- (1.5*2/3,4.4*2/3);
\draw (1.8*2/3,3.2*2/3) -- (1.5*2/3,4.4*2/3);
\draw (1.2*2/3,3.2*2/3) -- (1.5*2/3,4.4*2/3);
\draw (-1*2/3,2*2/3) -- (0.2*2/3,2.5*2/3);
\draw (-1*2/3,2*2/3) -- (0.2*2/3,1.9*2/3);
\draw (4*2/3,2*2/3) -- (2.8*2/3,2.5*2/3);
\draw (4*2/3,2*2/3) -- (2.8*2/3,1.9*2/3);
\draw (0,0) -- (0.5*2/3,0.9*2/3);
\draw (0,0) -- (0.9*2/3,0.5*2/3);
\draw (3*2/3,0) -- (2.5*2/3,0.9*2/3);
\draw (3*2/3,0) -- (2.1*2/3,0.5*2/3);
\draw[line width=2pt,color=red] (1.8*2/3,3.2*2/3) -- (0.9*2/3,0.5*2/3);
\draw[line width=2pt,color=red]  (2.8*2/3,1.9*2/3) -- (0.9*2/3,0.5*2/3);
\draw[dashed] (1.2*2/3,3.2*2/3) -- (0.5*2/3,0.9*2/3);
\draw[dashed] (2.8*2/3,2.5*2/3) --  (0.5*2/3,0.9*2/3);
\draw[dashed] (0.2*2/3,2.5*2/3) -- (2.8*2/3,2.5*2/3);
\draw[dashed] (0.2*2/3,1.9*2/3) --  (2.8*2/3,1.9*2/3);
\draw[dashed] (0.2*2/3,2.5*2/3) -- (2.5*2/3,0.9*2/3);
\draw[dashed] (0.2*2/3,1.9*2/3) --  (2.1*2/3,0.5*2/3);
\draw[dashed] (1.8*2/3,3.2*2/3) -- (2.5*2/3,0.9*2/3);
\draw[dashed] (1.2*2/3,3.2*2/3) --  (2.1*2/3,0.5*2/3);
\node at (0,0) [circle,draw=red!50,fill=red!20] {};
\node at (0,0) {2};
\node at (-1*2/3,2*2/3) [circle,draw=green!50,fill=green!20] {};
\node at (-1*2/3,2*2/3) {3};
\node at (3*2/3,0) [circle,draw=green!50,fill=green!20] {};
\node at (3*2/3,0) {3};
\node at (4*2/3,2*2/3) [circle,draw=yellow!50,fill=yellow!20] {};
\node at (4*2/3,2*2/3){4};
\node at (1.5*2/3,4.4*2/3) [circle,draw=black!50,fill=black!20] {};
\node at (1.5*2/3,4.4*2/3) {5};
\node at (1.8*2/3,3.2*2/3) [circle,draw=red!50,fill=red!20] {};
\node at (1.8*2/3,3.2*2/3){2};
\node at (1.2*2/3,3.2*2/3)[circle,draw=blue!50,fill=blue!20] {};
\node at (1.2*2/3,3.2*2/3){1};
\node at (0.9*2/3,0.5*2/3) [circle,draw=blue!50,fill=blue!20] {};
\node at (0.9*2/3,0.5*2/3) {1};
\node at (2.8*2/3,1.9*2/3) [circle,draw=red!50,fill=red!20] {};
\node at (2.8*2/3,1.9*2/3) {2};
\node at (2.8*2/3,2.5*2/3) [circle,draw=blue!50,fill=blue!20] {};
\node at (2.8*2/3,2.5*2/3) {1};
\node at (0.2*2/3,1.9*2/3) [circle,draw=blue!50,fill=blue!20] {};
\node at (0.2*2/3,1.9*2/3) {1};
\node at (2.5*2/3,0.9*2/3) [circle,draw=blue!50,fill=blue!20] {};
\node at (2.5*2/3,0.9*2/3) {1};
\node at (0.5*2/3,0.9*2/3)[ circle,draw=black,fill=black] {};
\node at (2.1*2/3,0.5*2/3)[ circle,draw=black,fill=black] {};
\node at (0.2*2/3,2.5*2/3)[ circle,draw=black,fill=black] {};
\node at (-0.3,-0.3){$u_{3}$};
\node at (-1*2/3-0.4,2*2/3){$u_{2}$};
\node at (3*2/3+0.3,-0.3){$u_{4}$};
\node at (4*2/3+0.4,2*2/3){$u_{5}$};
\node at (1.5*2/3+0.4,4.4*2/3+0.1){$u_{1}$};
\node at (0.7*2/3-0.45,0.7*2/3+0.15){$v_{3}$};
\node at (0.7*2/3+0.5,0.7*2/3-0.15){$v'_{3}$};
\node at (2.3*2/3+0.45,0.7*2/3+0.15){$v_{4}$};
\node at (2.3*2/3-0.3,0.7*2/3-0.3){$v'_{4}$};
\node at (0.2*2/3+0.15,2.2*2/3+0.5){$v_{2}$};
\node at (0.2*2/3,2.2*2/3-0.5){$v'_{2}$};
\node at (2.8*2/3,2.2*2/3+0.5){$v_{5}$};
\node at (2.8*2/3,2.2*2/3-0.5){$v'_{5}$};
\node at (1.5*2/3+0.6,3.2*2/3){$v_{1}$};
\node at (0.75*2/3,3.2*2/3) {$v'_{1}$};

\draw (0+5,0) -- (3*2/3+5,0);
\draw (0+5,0) -- (-1*2/3+5,2*2/3);
\draw (3*2/3+5,0) -- (4*2/3+5,2*2/3);
\draw (-1*2/3+5,2*2/3) -- (1.5*2/3+5,4.4*2/3);
\draw (4*2/3+5,2*2/3) -- (1.5*2/3+5,4.4*2/3);
\draw (1.8*2/3+5,3.2*2/3) -- (1.5*2/3+5,4.4*2/3);
\draw (1.2*2/3+5,3.2*2/3) -- (1.5*2/3+5,4.4*2/3);
\draw (-1*2/3+5,2*2/3) -- (0.2*2/3+5,2.5*2/3);
\draw (-1*2/3+5,2*2/3) -- (0.2*2/3+5,1.9*2/3);
\draw (4*2/3+5,2*2/3) -- (2.8*2/3+5,2.5*2/3);
\draw (4*2/3+5,2*2/3) -- (2.8*2/3+5,1.9*2/3);
\draw (5,0) -- (0.5*2/3+5,0.9*2/3);
\draw (5,0) -- (0.9*2/3+5,0.5*2/3);
\draw (3*2/3+5,0) -- (2.5*2/3+5,0.9*2/3);
\draw (3*2/3+5,0) -- (2.1*2/3+5,0.5*2/3);
\draw[line width=2pt,color=red] (1.8*2/3+5,3.2*2/3) -- (0.9*2/3+5,0.5*2/3);
\draw[dashed] (2.8*2/3+5,1.9*2/3) -- (0.9*2/3+5,0.5*2/3);

\draw[dashed] (1.2*2/3+5,3.2*2/3) -- (0.5*2/3+5,0.9*2/3);
\draw[dashed] (2.8*2/3+5,2.5*2/3) --  (0.5*2/3+5,0.9*2/3);
\draw[dashed] (0.2*2/3+5,2.5*2/3) -- (2.8*2/3+5,2.5*2/3);
\draw[line width=2pt,color=red] (0.2*2/3+5,1.9*2/3) --  (2.8*2/3+5,1.9*2/3);
\draw[dashed] (0.2*2/3+5,2.5*2/3) -- (2.5*2/3+5,0.9*2/3);
\draw[dashed] (0.2*2/3+5,1.9*2/3) --  (2.1*2/3+5,0.5*2/3);
\draw[dashed] (1.8*2/3+5,3.2*2/3) -- (2.5*2/3+5,0.9*2/3);
\draw[dashed] (1.2*2/3+5,3.2*2/3) --  (2.1*2/3+5,0.5*2/3);
\node at (0+5,0) [circle,draw=red!50,fill=red!20] {};
\node at (0+5,0) {2};
\node at (-1*2/3+5,2*2/3) [circle,draw=green!50,fill=green!20] {};
\node at (-1*2/3+5,2*2/3) {3};
\node at (3*2/3+5,0) [circle,draw=green!50,fill=green!20] {};
\node at (3*2/3+5,0) {3};
\node at (4*2/3+5,2*2/3) [circle,draw=yellow!50,fill=yellow!20] {};
\node at (4*2/3+5,2*2/3){4};
\node at (1.5*2/3+5,4.4*2/3) [circle,draw=black!50,fill=black!20] {};
\node at (1.5*2/3+5,4.4*2/3) {5};
\node at (1.8*2/3+5,3.2*2/3) [circle,draw=red!50,fill=red!20] {};
\node at (1.8*2/3+5,3.2*2/3){2};
\node at (0.9*2/3+5,0.5*2/3) [circle,draw=blue!50,fill=blue!20] {};
\node at (0.9*2/3+5,0.5*2/3) {1};
\node at (2.8*2/3+5,1.9*2/3) [circle,draw=red!50,fill=red!20] {};
\node at (2.8*2/3+5,1.9*2/3) {2};
\node at (2.8*2/3+5,2.5*2/3) [circle,draw=blue!50,fill=blue!20] {};
\node at (2.8*2/3+5,2.5*2/3) {1};
\node at (0.2*2/3+5,1.9*2/3) [circle,draw=blue!50,fill=blue!20] {};
\node at (0.2*2/3+5,1.9*2/3) {1};
\node at (2.5*2/3+5,0.9*2/3) [circle,draw=blue!50,fill=blue!20] {};
\node at (2.5*2/3+5,0.9*2/3) {1};
\node at (1.2*2/3+5,3.2*2/3)[circle,draw=blue!50,fill=blue!20] {};
\node at (1.2*2/3+5,3.2*2/3){1};
\node at (0.5*2/3+5,0.9*2/3)[ circle,draw=black,fill=black] {};
\node at (2.1*2/3+5,0.5*2/3)[ circle,draw=black,fill=black] {};
\node at (0.2*2/3+5,2.5*2/3)[ circle,draw=black,fill=black] {};
\node at (-0.3+5,-0.3){$u_{3}$};
\node at (-1*2/3-0.4+5,2*2/3){$u_{2}$};
\node at (3*2/3+0.3+5,-0.3){$u_{4}$};
\node at (4*2/3+0.4+5,2*2/3){$u_{5}$};
\node at (1.5*2/3+0.4+5,4.4*2/3+0.1){$u_{1}$};
\node at (0.7*2/3-0.45+5,0.7*2/3+0.15){$v_{3}$};
\node at (0.7*2/3+0.5+5,0.7*2/3-0.15){$v'_{3}$};
\node at (2.3*2/3+0.45+5,0.7*2/3+0.15){$v_{4}$};
\node at (2.3*2/3-0.3+5,0.7*2/3-0.3){$v'_{4}$};
\node at (0.2*2/3+5+0.15,2.2*2/3+0.5){$v_{2}$};
\node at (0.2*2/3+5,2.2*2/3-0.5){$v'_{2}$};
\node at (2.8*2/3+5,2.2*2/3+0.5){$v_{5}$};
\node at (2.8*2/3+5,2.2*2/3-0.5){$v'_{5}$};
\node at (1.5*2/3+0.6+5,3.2*2/3){$v_{1}$};
\node at (0.75*2/3+5,3.2*2/3) {$v'_{1}$};

\end{tikzpicture}
\end{center}
\caption{Two Grundy partial 5-colorings of a subgraph containing an induced neighbor-connected $C_5$.}
\label{figm52}
\end{figure}
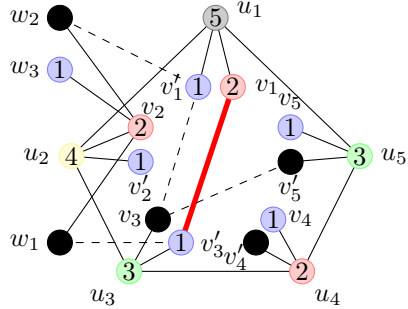
\begin{figure}[t]
\begin{center}
\begin{tikzpicture}[scale=1.15]
\draw (0,0) -- (3*2/3,0);
\draw (0,0) -- (-1*2/3,2*2/3);
\draw (3*2/3,0) -- (4*2/3,2*2/3);
\draw (-1*2/3,2*2/3) -- (1.5*2/3,4.4*2/3);
\draw (4*2/3,2*2/3) -- (1.5*2/3,4.4*2/3);
\draw (1.8*2/3,3.2*2/3) -- (1.5*2/3,4.4*2/3);
\draw (1.2*2/3,3.2*2/3) -- (1.5*2/3,4.4*2/3);
\draw (-1*2/3,2*2/3) -- (0.2*2/3,2.5*2/3);
\draw (-1*2/3,2*2/3) -- (0.2*2/3,1.9*2/3);
\draw (4*2/3,2*2/3) -- (2.8*2/3,2.5*2/3);
\draw (4*2/3,2*2/3) -- (2.8*2/3,1.9*2/3);
\draw (0,0) -- (0.5*2/3,0.9*2/3);
\draw (0,0) -- (0.9*2/3,0.5*2/3);
\draw (3*2/3,0) -- (2.5*2/3,0.9*2/3);
\draw (3*2/3,0) -- (2.1*2/3,0.5*2/3);
\draw[line width=2pt,color=red] (1.8*2/3,3.2*2/3) -- (0.9*2/3,0.5*2/3);
\draw[dashed] (1.2*2/3,3.2*2/3) -- (0.5*2/3,0.9*2/3);
\draw[dashed]  (2.8*2/3,1.9*2/3) --  (0.5*2/3,0.9*2/3); 
\draw (0.2*2/3,2.5*2/3) -- (-1.2*2/3,3.5*2/3);
\draw (0.2*2/3,2.5*2/3) -- (-1.2*2/3,4.4*2/3);
\draw (0.2*2/3,2.5*2/3) -- (-1.2*2/3,0.5*2/3);
\draw[dashed] (-1.2*2/3,0.5*2/3) -- (0.9*2/3,0.5*2/3);
\draw[dashed]  (-1.2*2/3,4.4*2/3) --  (1.2*2/3,3.2*2/3);
\node at (3*2/3,0) [circle,draw=red!50,fill=red!20] {};
\node at (3*2/3,0)  {2};
\node at (4*2/3,2*2/3) [circle,draw=green!50,fill=green!20] {};
\node at (4*2/3,2*2/3)  {3};
\node at (0,0) [circle,draw=green!50,fill=green!20] {};
\node at (0,0) {3};
\node at (-1*2/3,2*2/3) [circle,draw=yellow!50,fill=yellow!20] {};
\node at (-1*2/3,2*2/3) {4};
\node at (1.5*2/3,4.4*2/3) [circle,draw=black!50,fill=black!20] {};
\node at (1.5*2/3,4.4*2/3) {5};
\node at (1.8*2/3,3.2*2/3) [circle,draw=red!50,fill=red!20] {};
\node at (1.8*2/3,3.2*2/3){2};
\node at (0.9*2/3,0.5*2/3) [circle,draw=blue!50,fill=blue!20] {};
\node at (0.9*2/3,0.5*2/3) {1};
\node at (0.2*2/3,2.5*2/3) [circle,draw=red!50,fill=red!20] {};
\node at (0.2*2/3,2.5*2/3) {2};
\node at (2.8*2/3,2.5*2/3) [circle,draw=blue!50,fill=blue!20] {};
\node at (2.8*2/3,2.5*2/3) {1};
\node at (0.2*2/3,1.9*2/3) [circle,draw=blue!50,fill=blue!20] {};
\node at (0.2*2/3,1.9*2/3) {1};
\node at (2.5*2/3,0.9*2/3) [circle,draw=blue!50,fill=blue!20] {};
\node at (2.5*2/3,0.9*2/3) {1};
\node at (1.2*2/3,3.2*2/3)[circle,draw=blue!50,fill=blue!20] {};
\node at (1.2*2/3,3.2*2/3){1};
\node at (-1.2*2/3,3.5*2/3) [circle,draw=blue!50,fill=blue!20] {};
\node at (-1.2*2/3,3.5*2/3) {1};
\node at (2.8*2/3,1.9*2/3) [ circle,draw=black,fill=black] {};
\node at (0.5*2/3,0.9*2/3) [ circle,draw=black,fill=black] {};
\node at (2.2*2/3,0.5*2/3) [ circle,draw=black,fill=black] {};
\node at (-1.2*2/3,0.5*2/3) [ circle,draw=black,fill=black] {};
\node at (-1.2*2/3,4.4*2/3) [ circle,draw=black,fill=black] {};
\node at (-0.3,-0.3){$u_{3}$};
\node at (-1*2/3-0.4,2*2/3){$u_{2}$};
\node at (3*2/3+0.3,-0.3){$u_{4}$};
\node at (4*2/3+0.4,2*2/3){$u_{5}$};
\node at (1.5*2/3+0.4,4.4*2/3+0.1){$u_{1}$};
\node at (0.7*2/3-0.45,0.7*2/3+0.15){$v_{3}$};
\node at (0.7*2/3+0.5,0.7*2/3-0.15){$v'_{3}$};
\node at (2.3*2/3+0.45,0.7*2/3+0.15){$v_{4}$};
\node at (2.3*2/3-0.3,0.7*2/3-0.3){$v'_{4}$};
\node at (0.2*2/3+0.15,2.2*2/3+0.4){$v_{2}$};
\node at (0.2*2/3,2.2*2/3-0.5){$v'_{2}$};
\node at (2.8*2/3,2.2*2/3+0.5){$v_{5}$};
\node at (2.8*2/3,2.2*2/3-0.5){$v'_{5}$};
\node at (1.5*2/3+0.6,3.2*2/3){$v_{1}$};
\node at (0.75*2/3,3.2*2/3) {$v'_{1}$};
\node at (-1.8*2/3,4.4*2/3) {$w_{2}$};
\node at (-1.8*2/3,0.5*2/3) {$w_{1}$};
\node at (-1.8*2/3,3.5*2/3) {$w_{3}$};

\end{tikzpicture}
\end{center}
\caption{A Grundy partial 5-coloring of a subgraph containing an induced neighbor-connected $C_5$.}
\label{figm53}
\end{figure}
\begin{prop}
If $G$ is a 4-regular graph with girth $g=5$, then $\Gamma(G)=5$.
\label{cycle5}
\end{prop}
\begin{proof}
Let $C$ be a 5-cycle in $G$.
Assume that two neighbors of consecutive vertices of $C$, for example $v_1$ and $v_5$, have a common neighbor $w_1$.
The left part of Figure~\ref{figm54} illustrates a Grundy partial 5-coloring of the graph $G$.
In this figure the vertex $w_1$ can be possibly adjacent with $v'_2$, $v'_3$ or $v_4$, but in this case we can switch the color 1 from $v'_2$ to $v_2$, from $v'_3$ to $v_3$ or from $v_4$ to $v'_4$.
Hence, we can suppose that no neighbors of consecutive vertices of $C$ are adjacent. 
Among the neighbors of $v_1$, there exists one vertex $w_1$ not adjacent with both $v_4$ and $v'_4$ (otherwise $G$ would contain a $C_4$). 
Among the neighbor of $v'_5$, there exists one vertex, say $w_2$, not adjacent with $w_1$. The right part of Figure~\ref{figm54} illustrates a Grundy partial 5-coloring of the graph $G$.
In this figure the vertex $w_1$ can be possibly adjacent with $v_4$ and the vertex $w_2$ can be possibly adjacent with $v'_2$ or $v_4$, but in these cases we can switch the color 1 from $v'_2$ to $v_2$ or from $v_4$ to $v'_4$.
\begin{figure}[t]
\begin{center}
\begin{tikzpicture}[scale=1.15]
\draw (0,0) -- (3*2/3,0);
\draw (0,0) -- (-1*2/3,2*2/3);
\draw (3*2/3,0) -- (4*2/3,2*2/3);
\draw (-1*2/3,2*2/3) -- (1.5*2/3,4.4*2/3);
\draw (4*2/3,2*2/3) -- (1.5*2/3,4.4*2/3);
\draw (1.8*2/3,3.2*2/3) -- (1.5*2/3,4.4*2/3);
\draw (1.2*2/3,3.2*2/3) -- (1.5*2/3,4.4*2/3);
\draw (-1*2/3,2*2/3) -- (0.2*2/3,2.5*2/3);
\draw (-1*2/3,2*2/3) -- (0.2*2/3,1.9*2/3);
\draw (4*2/3,2*2/3) -- (2.8*2/3,2.5*2/3);
\draw (4*2/3,2*2/3) -- (2.8*2/3,1.9*2/3);
\draw (0,0) -- (0.5*2/3,0.9*2/3);
\draw (0,0) -- (0.9*2/3,0.5*2/3);
\draw (3*2/3,0) -- (2.5*2/3,0.9*2/3);
\draw (3*2/3,0) -- (2.1*2/3,0.5*2/3);
\draw (2.8*2/3,2.5*2/3) --  (3.8*2/3,4.4*2/3);
\draw  (1.8*2/3,3.2*2/3) --  (3.8*2/3,4.4*2/3);
\node at (0,0) [circle,draw=red!50,fill=red!20] {};
\node at (0,0) {2};
\node at (-1*2/3,2*2/3) [circle,draw=green!50,fill=green!20] {};
\node at (-1*2/3,2*2/3) {3};
\node at (3*2/3,0) [circle,draw=green!50,fill=green!20] {};
\node at (3*2/3,0) {3};
\node at (4*2/3,2*2/3) [circle,draw=yellow!50,fill=yellow!20] {};
\node at (4*2/3,2*2/3){4};
\node at (1.5*2/3,4.4*2/3) [circle,draw=black!50,fill=black!20] {};
\node at (1.5*2/3,4.4*2/3) {5};
\node at (1.8*2/3,3.2*2/3) [circle,draw=red!50,fill=red!20] {};
\node at (1.8*2/3,3.2*2/3){2};
\node at (0.9*2/3,0.5*2/3) [circle,draw=blue!50,fill=blue!20] {};
\node at (0.9*2/3,0.5*2/3) {1};
\node at (2.8*2/3,2.5*2/3) [circle,draw=red!50,fill=red!20] {};
\node at (2.8*2/3,2.5*2/3) {2};
\node at (2.8*2/3,1.9*2/3) [circle,draw=blue!50,fill=blue!20] {};
\node at (2.8*2/3,1.9*2/3) {1};
\node at (0.2*2/3,1.9*2/3) [circle,draw=blue!50,fill=blue!20] {};
\node at (0.2*2/3,1.9*2/3) {1};
\node at (2.5*2/3,0.9*2/3) [circle,draw=blue!50,fill=blue!20] {};
\node at (2.5*2/3,0.9*2/3) {1};
\node at (1.2*2/3,3.2*2/3)[circle,draw=blue!50,fill=blue!20] {};
\node at (1.2*2/3,3.2*2/3){1};
\node at (3.8*2/3,4.4*2/3) [circle,draw=blue!50,fill=blue!20] {};
\node at (3.8*2/3,4.4*2/3){1};
\node at (0.2*2/3,2.5*2/3) [ circle,draw=black,fill=black] {};
\node at (0.5*2/3,0.9*2/3) [ circle,draw=black,fill=black] {};
\node at (2.2*2/3,0.5*2/3) [ circle,draw=black,fill=black] {};
\node at (-0.3,-0.3){$u_{3}$};
\node at (-1*2/3-0.4,2*2/3){$u_{2}$};
\node at (3*2/3+0.3,-0.3){$u_{4}$};
\node at (4*2/3+0.4,2*2/3){$u_{5}$};
\node at (1.5*2/3+0.4,4.4*2/3+0.1){$u_{1}$};
\node at (0.7*2/3-0.45,0.7*2/3+0.15){$v_{3}$};
\node at (0.7*2/3+0.5,0.7*2/3-0.15){$v'_{3}$};
\node at (2.3*2/3+0.45,0.7*2/3+0.15){$v_{4}$};
\node at (2.3*2/3-0.3,0.7*2/3-0.3){$v'_{4}$};
\node at (0.2*2/3+0.1,2.2*2/3+0.5){$v_{2}$};
\node at (0.2*2/3,2.2*2/3-0.5){$v'_{2}$};
\node at (2.8*2/3,2.2*2/3+0.5){$v_{5}$};
\node at (2.8*2/3,2.2*2/3-0.5){$v'_{5}$};
\node at (1.5*2/3+0.6,3.2*2/3){$v_{1}$};
\node at (0.75*2/3,3.2*2/3) {$v'_{1}$};
\node at (3.8*2/3+0.4,4.4*2/3) {$w_{1}$};

\draw (0+5,0) -- (3*2/3+5,0);
\draw (0+5,0) -- (-1*2/3+5,2*2/3);
\draw (3*2/3+5,0) -- (4*2/3+5,2*2/3);
\draw (-1*2/3+5,2*2/3) -- (1.5*2/3+5,4.4*2/3);
\draw (4*2/3+5,2*2/3) -- (1.5*2/3+5,4.4*2/3);
\draw (1.8*2/3+5,3.2*2/3) -- (1.5*2/3+5,4.4*2/3);
\draw (1.2*2/3+5,3.2*2/3) -- (1.5*2/3+5,4.4*2/3);
\draw (-1*2/3+5,2*2/3) -- (0.2*2/3+5,2.5*2/3);
\draw (-1*2/3+5,2*2/3) -- (0.2*2/3+5,1.9*2/3);
\draw (4*2/3+5,2*2/3) -- (2.8*2/3+5,2.5*2/3);
\draw (4*2/3+5,2*2/3) -- (2.8*2/3+5,1.9*2/3);
\draw (0+5,0) -- (0.5*2/3+5,0.9*2/3);
\draw (0+5,0) -- (0.9*2/3+5,0.5*2/3);
\draw (3*2/3+5,0) -- (2.5*2/3+5,0.9*2/3);
\draw (3*2/3+5,0) -- (2.1*2/3+5,0.5*2/3);
\draw (1.8*2/3+5,3.2*2/3) --  (3.8*2/3+4.2,4.4*2/3);
\draw (1.8*2/3+5,3.2*2/3)--  (1.8*2/3+4.6,2*2/3);
\draw (1.8*2/3+5,3.2*2/3)--  (1.8*2/3+5,2*2/3);
\draw (2.8*2/3+5,1.9*2/3) -- (3.8*2/3+5,4.4*2/3-0.1);
\draw (2.8*2/3+5,1.9*2/3) -- (3.8*2/3+5,4.4*2/3-0.6);
\draw (2.8*2/3+5,1.9*2/3) -- (3.8*2/3+5,4.4*2/3-1.1);

\draw[dashed] (3.8*2/3+4.2,4.4*2/3) -- (3.8*2/3+5,4.4*2/3-0.1);
\draw[dashed] (0.5*2/3+5,0.9*2/3) --  (1.8*2/3+4.6,2*2/3);
\draw[dashed] (0.9*2/3+5,0.5*2/3)--  (1.8*2/3+5,2*2/3);

\node at (0+5,0) [circle,draw=red!50,fill=red!20] {};
\node at (0+5,0) {2};
\node at (-1*2/3+5,2*2/3) [circle,draw=green!50,fill=green!20] {};
\node at (-1*2/3+5,2*2/3) {3};
\node at (3*2/3+5,0) [circle,draw=green!50,fill=green!20] {};
\node at (3*2/3+5,0) {3};
\node at (4*2/3+5,2*2/3) [circle,draw=yellow!50,fill=yellow!20] {};
\node at (4*2/3+5,2*2/3){4};
\node at (1.5*2/3+5,4.4*2/3) [circle,draw=black!50,fill=black!20] {};
\node at (1.5*2/3+5,4.4*2/3) {5};
\node at (1.8*2/3+5,3.2*2/3) [circle,draw=red!50,fill=red!20] {};
\node at (1.8*2/3+5,3.2*2/3){2};
\node at (0.9*2/3+5,0.5*2/3) [circle,draw=blue!50,fill=blue!20] {};
\node at (0.9*2/3+5,0.5*2/3) {1};
\node at (2.8*2/3+5,1.9*2/3) [circle,draw=red!50,fill=red!20] {};
\node at (2.8*2/3+5,1.9*2/3) {2};
\node at (2.8*2/3+5,2.5*2/3) [circle,draw=blue!50,fill=blue!20] {};
\node at (2.8*2/3+5,2.5*2/3) {1};
\node at (0.2*2/3+5,1.9*2/3) [circle,draw=blue!50,fill=blue!20] {};
\node at (0.2*2/3+5,1.9*2/3) {1};
\node at (2.5*2/3+5,0.9*2/3) [circle,draw=blue!50,fill=blue!20] {};
\node at (2.5*2/3+5,0.9*2/3) {1};
\node at (1.2*2/3+5,3.2*2/3)[circle,draw=blue!50,fill=blue!20] {};
\node at (1.2*2/3+5,3.2*2/3){1};
\node at (3.8*2/3+4.2,4.4*2/3) [circle,draw=blue!50,fill=blue!20] {};
\node at (3.8*2/3+4.2,4.4*2/3){1};
\node at (3.8*2/3+5,4.4*2/3-1.1) [circle,draw=blue!50,fill=blue!20] {};
\node at (3.8*2/3+5,4.4*2/3-1.1){1};
\node at (0.2*2/3+5,2.5*2/3) [ circle,draw=black,fill=black] {};
\node at (0.5*2/3+5,0.9*2/3) [ circle,draw=black,fill=black] {};
\node at (2.2*2/3+5,0.5*2/3) [ circle,draw=black,fill=black] {};
\node at (1.8*2/3+4.6,2*2/3) [ circle,draw=black,fill=black] {};
\node at (1.8*2/3+5,2*2/3) [ circle,draw=black,fill=black] {};
\node at (3.8*2/3+5,4.4*2/3-0.1) [ circle,draw=black,fill=black] {};
\node at (3.8*2/3+5,4.4*2/3-0.6) [ circle,draw=black,fill=black] {};
\node at (-0.3+5,-0.3){$u_{3}$};
\node at (-1*2/3-0.4+5,2*2/3){$u_{2}$};
\node at (3*2/3+0.3+5,-0.3){$u_{4}$};
\node at (4*2/3+0.4+5,2*2/3){$u_{5}$};
\node at (1.5*2/3+0.4+5,4.4*2/3+0.1){$u_{1}$};
\node at (0.7*2/3-0.45+5,0.7*2/3+0.15){$v_{3}$};
\node at (0.7*2/3+0.5+5,0.7*2/3-0.15){$v'_{3}$};
\node at (2.3*2/3+0.45+5,0.7*2/3+0.15){$v_{4}$};
\node at (2.3*2/3-0.3+5,0.7*2/3-0.3){$v'_{4}$};
\node at (0.2*2/3+5+0.1,2.2*2/3+0.5){$v_{2}$};
\node at (0.2*2/3+5,2.2*2/3-0.5){$v'_{2}$};
\node at (2.8*2/3+4.7,2.2*2/3+0.2){$v_{5}$};
\node at (2.8*2/3+5,2.2*2/3-0.5){$v'_{5}$};
\node at (1.5*2/3+5+0.6,3.2*2/3){$v_{1}$};
\node at (0.75*2/3+5,3.2*2/3) {$v'_{1}$};
\node at (4*2/3+0.4+5,4.4*2/3-1.1){$w_{2}$};
\node at (4*2/3+0.2+4.2,4.4*2/3-0.2){$w_{1}$};
\end{tikzpicture}
\end{center}
\caption{Two Grundy partial 5-colorings of a subgraph containing an induced $C_5$.}
\label{figm54}
\end{figure}
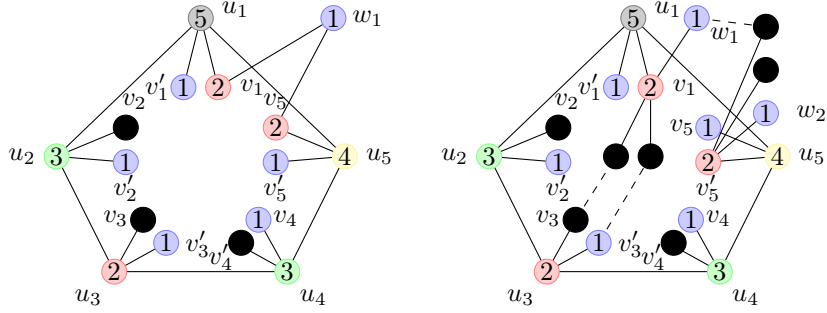
\end{proof}
In the following lemma and proposition, we consider a graph $G$ of girth $g=6$. Let $u_1$, $u_2$, $u_3$, $u_4$, $u_5$ and $u_6$ be the vertices in an induced $C_6$. Let $v_1$, $v'_1$, $v_2$, $v'_2$, $v_3$, $v'_3$, $v_4$, $v'_4$, $v_5$, $v'_5$, $v_6$ and $v'_6$ be the remaining neighbors of respectively $u_1$, $u_2$, $u_3$, $u_4$, $u_5$ and $u_6$ (all different as $g=6$).
\begin{lem}
If $G$ is a 4-regular graph with girth $g=6$ which contains a neighbor-connected $C_6$ as induced subgraph, then $\Gamma(G)=5$.
\label{lnc6}
\end{lem}
\begin{proof}
Firstly, suppose that there are two edges which connect the neighbors in the same way than in the left part of Figure~\ref{figm61}. Let $w_1$ be a neighbor of $v'_1$ not adjacent with $v_4$. The graph $G$ admits a Grundy partial 5-coloring as the left part of Figure~\ref{figm61} illustrates it. 
Secondly, suppose that there is one edge (or more) which connect the neighbors without the configuration from the previous case. Let $w_1$ be a neighbor of $v_3$ not adjacent with $v_2$ and let $w_2$ be a neighbor of $v'_1$ not adjacent with $w_1$. The graph $G$ admits a Grundy partial 5-coloring as the right part of Figure~\ref{figm61} illustrates it. 
\end{proof}
\begin{figure}[t]
\begin{center}
\begin{tikzpicture}[scale=1.15]
\draw (0,0) -- (-1.2,1.2);
\draw (0,0) -- (1.2,1.2);
\draw (1.2,1.2) -- (1.2,2.6);
\draw (-1.2,1.2) -- (-1.2,2.6);
\draw (1.2,1.2) -- (1.2,2.6);
\draw (1.2,2.6) -- (0,4);
\draw (-1.2,2.6) -- (0,4);
\draw (0,0) -- (0.2,0.7);
\draw (0,0) -- (-0.2,0.7);
\draw (0,4) -- (0.2,3.3);
\draw (0,4) -- (-0.2,3.3);
\draw (1.2,1.2) -- (0.5,1.2);
\draw (1.2,1.2) -- (0.5,1.6);
\draw (-1.2,1.2) -- (-0.5,1.2);
\draw (-1.2,1.2) -- (-0.5,1.6);
\draw (1.2,2.6) -- (0.5,2.6);
\draw (1.2,2.6) -- (0.5,2.2);
\draw (-1.2,2.6) -- (-0.5,2.6);
\draw (-1.2,2.6) -- (-0.5,2.2);

\draw[line width=2pt,color=red]  (0.5,1.2) -- (-0.5,2.2);
\draw[line width=2pt,color=red]  (-0.5,1.2) -- (0.5,2.2);
\draw[dashed] (-0.5,1.6)--(0.5,2.6);
\draw[dashed] (0.5,1.6) -- (-0.5,2.6);
\draw[dashed] (-0.2,0.7)--(-0.2,3.3);
\draw[dashed] (0.2,0.7)--(0.2,3.3);

\draw(0.2,3.3) -- (0.7,4);
\draw(0.2,3.3) -- (1.2,4);
\draw[dashed] (-0.5,1.2)--(0.7,4);
\draw[dashed] (0.5,1.2) .. controls (2.4,2.6) .. (1.2,4);

\node at (0,0)[circle,draw=blue!50,fill=blue!20] {};
\node at (0,0) {1};
\node at (-1.2,1.2)[circle,draw=red!50,fill=red!20] {};
\node at (-1.2,1.2) {2};
\node at (-1.2,2.6) [circle,draw=green!50,fill=green!20] {};
\node at (-1.2,2.6) {3};
\node at (1.2,1.2) [circle,draw=green!50,fill=green!20] {};
\node at (1.2,1.2) {3};
\node at (1.2,2.6)[circle,draw=yellow!50,fill=yellow!20] {};
\node at (1.2,2.6) {4};
\node at (0,4) [circle,draw=black!50,fill=black!20] {};
\node at (0,4) {5};

\node at (0.5,1.2)[circle,draw=red!50,fill=red!20] {};
\node at (0.5,1.2) {2};
\node at (0.5,2.2)[circle,draw=red!50,fill=red!20] {};
\node at(0.5,2.2) {2};
\node at (0.5,1.2)[circle,draw=red!50,fill=red!20] {};
\node at (0.5,1.2) {2};
\node at (0.5,2.2)[circle,draw=red!50,fill=red!20] {};
\node at (0.5,2.2) {2};
\node at (0.2,3.3)[circle,draw=red!50,fill=red!20] {};
\node at (0.2,3.3) {2};
\node at (-0.5,2.2)[circle,draw=blue!50,fill=blue!20] {};
\node at (-0.5,2.2) {1};
\node at (-0.5,1.2)[circle,draw=blue!50,fill=blue!20] {};
\node at (-0.5,1.2) {1};
\node at (-0.2,3.3)[circle,draw=blue!50,fill=blue!20] {};
\node at (-0.2,3.3) {1};
\node at (0.5,2.6)[circle,draw=blue!50,fill=blue!20] {};
\node at (0.5,2.6) {1};
\node at (1.2,4) [circle,draw=blue!50,fill=blue!20] {};
\node at (1.2,4) {1};
\node at (-0.5,1.6) [ circle,draw=black,fill=black] {};
\node at (0.5,1.6) [ circle,draw=black,fill=black] {};
\node at (-0.5,2.6) [ circle,draw=black,fill=black] {};
\node at (0.2,0.7) [ circle,draw=black,fill=black] {};
\node at (-0.2,0.7) [ circle,draw=black,fill=black] {};
\node at (0.7,4) [ circle,draw=black,fill=black] {};
\node at (0.4,0) {$u_6$};
\node at (-0.5,0.7) {$v_6$};
\node at (0.45,0.7) {$v'_6$};
\node at (-1.6,1.2) {$u_4$};
\node at (-0.5,0.9) {$v_4$};
\node at (-0.8,1.8)  {$v'_4$};
\node at (-1.6,2.6) {$u_2$};
\node at (-0.8,2.2) {$v_2$};
\node at (-0.8,2.8)  {$v'_2$};

\node at (1.6,1.2) {$u_5$};
\node at (0.5,0.9) {$v_5$};
\node at (0.8,1.8)  {$v'_5$};
\node at (1.6,2.6) {$u_3$};
\node at (0.8,2.2) {$v_3$};
\node at (0.8,2.8)  {$v'_3$};
\node at (0.4,4) {$u_1$};
\node at (-0.5,3.2) {$v_1$};
\node at (0.55,3.2) {$v'_1$};
\node at (1.6,4) {$w_1$};

\draw (0+4.2,0) -- (-1.2+4.2,1.2);
\draw (0+4.2,0) -- (1.2+4.2,1.2);
\draw (1.2+4.2,1.2) -- (1.2+4.2,2.6);
\draw (-1.2+4.2,1.2) -- (-1.2+4.2,2.6);
\draw (1.2+4.2,1.2) -- (1.2+4.2,2.6);
\draw (1.2+4.2,2.6) -- (0+4.2,4);
\draw (-1.2+4.2,2.6) -- (0+4.2,4);

\draw (0+4.2,0) -- (0.2+4.2,0.7);
\draw (0+4.2,0) -- (-0.2+4.2,0.7);
\draw (0+4.2,4) -- (0.2+4.2,3.3);
\draw (0+4.2,4) -- (-0.2+4.2,3.3);
\draw (1.2+4.2,1.2) -- (0.5+4.2,1.2);
\draw (1.2+4.2,1.2) -- (0.5+4.2,1.6);
\draw (-1.2+4.2,1.2) -- (-0.5+4.2,1.2);
\draw (-1.2+4.2,1.2) -- (-0.5+4.2,1.6);
\draw (1.2+4.2,2.6) -- (0.5+4.2,2.6);
\draw (1.2+4.2,2.6) -- (0.5+4.2,2.2);
\draw (-1.2+4.2,2.6) -- (-0.5+4.2,2.6);
\draw (-1.2+4.2,2.6) -- (-0.5+4.2,2.2);

\draw[line width=2pt,color=red]  (0.5+4.2,1.2) -- (-0.5+4.2,2.2);
\draw[dashed] (0.5+4.2,1.6) -- (-0.5+4.2,2.6);

\draw(0.2+4.2,3.3) -- (0.7+4.2,4);
\draw(0.2+4.2,3.3) -- (1.2+4.2,4);
\draw(0.5+4.2,2.2) -- (1.2+4.2,3.2);
\draw(0.5+4.2,2.2) -- (0+4.2,2.7);
\draw[dashed] (-0.5+4.2,2.2) -- (0+4.2,2.7);
\draw[dashed] (1.2+4.2,3.2) -- (1.2+4.2,4);
\draw[dashed] (0.5+4.2,1.2) .. controls (2.4+4.2,3) .. (0.7+4.2,4);

\node at (0+4.2,0)[circle,draw=blue!50,fill=blue!20] {};
\node at (0+4.2,0) {1};
\node at (-1.2+4.2,1.2)[circle,draw=red!50,fill=red!20] {};
\node at (-1.2+4.2,1.2) {2};
\node at (-1.2+4.2,2.6) [circle,draw=green!50,fill=green!20] {};
\node at (-1.2+4.2,2.6) {3};
\node at (1.2+4.2,1.2) [circle,draw=green!50,fill=green!20] {};
\node at (1.2+4.2,1.2) {3};
\node at (1.2+4.2,2.6)[circle,draw=yellow!50,fill=yellow!20] {};
\node at (1.2+4.2,2.6) {4};
\node at (0+4.2,4) [circle,draw=black!50,fill=black!20] {};
\node at (0+4.2,4) {5};
\node at (0.5+4.2,1.2)[circle,draw=red!50,fill=red!20] {};
\node at (0.5+4.2,1.2) {2};
\node at (0.5+4.2,2.2)[circle,draw=red!50,fill=red!20] {};
\node at(0.5+4.2,2.2) {2};
\node at (0.5+4.2,1.2)[circle,draw=red!50,fill=red!20] {};
\node at (0.5+4.2,1.2) {2};
\node at (0.5+4.2,2.2)[circle,draw=red!50,fill=red!20] {};
\node at (0.5+4.2,2.2) {2};
\node at (0.2+4.2,3.3)[circle,draw=red!50,fill=red!20] {};
\node at (0.2+4.2,3.3) {2};
\node at (-0.5+4.2,2.2)[circle,draw=blue!50,fill=blue!20] {};
\node at (-0.5+4.2,2.2) {1};
\node at (-0.2+4.2,3.3)[circle,draw=blue!50,fill=blue!20] {};
\node at (-0.2+4.2,3.3) {1};
\node at (0.5+4.2,2.6)[circle,draw=blue!50,fill=blue!20] {};
\node at (0.5+4.2,2.6) {1};
\node at (0.7+4.2,4) [circle,draw=blue!50,fill=blue!20] {};
\node at (0.7+4.2,4) {1};
\node at (1.2+4.2,3.2) [circle,draw=blue!50,fill=blue!20] {};
\node at (1.2+4.2,3.2) {1};
\node at (-0.5+4.2,1.6) [ circle,draw=black,fill=black] {};
\node at (0.5+4.2,1.6) [ circle,draw=black,fill=black] {};
\node at (-0.5+4.2,2.6) [ circle,draw=black,fill=black] {};
\node at (-0.5+4.2,1.2) [ circle,draw=black,fill=black] {};
\node at (0+4.2,2.7) [ circle,draw=black,fill=black] {};
\node at (1.2+4.2,4) [ circle,draw=black,fill=black] {};
\node at (0.2+4.2,0.7) [ circle,draw=black,fill=black] {};
\node at (-0.2+4.2,0.7) [ circle,draw=black,fill=black] {};
\node at (0.4+4.2,0) {$u_6$};
\node at (-0.5+4.2,0.7) {$v_6$};
\node at (0.45+4.2,0.7) {$v'_6$};
\node at (-1.6+4.2,1.2) {$u_4$};
\node at (-0.5+4.2,0.9) {$v_4$};
\node at (-0.8+4.2,1.8)  {$v'_4$};
\node at (-1.6+4.2,2.6) {$u_2$};
\node at (-0.8+4.2,2.2) {$v_2$};
\node at (-0.8+4.2,2.8)  {$v'_2$};

\node at (1.6+4.2,1.2) {$u_5$};
\node at (0.5+4.2,0.9) {$v_5$};
\node at (0.8+4.2,1.8)  {$v'_5$};
\node at (1.6+4.2,2.6) {$u_3$};
\node at (0.8+4.2,2.2) {$v_3$};
\node at (0.8+4.2,2.8)  {$v'_3$};
\node at (0.4+4.2,4) {$u_1$};
\node at (-0.5+4.2,3.2) {$v_1$};
\node at (0.55+4.2,3.2) {$v'_1$};
\node at (1.6+4.2,3.2) {$w_1$};
\node at (0.9+4.2,4.2) {$w_2$};

\end{tikzpicture}
\end{center}
\caption{Two Grundy partial 5-colorings of a subgraph containing an induced neighbor-connected $C_6$.}
\label{figm61}
\end{figure}
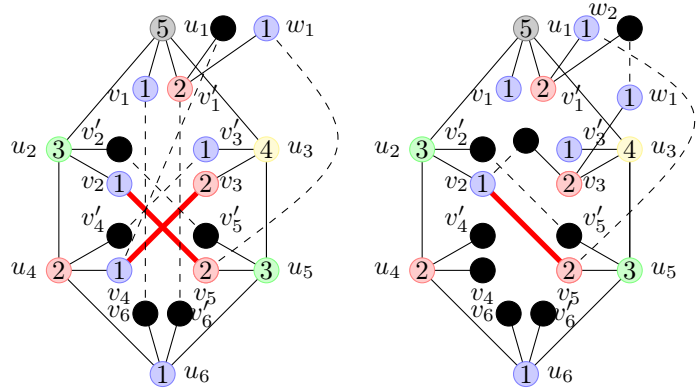
\begin{prop}
If $G$ is a 4-regular graph with girth $g=6$, then $\Gamma(G)=5$.
\label{cycle6}
\end{prop}
\begin{proof}
By Lemma~\ref{lnc6}, assume that no neighbors of the vertices of the induced $C_6$ are adjacent.
Firstly, suppose that there are two neighbors at distance 4 along the cycle $C_6$, for example $v'_1$ and $v_5$, which have a common neighbor $w_1$. Let $w_2$ be a neighbor of $v_3$ not adjacent with $w_1$. $G$ admits a Grundy partial 5-coloring as the left part of Figure~\ref{figm62} illustrates it. 
Secondly, suppose that there are no two neighbors at distance 4 along the cycle $C_6$ which have a common neighbor. Let $w_1$ be a neighbor of $v'_1$ not adjacent with a neighbor of $v_5$ or a neighbor of $v_3$, let $w_2$ be a neighbor of $v_3$ not adjacent with a neighbor of $v_5$, and let $w_3$ be a neighbor of $v_5$. The graph $G$ admits a Grundy partial 5-coloring as the right part of Figure~\ref{figm62} illustrates it. 
\end{proof} 
\begin{figure}[t]
\begin{center}
\begin{tikzpicture}[scale=1.15]
\draw (0,0) -- (-1.2,1.2);
\draw (0,0) -- (1.2,1.2);
\draw (1.2,1.2) -- (1.2,2.6);
\draw (-1.2,1.2) -- (-1.2,2.6);
\draw (1.2,1.2) -- (1.2,2.6);
\draw (1.2,2.6) -- (0,4);
\draw (-1.2,2.6) -- (0,4);

\draw (0,0) -- (0.2,0.7);
\draw (0,0) -- (-0.2,0.7);
\draw (0,4) -- (0.2,3.3);
\draw (0,4) -- (-0.2,3.3);
\draw (1.2,1.2) -- (0.5,1.2);
\draw (1.2,1.2) -- (0.5,1.6);
\draw (-1.2,1.2) -- (-0.5,1.2);
\draw (-1.2,1.2) -- (-0.5,1.6);
\draw (1.2,2.6) -- (0.5,2.6);
\draw (1.2,2.6) -- (0.5,2.2);
\draw (-1.2,2.6) -- (-0.5,2.6);
\draw (-1.2,2.6) -- (-0.5,2.2);

\draw(0.2,3.3) -- (1.2,4);
\draw(0.5,1.2) -- (1.2,4);
\draw(0.5,2.2) -- (1.2,3.35);
\draw(0.5,2.2) -- (1.2,3);
\draw[dashed] (1.2,4) -- (1.2,3.35);

\node at (0,0)[circle,draw=blue!50,fill=blue!20] {};
\node at (0,0) {1};
\node at (-1.2,1.2)[circle,draw=red!50,fill=red!20] {};
\node at (-1.2,1.2) {2};
\node at (-1.2,2.6) [circle,draw=green!50,fill=green!20] {};
\node at (-1.2,2.6) {3};
\node at (1.2,1.2) [circle,draw=green!50,fill=green!20] {};
\node at (1.2,1.2) {3};
\node at (1.2,2.6)[circle,draw=yellow!50,fill=yellow!20] {};
\node at (1.2,2.6) {4};
\node at (0,4) [circle,draw=black!50,fill=black!20] {};
\node at (0,4) {5};

\node at (0.5,1.2)[circle,draw=red!50,fill=red!20] {};
\node at (0.5,1.2) {2};
\node at (0.5,2.2)[circle,draw=red!50,fill=red!20] {};
\node at(0.5,2.2) {2};
\node at (0.5,1.2)[circle,draw=red!50,fill=red!20] {};
\node at (0.5,1.2) {2};
\node at (0.5,2.2)[circle,draw=red!50,fill=red!20] {};
\node at (0.5,2.2) {2};
\node at (0.2,3.3)[circle,draw=red!50,fill=red!20] {};
\node at (0.2,3.3) {2};
\node at (-0.5,2.2)[circle,draw=blue!50,fill=blue!20] {};
\node at (-0.5,2.2) {1};
\node at (-0.2,3.3)[circle,draw=blue!50,fill=blue!20] {};
\node at (-0.2,3.3) {1};
\node at (0.5,2.6)[circle,draw=blue!50,fill=blue!20] {};
\node at (0.5,2.6) {1};
\node at (1.2,4) [circle,draw=blue!50,fill=blue!20] {};
\node at (1.2,4) {1};
\node at (1.2,3) [circle,draw=blue!50,fill=blue!20] {};
\node at (1.2,3) {1};
\node at (-0.5,1.6) [ circle,draw=black,fill=black] {};
\node at (0.5,1.6) [ circle,draw=black,fill=black] {};
\node at (-0.5,1.2) [ circle,draw=black,fill=black] {};
\node at (-0.5,2.6) [ circle,draw=black,fill=black] {};
\node at (0.2,0.7) [ circle,draw=black,fill=black] {};
\node at (-0.2,0.7) [ circle,draw=black,fill=black] {};
\node at (1.2,3.35) [ circle,draw=black,fill=black] {};
\node at (0.4,0) {$u_6$};
\node at (-0.5,0.7) {$v_6$};
\node at (0.45,0.7) {$v'_6$};
\node at (-1.6,1.2) {$u_4$};
\node at (-0.5,0.9) {$v_4$};
\node at (-0.8,1.8)  {$v'_4$};
\node at (-1.6,2.6) {$u_2$};
\node at (-0.8,2.2) {$v_2$};
\node at (-0.8,2.8)  {$v'_2$};

\node at (1.6,1.2) {$u_5$};
\node at (0.8,0.9) {$v_5$};
\node at (0.8,1.8)  {$v'_5$};
\node at (1.6,2.6) {$u_3$};
\node at (0.8,2.2) {$v_3$};
\node at (0.5,2.8)  {$v'_3$};
\node at (0.4,4) {$u_1$};
\node at (-0.5,3.2) {$v_1$};
\node at (0.55,3.2) {$v'_1$};
\node at (1.6,4) {$w_1$};
\node at (1.6,3) {$w_2$};

\draw (0+4.2,0) -- (-1.2+4.2,1.2);
\draw (0+4.2,0) -- (1.2+4.2,1.2);
\draw (1.2+4.2,1.2) -- (1.2+4.2,2.6);
\draw (-1.2+4.2,1.2) -- (-1.2+4.2,2.6);
\draw (1.2+4.2,1.2) -- (1.2+4.2,2.6);
\draw (1.2+4.2,2.6) -- (0+4.2,4);
\draw (-1.2+4.2,2.6) -- (0+4.2,4);

\draw (0+4.2,0) -- (0.2+4.2,0.7);
\draw (0+4.2,0) -- (-0.2+4.2,0.7);
\draw (0+4.2,4) -- (0.2+4.2,3.3);
\draw (0+4.2,4) -- (-0.2+4.2,3.3);
\draw (1.2+4.2,1.2) -- (0.5+4.2,1.2);
\draw (1.2+4.2,1.2) -- (0.5+4.2,1.6);
\draw (-1.2+4.2,1.2) -- (-0.5+4.2,1.2);
\draw (-1.2+4.2,1.2) -- (-0.5+4.2,1.6);
\draw (1.2+4.2,2.6) -- (0.5+4.2,2.6);
\draw (1.2+4.2,2.6) -- (0.5+4.2,2.2);
\draw (-1.2+4.2,2.6) -- (-0.5+4.2,2.6);
\draw (-1.2+4.2,2.6) -- (-0.5+4.2,2.2);

\draw(0.2+4.2,3.3) -- (-0.4+4.2,4);
\draw(0.2+4.2,3.3) -- (-0.8+4.2,4);
\draw(0.2+4.2,3.3) -- (-1.2+4.2,4);
\draw(0.5+4.2,2.2) -- (1.2+4.2,3.6);
\draw(0.5+4.2,2.2) -- (1.2+4.2,3.2);
\draw(0.5+4.2,1.2) -- (1.2+4.2,0.6);
\draw(0.5+4.2,1.2) -- (1.2+4.2,0.2);
\draw[dashed](-0.4+4.2,4) -- (1.2+4.2,3.2);
\draw[dashed](-0.8+4.2,4) -- (1.2+4.2,0.2);
\draw[dashed](1.2+4.2,3.6) .. controls (1.2+5.2,2.1)..  (1.2+4.2,0.6);

\node at (0+4.2,0)[circle,draw=blue!50,fill=blue!20] {};
\node at (0+4.2,0) {1};
\node at (-1.2+4.2,1.2)[circle,draw=red!50,fill=red!20] {};
\node at (-1.2+4.2,1.2) {2};
\node at (-1.2+4.2,2.6) [circle,draw=green!50,fill=green!20] {};
\node at (-1.2+4.2,2.6) {3};
\node at (1.2+4.2,1.2) [circle,draw=green!50,fill=green!20] {};
\node at (1.2+4.2,1.2) {3};
\node at (1.2+4.2,2.6)[circle,draw=yellow!50,fill=yellow!20] {};
\node at (1.2+4.2,2.6) {4};
\node at (0+4.2,4) [circle,draw=black!50,fill=black!20] {};
\node at (0+4.2,4) {5};

\node at (0.5+4.2,1.2)[circle,draw=red!50,fill=red!20] {};
\node at (0.5+4.2,1.2) {2};
\node at (0.5+4.2,2.2)[circle,draw=red!50,fill=red!20] {};
\node at(0.5+4.2,2.2) {2};
\node at (0.5+4.2,1.2)[circle,draw=red!50,fill=red!20] {};
\node at (0.5+4.2,1.2) {2};
\node at (0.5+4.2,2.2)[circle,draw=red!50,fill=red!20] {};
\node at (0.5+4.2,2.2) {2};
\node at (0.2+4.2,3.3)[circle,draw=red!50,fill=red!20] {};
\node at (0.2+4.2,3.3) {2};
\node at (-0.5+4.2,2.2)[circle,draw=blue!50,fill=blue!20] {};
\node at (-0.5+4.2,2.2) {1};
\node at (-0.2+4.2,3.3)[circle,draw=blue!50,fill=blue!20] {};
\node at (-0.2+4.2,3.3) {1};
\node at (0.5+4.2,2.6)[circle,draw=blue!50,fill=blue!20] {};
\node at (0.5+4.2,2.6) {1};
\node at (-1.2+4.2,4) [circle,draw=blue!50,fill=blue!20] {};
\node at (-1.2+4.2,4) {1};
\node at (1.2+4.2,3.6) [circle,draw=blue!50,fill=blue!20] {};
\node at (1.2+4.2,3.6) {1};
\node at (1.2+4.2,0.2) [circle,draw=blue!50,fill=blue!20] {};
\node at (1.2+4.2,0.2){1};
\node at (-0.5+4.2,1.6) [ circle,draw=black,fill=black] {};
\node at (0.5+4.2,1.6) [ circle,draw=black,fill=black] {};
\node at (-0.5+4.2,1.2) [ circle,draw=black,fill=black] {};
\node at (-0.5+4.2,2.6) [ circle,draw=black,fill=black] {};
\node at (0.2+4.2,0.7) [ circle,draw=black,fill=black] {};
\node at (-0.2+4.2,0.7) [ circle,draw=black,fill=black] {};
\node at (-0.4+4.2,4)[ circle,draw=black,fill=black] {};
\node at (-0.8+4.2,4) [ circle,draw=black,fill=black] {};
\node at (1.2+4.2,3.2) [ circle,draw=black,fill=black] {};
\node at (1.2+4.2,0.6) [ circle,draw=black,fill=black] {};
\node at (0.4+4.2,0) {$u_6$};
\node at (-0.5+4.2,0.7) {$v_6$};
\node at (0.45+4.2,0.7) {$v'_6$};
\node at (-1.6+4.2,1.2) {$u_4$};
\node at (-0.8+4.2,0.9) {$v_4$};
\node at (-0.8+4.2,1.8)  {$v'_4$};
\node at (-1.6+4.2,2.6) {$u_2$};
\node at (-0.8+4.2,2.2) {$v_2$};
\node at (-0.8+4.2,2.8)  {$v'_2$};

\node at (1.6+4.2,1.2) {$u_5$};
\node at (0.5+4.2,0.9) {$v_5$};
\node at (0.8+4.2,1.8)  {$v'_5$};
\node at (1.6+4.2,2.6) {$u_3$};
\node at (0.8+4.2,2.2) {$v_3$};
\node at (0.5+4.2,2.8)  {$v'_3$};
\node at (0.4+4.2,4) {$u_1$};
\node at (-0.5+4.2,3.2) {$v_1$};
\node at (0.55+4.2,3.2) {$v'_1$};
\node at (1.6+4.2,3.6) {$w_2$};
\node at (-1.6+4.2,4) {$w_1$};
\node at (1.6+4.2,0.2) {$w_3$};

\end{tikzpicture}
\end{center}
\caption{Two Grundy partial 5-colorings of a subgraph containing an induced $C_6$.}
\label{figm62}
\end{figure}
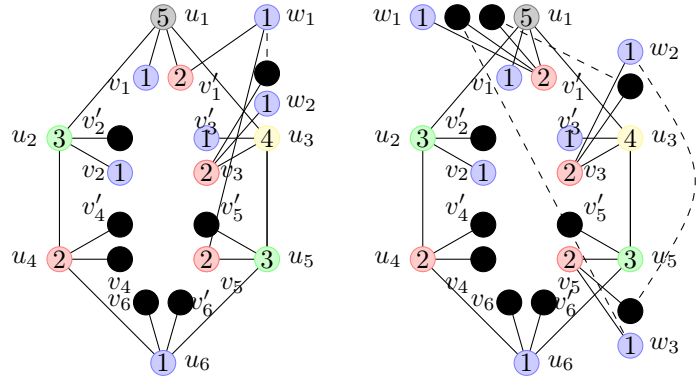
\begin{prop}
If $G$ is a 4-regular graph with girth $g\ge 7$, then $\Gamma(G)=5$. 
\label{propg}
\end{prop}
\begin{proof}
Suppose that $G$ contains a 7-cycle. We denote the 5-atom which is a tree by $T_5$ (the binomial tree with maximum degree 4). It can be easily verified that $G$ contains $T_5$ where two leaves are merged (which is a 5-atom). Moreover, if $G$ does not contain a 7-cycle, then it contains $T_5$ as induced subgraph.
\end{proof}
\begin{theo}
Let $G$ be a 4-regular graph. If $G$ does not contain an induced $C_4$, then $\Gamma(G)=5$. 
\label{indc4}
\end{theo}
\begin{proof}
Suppose that $G$ does not contain an induced $C_4$.
Using Proposition~\ref{propg} for the case $g\ge 7$, Propositions~\ref{cycle5} and~\ref{cycle6} for the case $g=5,6$, and Proposition~\ref{cycle3} when $G$ contains a $C_3$ yields the desired result.
\end{proof}
By Proposition \ref{indc2}, Corollary \ref{indc3} and Theorem \ref{indc4}, any $r$-regular graph with $r\le4$  and without induced $C_4$ has Grundy number $r+1$. Therefore, it is natural to propose Conjecture 1.

\bibliographystyle{plain}
\bibliography{bib}
\end{document}